\documentclass[12pt]{article}

\RequirePackage[OT1]{fontenc}
\usepackage{amsmath,amssymb,amsthm}
\usepackage{natbib}
  \bibpunct{(}{)}{;}{a}{}{,} % asa
\usepackage{graphicx}
\usepackage{bm}
\usepackage[vlined]{algorithm2e}
\usepackage[pdftex,colorlinks=true,linkcolor=blue,citecolor=blue,urlcolor=blue,bookmarks=false,pdfpagemode=None]{hyperref}
\usepackage{verbatim}
\usepackage[top=1in, bottom=1in, left=1in, right=1in]{geometry}
\usepackage{array}
\usepackage{rotating}
\usepackage{float}
\usepackage{mathtools}
\RequirePackage{subcaption}
\RequirePackage{graphicx}
\RequirePackage{xspace}
\RequirePackage{tabularx}
\RequirePackage{comment}
\RequirePackage{subdepth}
\usepackage{multirow}
\usepackage{fancyhdr}
\newcommand{\bv}{\begin{array}}
\newcommand{\E}{\mathbb{E}}

\newtheorem{example}{Example}

\newtheorem{remark}{Remark}
\newtheorem{theorem}{Theorem}
\newtheorem{lemma}{Lemma}
\newtheorem{proposition}{Proposition}

\newtheorem{definition}{Definition}

\newtheorem{corollary}{Corollary}

\numberwithin{theorem}{section}
\numberwithin{lemma}{section}
\numberwithin{proposition}{section}

\def\revision{}

\ifdefined\revision
  \newenvironment{revtwo}{\color{black}}{}
  
\else

\fi

\newcommand{\pr}[1]{\mathbb{P}\left( #1 \right)}

\newcommand{\Perp}{\perp \! \! \! \perp}

\usepackage{tikz}
\usetikzlibrary{matrix,fit,backgrounds}
\usetikzlibrary{decorations.pathreplacing}

\begin{document}
\graphicspath{{./figures/}}

\title{A systematic investigation of classical causal inference strategies under mis-specification due to network interference}

\author{Vishesh Karwa \and Edoardo M. Airoldi}
{\makeatletter
\renewcommand*{\@makefnmark}{}
\footnotetext{Vishesh Karwa is an Assistant Professor in the Department of Statistical Science, Fox Business School at Temple University. %()
Edoardo M.~Airoldi is the Millard E. Gladfelter Professor of Statistics \& Data Science and the Director of the Data Science Center at the Fox Business School, Temple University. 
This work was partially supported 
 by the National Science Foundation under grants 
  CAREER IIS-1149662 and IIS-1409177,
 by the Office of Naval Research under grants 
  YIP N00014-14-1-0485 and N00014-17-1-2131, to Harvard University,
 and 
 by a Shutzer Fellowship and a Sloan Research Fellowship to EMA.}
\makeatother}

\date{}

\maketitle
\thispagestyle{empty}
\begin{abstract}

	We systematically investigate issues due to mis-specification that arise in estimating causal effects when (treatment) interference is informed by a network available pre-intervention, i.e., in situations where the outcome of a unit may depend on the treatment assigned to other units. We develop theory for several forms of interference through the concept of “exposure neighborhood”, and develop the corresponding semi-parametric representation for potential outcomes as a function of the exposure neighborhood. Using this representation, we extend the definition of two popular classes of causal estimands considered in the literature, \emph{marginal} and \emph{average} causal effects, to the case of network interference. We then turn to characterizing the bias and variance one incurs when combining classical randomization strategies (namely, Bernoulli, Completely Randomized, and Cluster Randomized designs) and estimators (namely, difference-in-means and Horvitz-Thompson) used to estimate average treatment effect and on the total treatment effect, under misspecification due to interference. We illustrate how difference-in-means estimators can have arbitrarily large bias when estimating average causal effects, depending on the form and strength of interference, which is unknown at design stage. Horvitz-Thompson estimators are unbiased when the correct weights are specified. Here, we derive the Horvitz-Thompson weights for unbiased estimation of different estimands, and illustrate how they depend on the design, the form of interference, which is unknown at design stage, and the estimand. More importantly, we show that Horvitz-Thompson estimators are in-admissible for a large class of randomization strategies, in the presence of interference. We develop new model-assisted and model-dependent strategies to improve Horvitz-Thompson estimators, and we develop new randomization strategies for estimating the average treatment effect and total treatment effect.

\end{abstract}

\newpage
\tableofcontents

 %%% %%% %%%
 %%% %%% %%%
 %%% %%% %%% 

\newpage

\pagestyle{fancy}
\setcounter{page}{1}
\section{Introduction} 
\label{sec:Intro}

The estimation of causal effects is a fundamental goal of many scientific studies. The framework of Potential Outcomes \citep{splawa1990application, rubin1974estimating, holland1986statistics} is a popular approach to formalize the problem of estimating causal effects of a treatment on an outcome, from a finite population of $n$ units. For instance, one can use the potential outcomes framework to formally define causal effects of interest called \emph{estimands} or \emph{inferential targets}, construct estimators that have desirable properties, such as, unbiasedness with respect to the randomization distribution, and formulate the assumptions under which the estimands and the estimators are well defined and causal conclusions hold.

An important assumption made in the classical potential outcomes framework is the \emph{no treatment-interference} (or simply no interference\footnote{There can be other forms of interference; for e.g. the outcome of a unit may depend on the outcome of others. We are concerned only with treatment interference}) assumption which can be stated as follows: The outcome of any unit depends only on its own treatment. In particular, the outcome of a unit does not depend on the treatment assigned to (or selected by) other units in the finite population of $n$ units. This assumption is implied by the so called Stable Unit Treatment Value Assumption or SUTVA as formulated in \cite{rubin1980randomization}, see also \cite{rubin1986comment}.
 %Several issues arise when one attempts to relax the assumption of no treatment interference. 
 It is clear and well known (see for e.g. Section 3 in \cite{rubin1990application}) that the classical framework of potential outcomes needs to be extended when estimating causal effects under interference. 
 
 When extending the classical potential outcomes framework and relaxing the assumption of no treatment-interference, the key natural question that arises is the following: What should be the form of interference? It is straightforward to specify what we mean by no treatment-interference, but the existence of treatment interference is not a concrete modeling assumption - there are many different ways to specify the exact form of interference and one needs to choose from various models of interference. %and one needs specify the exact form of interference. 
 
 Once a model for interference is fixed, the next steps are to define causal estimands and develop designs and corresponding estimators. The classical versions of average treatment effects are no longer well defined when there is interference between units. This is due to the fact that the space of potential outcomes for each unit changes with the form of interference. In particular, the number of potential outcomes for each unit becomes a function of the form of interference. For example, consider a binary treatment $T$. Under the no treatment-interference assumption, each unit $i$ has two potential outcomes  $Y_i(0)$ and $Y_i(1)$. The average causal effect is defined as the average of differences between these two potential outcomes. However, when there is arbitrary treatment-interference, the number of potential outcomes for each unit can be as large as $2^n$ and $Y_i(0)$ and $Y_i(1)$ are not well defined. There are many non-equivalent ways to define an estimand under interference and the choice depends on the scientific question that one is interested in answering. Once a choice has been made regarding the nature of interference and an estimand has been proposed, the next step is to develop (idealized) experimental designs along with corresponding estimators with good properties, such as unbiasedness with respect to the design, that allow us to estimate causal estimands.  

 \paragraph{Related work} Relaxing the assumption of no treatment-interference has been the subject of many works, see \cite{halloran2016dependent} for a recent review. A classical line of work proceeds by limiting the interference to non-overlapping groups and assuming that there is no interference between groups. This setting is often referred to as \emph{partial interference}  \citep[e.g., see][]{sobel2006randomized, hudgens2012toward, tchetgen2012causal, liu2014large,kang2016peer,liu2016inverse,rigdon2015exact,basse2017analyzing,forastiere2016identification, loh2018randomization}. Various types of estimands have been defined under  partial interference. For instance, \cite{sobel2006randomized} defined estimands that naturally arise in housing mobility studies and noted that under partial interference the classical estimators may be biased.  \cite{hudgens2012toward} considered potential outcomes marginalized over a randomization distribution, and use these marginal potential outcomes to define estimands. They considered two-stage designs and developed unbiased estimators for these marginal estimands. A different line of work has focused on designing  experiments that eliminate or reduce partial interference, so that estimation can be carried out by ignoring interference \citep[e.g., see][]{david1996designs}. In the modern setting, the assumption of partial interference has been relaxed by several authors to allow for arbitrary interference, or interference encoded by a network, see \cite{bowers2012reasoning,manski2013identification,goldsmith2013social,toulis2013estimation,ugander2013graph,aronow2013estimating,basse2015optimal,forastiere2016identification,halloran2016dependent, choi2017estimation,athey2017exact}. \cite{manski2013identification} considered the problem of whether causal effects are identifiable in presence of arbitrary interference. \cite{aronow2013estimating} proposed Horvitz-Thompson estimators for estimating causal effects when there is arbitrary interference. \cite{ugander2013graph} and \cite{eckles2014design} consider a cluster randomization design to reduce bias in estimating  a specific estimand (i.e., total treatment effect). More recently, \cite{savje2017average} study the large sample properties of estimating treatment effects, when the interference structure is unknown. They show (somewhat surprisingly) that in a large sample setup, the Horvitz-Thompson and Hajek estimators can be used to consistently estimate the \emph{expected average treatment effect}, even if the structure of interference is incorrect. \cite{jagadeesan2017designs} have proposed new designs for estimating the direct effect under interference. \cite{ogburn2014causal} approach the problem of interference by using causal diagrams, and they present various causal Bayesian networks under different types of interference. \cite{ogburn2017vaccines} use causal diagrams to develop GLM type estimators for contagion. Finally, \cite{li2018randomization} study peer effects using randomization based inference.

In this paper, we initiate a systematic investigation of issues that arise in definition and unbiased estimation of causal effects under arbitrary interference and develop possible solutions. Some of the key goals of our work are (a) to develop models of interference, (b) organize and place different estimands and estimators that have appeared in the literature under a common framework, (c) to clarify the issues present in existing definitions and estimators of causal effects and (d) study designs and unbiased estimation strategies under interference.

\subsection{Summary of Contributions and Organization}
Section \ref{sec:mainresults} provides an overview of the key results of the paper. Here we present a informal summary of contributions.  

\emph{Models for Potential Outcomes under Interference:} We begin by revisiting the framework of potential outcomes under arbitrary interference in Section \ref{sec:POandNetwork}. Using the concept of \emph{exposure neighborhood}, in Section \ref{sec:POmodel}, we develop non-parametric models potential outcomes to formalize the nature and form of interference. The exposure neighborhood allows one to explicitly model the form of interference, whereas the \emph{structural models} formulate assumptions on the structure of potential outcomes under the assumed form of interference. %Moreover, the linear non-parametric representation gives a convenient interpretable characterization of causal effects of interest.

\emph{Choice of estimands:} Unlike the classical no-interference setting, there are several non-equivalent ways of defining causal estimands under interference. In Section \ref{sec:estimands}, we consider two different (overlapping) classes of estimands for formally defining causal effects - \emph{marginal causal effects} (in the spirit of \cite{hudgens2012toward}) and \emph{average causal effects}. Marginal causal effects are defined as contrasts between expected values of potential outcomes under a fixed randomization scheme also called as \emph{policy}, where as the average causal effects are defined as contrasts between averages of fixed potential outcomes. %It is important to note that the policy and the definition of a causal effect are independent of the design used to estimate it. 
These classes include several estimands that have appeared in the literature as special cases. %The linear non-parametric representation of Potential Outcomes allows for easy interpretation of the estimands as direct and indirect effects. We study the relationship of these causal effects to classical estimands. Some of the estimands we consider reduce to their classical versions under no-interference. For example, there are at least two different average treatment effects under interference, which are non-equivalent in general. However, under the no-interference assumption, they reduce to the classical definition of average treatment effect. 

\emph{Bias due to interference in difference-in-means estimators:} In Section \ref{sec:Bias}, we address some folklore about estimation strategies. In many cases, it is common to use a design along with classical difference-in-means like estimators to estimate an average causal effect, even when there is interference, with the hope that there might be little or no bias. In some settings, however, the definition of causal effect that is being estimated (the estimand) is not well specified. Our analysis makes it clear that certain classic versions of causal effects are not well defined when there is interference. In the cases where the estimand is well defined, we show that difference-in-means estimators can be biased for many types of estimands. We characterize the nature and sources of bias in estimating a large class of estimands. Our results also illustrate settings where simple estimators can yield  little or no bias. For instance, when estimating the so called \emph{marginal causal effects}, the difference-in-means estimators are unbiased. In general, the unbiasedness of the difference-of-means estimators depends on the nature and structure of interference, which we characterize in Section \ref{sec:biasPOModels}.

\emph{Liner Unbiased Estimation:} We then consider the problem of unbiased estimation of causal effects with commonly used estimation \emph{strategies}, in Section \ref{sec:lue}.  We consider the Bernoulli, Completely Randomized and Cluster Randomized Designs and focus on the problem of unbiased estimation.  

A popular estimation strategy is to use Horvitz-Thompson (HT) like estimators, which is the subject of Section \ref{sec:lue}. For instance, \cite{aronow2013estimating} proposed using HT estimators for particular estimands (i.e., contrasts between potential outcomes corresponding to two different treatment assignment vectors).
% EDO : NEED TO BE PRECISE
We consider the class of all linear weighted unbiased estimators and show that HT estimators can be used to obtain unbiased estimates for any estimand and design, as long as the {\it correct} weights are used and some regularity conditions on the design hold. However, we note that the weights depend on the design, the structure of interference (as specified by an interference model) and the estimand. We explicitly derive the weights of HT estimators for commonly used designs and estimands. 
%In this sense, HT estimators are also a class of estimators and we obtain a different estimator for different designs and estimand. 
% EDO : UNCLEAR, POSSIBLY REDUNDANT 
We also show that the correct weights that endow HT estimators with good properties need not be unique. The question of optimality (e.g., minimum variance, unbiased) of HT weights is difficult, and has been recently addressed, in part \citep{sussman2017elements}.

We prove that Horvitz-Thompson estimators are inadmissible for estimating a large class of estimands and a large class of designs. The HT estimator is one of many estimators in the class of linear weighted unbiased estimators. Using ideas from survey sampling literature, we consider two strategies to improve upon the HT estimator. The non-parametric linear representation of potential outcomes we develop lends itself naturally to develop improved estimators either in a model dependent or a model assisted framework \citep[\`a la][]{basse2015optimal}. Finally, in section \ref{sec:newdesigns}, we explore new designs to estimate two commonly used estimands: average treatment effect and total treatment effect, defined in Section \ref{sec:mainresults}.  A key observation is that the optimal design for estimation may depend on the estimand.

 %%% %%% %%%
 %%% %%% %%%
 %%% %%% %%%

\section{Overview of the main results: Modes of failures and solutions}
\label{sec:mainresults}

Consider a finite population of $n$ units indexed by $\{1, \ldots, n\}$. Let 
$\textbf{z} = (z_1, \ldots, z_n)$ denote a vector of binary treatment assignments where each $z_i \in \{0,1\}$. Let $(Y_i(\textbf{z}))_{i=1}^n$ denote the vector of potential outcomes when the finite population of $n$ units gets assigned the treatment vector $\textbf{z}$. For each unit $i$, $Y_i(\textbf{z})$ is a function of $\textbf{z}$. The no treatment-interference assumption ensures that for each $i = 1, \ldots, n$, we can write the potential outcomes function as,
\begin{align}
\label{eq:intro:no-interference}
Y_i(\textbf{z}) = Y_i(z_i).
\end{align}
Thus, under the no treatment-interference assumption the total number of potential outcomes for each unit $i$ is $2$. However, when there is interference, we can write
\begin{align}
Y_i(\textbf{z}) = Y_i(z_i, \textbf{z}_{-i}).
\end{align}
where $\textbf{z}_{-i}$ is the vector of treatment assignments of all units except $i$.

\paragraph{Explosion of Potential Outcomes} When there is arbitrary interference, the number of potential outcomes for each unit may explode, rendering causal inference impossible without modeling potential outcomes. In general, the total number of potential outcomes for each unit $i$ can be as high as $2^n$.
\begin{proposition}
	\label{intro:impossible}
	Without any further assumptions on the function $Y_i(\textbf{z})$, causal inference is impossible.
\end{proposition}

Proposition \ref{intro:impossible} is simple, but has far reaching consequences. The key consequence is that under arbitrary treatment interference, one must model the potential outcomes, even under randomization inference. Indeed, the no-interference assumption is also a modeling assumption. Thus, the question becomes which model to use. We develop models for potential outcomes by specifying three components: An \emph{interference neighborhood}, an \emph{exposure function} and \emph{structural} assumptions. Modeling potential outcomes allows one to reduce the number of potential outcomes per unit to a more manageable size. The total number of potential outcomes per unit is directly related to the interference neighborhood and an exposure model. 

Table \ref{tab:countingPO} gives three examples of exposure models and the number of potential outcomes per unit. We refer the reader to see Section \ref{sec:POmodel} for precise definitions of these exposure models.
Under the simplest exposure model, called the \emph{binary exposure}, each unit has $4$ potential outcomes - this is twice as many when compared to the case of no-interference. On the other hand, for \emph{symmetric exposure}, the number of potential outcomes for each unit grows linearly with the size of the exposure neighborhood. Finally, for a general exposure model, the number of potential outcomes for each unit grows exponentially with the size of the exposure neighborhood. 

\begin{table}[h]
	\centering
	\begin{tabular}{|c|c|c|}
		\hline  
		Binary Exposure 	& Symmetric Exposure 	& General Exposure\\
		\hline
		$4$ 		& $ 2 \cdot d_i$ 				& $ 2^{d_i+1}$  \\ 
		\hline 
	\end{tabular} 
	\caption{Number of potential outcomes per unit for different exposure models. Here $d_i$ is the size of the exposure neighborhood, i.e. the number of units whose treatment status effect the outcome of unit $i$.}
	\label{tab:countingPO}
\end{table}

\paragraph{Non-parametric Decomposition of Potential Outcomes} Under arbitrary interference, we develop a non-parametric linear decomposition of the potential outcomes:

\begin{proposition}
	\label{intro:decomposition}
	Let $\textbf{z}_{-i}$ denote the vector of treatment assignments assigned to all but unit $i$.  There exist functions $A_i(\cdot), B_i(\cdot), C_i(\cdot)$ and $f$ where if $e_i = f(\textbf{z}_{-i})$, then every potential outcome function for unit $i$ can be decomposed as 
	$$Y_i(\textbf{z}) = A_i(z_i) + B_i(e_i) + z_iC_i(e_i).$$
\end{proposition}
%{\color{red} The following may be a cleaner way to write this result for the intro:}
%\begin{proposition}
%	\label{intro:decomposition2}
%	Let $\textbf{z}_{-i}$ denote the vector of treatment assignments assigned to all but unit $i$.  There exist functions $A_i(\cdot), B_i(\cdot), C_i(\cdot)$ such that every potential outcome function for unit $i$ can be decomposed as 
%	$$Y_i(\textbf{z}) = A_i(z_i) + B_i(e_i) + z_iC_i(\textbf{z}_{-i}).$$
%\end{proposition}

Proposition \ref{intro:decomposition} states that the potential outcome function for every unit $i$ can be decomposed linearly into three components: A component that depends on unit $i's$ treatment, a component that depends on the treatment of all other units, and an interaction term. 
At first glance, this representation appears to be redundant as it is over-parametrized. But this decomposition offers three benefits: Firstly, the number of parameters and hence the number of potential outcomes can be now modeled by specification of these functions. Indeed, the explicit construction of the functions $A_i,B_i, C_i$ and $f$ in Proposition \ref{intro:decomposition} requires modeling assumptions on the Potential outcomes which is the subject of Section \ref{sec:PODef}. 
Secondly, the decomposition makes it clear that classical causal effects are ill-defined when there is interference because they ignore two components of the potential outcomes and use only the first component of direct effect. This decomposition allows us to define different types of causal estimands that focus on direct effects, interference effects or the interaction between the two. Finally, the decomposition also allows us to gain deeper insights into the nature and sources of biases for various classical estimators to estimate causal effects. We will discuss these issues next.

	 \paragraph{Different types of Average Treatment effects} We show that in presence of interference, there are many non-equivalent ways to define a treatment effect. In this summary, we will focus on two most popular treatment effects that fall under the class of \emph{average treatment effects}. We also consider a different class called \emph{marginal treatment effects}, see Section \ref{sec:estimands}. The two average treatment effects that we consider are the \emph{direct effect}
	 \begin{align}
	 \label{eq:intro:de}
	 \beta_1 = \frac{1}{n} \sum_i \left( Y_i(1,\textbf{0}) - Y_i(0,\textbf{0}) \right),
	 \end{align}
	 and the \emph{total effect}
	 \begin{align}
	 \label{eq:intro:te}
	 \beta_2 = \frac{1}{n} \sum_i \left(Y_i(1,\textbf{1}) - Y_i(0,\textbf{0}) \right).
	 \end{align}
	
	 \begin{proposition}
	 	Under the no-interference assumption, $\beta_1 = \beta_2$. Under interference, $\beta_1 \neq \beta_2$.
	 \end{proposition}
	 In fact, one can show that under no-interference assumption, the marginal and the average causal effects are equivalent. This is no longer true in presence of interference, so one needs to be careful in defining what one is interested in estimating. An important point to note is that the causal estimands should not depend on the randomization design and must be defined independent of the actual design that was implemented. %{\color{red}Why? Can we add a line as to why this may be an issue?}
	 
	 In Section \ref{sec:designs}, we introduce the concept of an estimation \emph{strategy} - a combination of a design and an estimator for estimating a particular estimand and in Section \ref{sec:existenceofEstimators} we study conditions under which unbiased estimators exist.

	  \paragraph{Commonly used designs and estimators are biased} An intuitive approach to estimate average causal effects under interference in the literature is to use a difference-in-means estimator, with the hope that a mild form of interference may not effect the bias of the estimator %{\color{red}CITATIONS??}. 
	  We formalize this intuition and study the nature and sources of bias in various difference-in-means estimators under interference. Unfortunately, the situation is more complex. The nature of bias depends on the estimand, the exact form of the difference-in-means estimator, the design and finally the model for interference. This is the subject of Section \ref{sec:Bias}. For estimating the marginal effects, the difference-in-means estimators are unbiased under certain mild conditions on the design. For estimating the total and the direct effect defined in equations \ref{eq:intro:te} and \ref{eq:intro:de}, the situation is more nuanced.
	  
	  In general, there are two sources of bias in estimating the direct and the total effects, see Proposition \ref{prop:biasGeneralCase}. The first source of bias is due to the so called \emph{nuisance potential outcomes}. These are the potential outcomes that do not appear in the definition of the estimand and are irrelevant for estimation of certain classes of estimands, specially the average causal effects. The nuisance potential outcomes form a source of bias when estimating average causal effects, as shown in Section \ref{sec:Bias} and Proposition \ref{prop:biasGeneralCase}. 
	  
	  The second source of bias is due to incorrect weights used in the estimator. In some difference-in-means estimators and designs, the first source of bias can be completely eliminated, see Proposition \ref{prop:expDoMs} for an example. The second source of bias is due to the use of incorrect weights; these weights depend on the design and the nature of interference. For many commonly used designs such as the Completely Randomized Design, Bernoulli Design and the Cluster Randomized Design, assuming a mild form of interference, the second source of bias can be made very small. However, we must point out that the reduction of bias depends on the model of interference, which is not known in general. An incorrect assumption on the interference model may lead to bias, we do not investigate this source of bias. 
	  
	  \begin{proposition}
	  	\label{prop:intro:bias}
	  	Assume that the interference is specified by a graph $G_n$ on $n$ units, i.e., the treatment of unit $i$ effects the outcome of unit $j$ iff there is an edge between nodes $i$ and $j$ in $G_n$. Further, assume that the Potential Outcomes follow a linear model:
	  	$Y_i = \alpha_i + \beta_i z_i + \gamma e_i$, where $e_i$ denotes the number of treated neighbors of unit $i$ in graph $G_n$. Let $m$ be the total number of edges in $G_n$.
	  	Under a completely randomized design and a Bernoulli trial, (defined in Section \ref{sec:designs}), consider the naive difference-in-means estimator:
	  	$$
	  	\hat \beta_{naive} = \frac{\sum_i Y_i^{obs}Z_i}{\sum_i Z_i} - \frac{\sum_i Y_i^{obs}(1-Z_i)}{\sum_i(1-Z_i)}.
	  	$$
	  	The bias of the difference-in-means estimator for estimating the direct effect $\beta_1$ given in equation \ref{eq:intro:de} is
	  	$$\mathbb E[\hat \beta_{naive}] - \beta_1 = -\gamma\frac{m}{{n \choose 2}}.$$
	  \end{proposition}
	  Proposition \ref{prop:intro:bias} is an example of the type of characterization of the nature and source of bias developed in Section \ref{sec:Bias} for various models of interference. This result shows that even under a simple linear model of potential outcomes, the difference-in-means estimator is biased for estimating the direct effect. The bias depends on the unknown interference parameter $\gamma$ and the density of the interference graph given by $m/{n \choose 2}$. The bias is in the opposite direction of the interference effect: If there is positive interference, the estimated direct effect is smaller than the true direct effect and vice versa. Also, if the interference effect is small and the interference graph is sparse, then the bias is very small. However, as we can see, even in such a simple model, the nature of bias depends on unknown parameters such as the density of the interference graph $G_n$ and $\gamma$.  For more general models, the qualitative results are similar, and the reader is refereed to Section \ref{sec:Bias} for more details.

	 \paragraph{Linear unbiased estimators and inadmissibility of the Horvitz-Thompson estimator} Section \ref{sec:lue} is devoted to the theory of linear unbiased estimation. For any design, weighted unbiased linear estimators can be constructed using techniques from sampling theory. We study two classes of weighted linear unbiased estimators. We show that under some regularity conditions, there are infinitely many weighted linear unbiased estimators, see Theorem \ref{thm:generalLinear}. Moreover, when the weights are allowed to depend only on the treatment and exposure status of a unit, the Horvitz-Thompson estimator is the only unbiased estimator, see Theorem \ref{thm:HT}. The weights used in the HT estimator depend on the interference model, the design and the estimand. In Theorem \ref{thm:propensityscores}, we derive the formula for the weights used in HT estimators for Bernoulli, CRD and the Cluster Randomized designs for different interference models for estimating the direct effect. A point to note is that unbiased estimators of the direct effect do not exist when using cluster randomized designs. This illustrates the fact that an estimation \emph{strategy} that is considered optimal for one type of estimand may not necessarily be optimal for a different estimand, in fact, it can be far from optimal. The optimality criteria can be as simple as existence or unbiasedness.\footnote{It remains an open question to find estimation strategies that can be simultaneously optimal for a large class of estimands.}

	 Although the H-T estimator is unbiased, its performance can be very poor in practice because of high variance. An estimator is inadmissible if there exists a uniformly better estimator in terms of the mean squared error. In Theorem \ref{thm:HTinadmissible}, we show that for a large class of designs that satisfy some natural regularity conditions, the HT estimator is inadmissible. We discuss various improvements to the HT estimator, which are inspired by the survey sampling literature, that aim to reduce the variance at the cost of a mild bias.

\paragraph{New Designs}
	We consider new designs for estimating causal effects when there is treatment interference. There are two key considerations when thinking about new designs. The first consideration is that the optimality of a design may depend on the estimand: A design that is considered optimal for estimating the direct effect may be far from optimal for estimating the total effect. The second consideration is that the optimal design may need to depend on the interference graph and the exposure model. Classical designs such as CRD and Bernoulli designs are oblivious to the interference graph and the exposure model. They can generate units with potential outcomes that are nuisance when estimating the direct and the total effect. 
	
	To this end, we discuss two designs, one old and one new for estimating the direct and the total effect under the symmetric exposure model, when the interference graph is known. For estimating the direct effect, we develop a new design inspired by the concept of an independent set in graph theory. The independent set design attempts to maximize the number of units that reveal the relevant potential outcomes required for estimating the direct treatment effect. For estimating the total effect, we consider the cluster randomized design discussed in \cite{ugander2013graph}.
	
\paragraph{Optimality of Estimation Strategies.}
	We evaluate several estimation strategies for estimating the total and the direct effect using simulation studies. The key lessons of the simulation studies can be summarized as follows: The bias of the difference of means estimator in estimating the direct effect depends on the unknown interference effects. Estimation strategies that are unbiased for one estimand may be severly biased for a different estimand. For e.g. we find that the Independent set designs along with any estimator is approximately unbiased for the \emph{direct effect} and has superior performance in terms of mean squared error when compared with other designs. On the other hand, the cluster randomized design along with any estimator is approximately unbiased for estimating the total effect. Moreover, the Horvitz-Thompson estimator has the worst performance in terms of mean-squared error - even the biased naive difference-in-means estimator is beats it.
	
	 %\emph{variance estimation} We briefly consider the important topic of variance estimation.  Unbiased estimators of variance may not exist! Need the joint exposure probabilities to be positive. 

 %%% %%% %%%
 %%% %%% %%%
 %%% %%% %%%

\section{Revisiting the Potential Outcomes framework under arbitrary treatment interference}
\label{sec:POandNetwork}
%(Do this after sections 3--5 are filled in.)
%
%Potential outcomes and table of science
%
%SUTVA vs NIA
%
%Define the ATE, and discuss variants of it
%
% %%% %%% %%%
% %%% %%% %%%
% %%% %%% %%%

In this section, we revisit the definition of potential outcomes when there is arbitrary treatment interference. We develop a framework for specifying models for potential outcomes under interference. Such models are necessary when there is treatment-interference. We consider two classes of causal effects and study the conditions under which unbiased estimators exist for estimating causal effects. 
 
 \subsection{Potential Outcomes under arbitrary interference}
 \label{sec:PODef}
 Consider a finite population of $n$ units indexed by the set $\{1,\ldots, n\}$ and a binary treatment $z_i \in \{0,1\}$ for each unit $i$. Let $\textbf{z} = \left(z_1, \ldots, z_n \right)$ denote the vector of treatment assignments. Let $\Omega$ be the set of \emph{relevant} treatment assignments and let $|\Omega| = m$. %The definition of \emph{relevant} is intentionally left vague, as it depends on the context. 
 In general, $\Omega = \{0,1\}^n$ and $m = 2^n$. %Depending on the application, not all possible $2^n$ may be relevant for a study. Moreover, for defining potential outcomes, some treatment assignments may be considered equivalent to each other, see below {\color{red} below where?}. For now, it suffices to assume that $\Omega \subsetneq \{0,1\}^n$ and  $m \leq 2^n$.
   
 Under arbitrary treatment interference, let $Y_i(\textbf{z})$ be the fixed potential outcome of unit $i$ under the treatment assignment vector $\textbf{z}$. The potential outcome of a unit $i$ can also be considered as a function from the set of possible treatment assignments $\Omega$ to $\mathbb{R}$. For example, in case of a binary outcome, $Y_i(\textbf{z}): \Omega \rightarrow \{0,1\}$. Under this notation, the potential outcome of unit $i$ depends on the treatment assignment of all units under the study. Thus, there are a total of $n \times m$ potential outcomes, which can be assembled in the form of an $n \times m$ table, as shown in Table \ref{tab:tableofScience}. The rows in Table \ref{tab:tableofScience} correspond to the units and the columns correspond to the treatment assignments; the $(i,j)^{th}$ entry corresponds to the potential outcome of unit $i$ under treatment represented by column $j$. This table is referred to as the \emph{Table of Science} and denoted by $\mathbb T$.
 
 \begin{remark}
 	We have made an implicit assumption of \emph{no hidden versions of a treatment} which appears as the second part of the SUTVA, see section \ref{sec:existenceofEstimators} for more details. 
 \end{remark}
 Causal effects are defined as functions of the entries of Table of Science. In particular,  Causal effects can be defined as contrasts between functions of potential outcomes under two distinct treatment assignments. For example, let $\textbf{z}_0$ and $\textbf{z}_1$ be two distinct treatment allocations in $\Omega$, i.e., $\textbf{z}_0 \neq \textbf{z}_1$, then an example of a causal effect is 

 $$\frac{1}{n} \left(\sum_i Y_i(\textbf{z}_1) - \sum_i Y_i(\textbf{z}_0) \right).$$ 
 In Section \ref{sec:estimands} we consider two different classes of estimands or causal effects.
 The fundamental problem of Causal Inference is that the table of science is unknown and only one entry of Table \ref{tab:tableofScience} can be observed. 
 
 More specifically,  let $\textbf{Z} = (Z_1, \ldots, Z_n)$ denote the random vector of treatment assignments and  $p(\textbf{Z} = \textbf{z})$ be a probability distribution defined over the set of all possible treatment assignments $\Omega$. $p(\bf{Z})$ is called the treatment assignment mechanism or a design. In many cases, we can also restrict ourselves to $\Omega_{p} = \{\textbf{z}: p(\textbf{Z} = \textbf{z}) > 0\}$, the support of the treatment assignment mechanism. %We will consider restricted randomizations, and hence $\Omega \subsetneq \{0,1\}^n$.
 Under a random treatment assignment $p(\textbf{Z})$, without any further assumptions, only one random entry of each row of Table \ref{tab:tableofScience} can be observed, i.e. for each unit $i$, only one of it's potential outcome can be observed. For example, if the realized treatment $\textbf{Z}$ corresponds to column $j$, then only column $j$ is observed. Since causal effects are defined as contrasts between two different treatment assignments, they cannot be estimated if only one column is observed. %If $\Omega = \{0,1\}^n$ and there are no further restrictions, then each unit has $2^n$ possible values of potential outcomes.

\begin{table} 
\caption{Table of Science}
\centering
\label{tab:tableofScience}
 \begin{tabular}{l|cccccc}
 	\hline 
 	& \multicolumn{6}{c}{Treatment} \\
 	\cline{2-7} 
 	Units & $1$ & $2$ & $\ldots$ & $j$ & $\ldots$ & $m$ \\
 	\hline 
 	$1$ & $Y_1(\bf{z}_1)$ & $Y_1(\bf{z}_2)$ & $\ldots$ & $Y_1(\textbf{z}_j)$ & $\ldots$ & $Y_1(\textbf{z}_m)$ \\
 	$\vdots$ & $\vdots$ & $\vdots$ & $\ddots$ & $\vdots$ & $\ddots$ & $\vdots$ \\
 	$i$ & $Y_i(\bf{z}_1)$ & $Y_i(\bf{z}_2)$ & $\ldots$ & $Y_i(\textbf{z}_j)$ & $\ldots$ & $Y_i(\textbf{z}_m)$ \\
 	$\vdots$ & $\vdots$ & $\vdots$ & $\ddots$ & $\vdots$ & $\ddots$ & $\vdots$  \\
 	$n$ & $Y_n(\bf{z}_1)$ & $Y_n(\bf{z}_2)$ & $\ldots$ & $Y_n(\textbf{z}_j)$ & $\ldots$ & $Y_n(\textbf{z}_m)$ \\
 	\hline 
 	%Mean & $\bar{Y}_{.1.}$ & $\overline{Y}_{.2.}$ & $\ldots$ & $\overline{Y}_{.j.}$ & $\ldots$ & $\overline{Y}_{.e.}$ \\
 	%$\hline 
 \end{tabular}
\end{table}

\begin{proposition}
	\label{prop:impossible}
 	Causal effects are unidentifiable without any assumptions on the potential outcome functions $Y_i(\bf{z})$.
 \end{proposition}
 \begin{proof}
 	Since there are no further assumptions on the potential outcomes, only one entry of the table of science is observable due to the fundamental problem of causal inference.  As causal effects are defined as contrasts between two distinct treatment assignments, they are unidentifiable as only one entry of the Table \ref{tab:tableofScience} is observed.
 \end{proof}
 
 Proposition \ref{prop:impossible} is a simple observation but has profound consequences. It implies that Causal Inference is impossible without further assumptions on the potential outcomes. Hence we are forced to make modeling choices to make progress. Indeed the bulk of Causal Inference since past 40 years has been centered around the no-treatment interference assumption, which is embedded in \emph{SUTVA} assumption. One can consider the no-interference assumption as a very specific model on the potential outcomes. Under SUTVA, $Y_i(\textbf{z}) = Y_i(z_i)$ and Table \ref{tab:tableofScience} reduces to an $n \times 2$ table. Thus, the question is not ``why a model'', but rather ``which model''? We discuss a series of modeling assumptions on potential outcomes that allow tractable causal inference.

 \subsection{Modeling Potential Outcomes under Network Interference}
 \label{sec:POmodel}
 In this section, we describe models for potential outcomes when there is arbitrary interference due to treatment. As we saw in the previous section, when there is arbitrary interference, modeling potential outcomes becomes necessary without which Causal Inference is impossible. Our framework makes these modeling choices easy to specify and transparent to present. This framework unifies existing models for potential outcomes under treatment interference - many existing models can be instantiated as special cases of our framework. We also develop a linear decomposition of potential outcomes that is useful for interpreting causal effects and studying estimators.
 
 Models of potential outcomes are specified by specifying three different components: an \emph{interference neighborhood}, an \emph{exposure model} and a \emph{structural model}. These components are hierarchical in nature. Each of these components build upon the other, and they need to be defined in this order. For example, to define an exposure model, we need to define the interference neighborhood, and so on. We give an informal description of these components before moving to the formal definitions.
 
 The \emph{interference neighborhood}, denoted by $N_i$, defines the set of units whose treatment assignment can potentially influence unit $i$'s outcomes. Any unit outside the interference neighborhood cannot influence $i$'s potential outcomes. For example, in educational studies, $N_i$ can be the school that unit $i$ belongs to. Any unit outside unit $i$'s school does not effect the outcome of unit $i$. In this example, interference neighborhoods can be partitioned into non-overlapping sets. In more general settings, e.g. in the context of social networks or vaccination studies, the interference neighborhoods of units may overlap with each other and can be more complex. Next, the \emph{exposure model} defines what it means for a unit to be \emph{exposed} and defines the set of relevant exposure conditions of a unit $i$. For example, a unit $i$ may be said to be \emph{exposed} to the treatment if all the units in it's interference neighborhood are treated, or if a fraction of them are treated and so on. The exposure level of a unit need not be a binary variable, but a continuous quantity. For example, it could be the case that there is a gradual increase in exposure, i.e. as more and more units in $i$'s interference neighborhood get treated, $i$ gets ``more'' exposed. Finally, a \emph{structural model} defines or imposes structural constraints on different potential outcomes of each unit $i$. One can think of structural models as a way to specify null hypothesis of interests on individual level potential outcomes. For example, a linear model specifies that the potential outcomes are linearly related to the treatment and the exposure conditions. We will discuss these three components in more detail and give their formal definitions along with several examples.

 \subsubsection{Interference Neighborhood}
 The \emph{interference neighborhood} or \emph{neighborhood} of a unit $i$ (not to be confused with the neighborhood of a node in a graph) is denoted by $N_i \subset [n] \backslash \{i\}$ and is defined as the set of units whose treatment status may effect the outcome of unit $i$. Let $\textbf{z}_{N_i}$ denote the vector $\textbf{z}$ sub-setted by the indices in $N_i$. Let $\textbf{z}$ and $\textbf{z}'$ be two distinct potential outcomes. Then given a choice of the interference neighborhood for each unit $i$, we make the following assumption:
 \begin{align}
 \label{assumption:nia}
 Y_i(\textbf z) = Y_i( \textbf z') \text{ iff } \textbf{z}_{N_i} = \textbf{z}'_{N_i}
 \end{align}
 This allows us to write down the potential outcome of each unit $i$ in the following manner:
 \begin{align}
 \label{eq:NIA}
 Y_i(\textbf z) = Y(z_i, \textbf{z}_{N_i}),
 \end{align} 
 where $z_i$ denotes the treatment assigned to unit $i$  and $\textbf{z}_{N_i}$ denotes the vector of treatment assigned to units in the \emph{interference neighborhood} of unit $i$. 
 
 \begin{remark}
Note that the interference neighborhood of each unit can be different and hence $\textbf{z}_{N_i}$ can be of different length. Moreover, a unit $i$ may be in unit $j$'s interference neighborhood, but $j$ may not be in $i$'s neighborhood. Finally, Interference neighborhoods of two units may overlap, they may be disjoint or they can also be the same. 
 \end{remark}
 We will now consider two simple, but extreme examples of interference neighborhoods. 
 \begin{enumerate}
 	\item \emph{No treatment interference}: $N_i = \emptyset$ for each unit $i$
 	\item \emph{Complete interference}: $N_i = [n]/\{i\}$
 \end{enumerate}
 The simplest example is the setting of no treatment interference, which amounts to saying that the outcome of unit $i$ does not depend on the treatment of any other unit. At the other extreme is complete interference, where the treatment of every unit can effect the outcome of unit $i$. 
 The first example reduces to the classical SUTVA setting, and in the second example, there is no causal inference possible, unless we make additional assumptions (specified by an exposure model and/or a structural model to be defined below). The most interesting cases are when we can consider interference neighborhoods that lie in between no-interference and complete-interference. To model these intermediate cases, it turns out to be convenient to define interference neighborhoods using a graph.
 
 \paragraph{Graph Induced Interference Neighborhoods}
 A convenient way to specify the interference neighborhood of a unit $i$ is by the means of an \emph{interference} graph. Let $G$ be a fixed, known graph on $n$ nodes with $V$ as its vertex set and $E$ as its edge set. %We assume that if $i$ belongs to $j$'s interference neighborhood then a directed path exists from unit $i$ to unit $j$.
 %Note that the converse is not true, i.e. if a path exists from unit $i$ to unit $j$, it does not imply that $i$ belongs to $j$'s interference neighborhood. The absence of a path from $i$ to $j$ rules out the possibility of $i$ being in $j$'s interference neighborhood.  
 The introduction of an interference graph allows us to introduce additional structure into the nature of interference. 
 %For example, consider a sample of $n$ units and let $N_1 = \{2,3,4,5\}$ be the interference neighborhood of unit $1$. It may be the case that the effect of units $2$ and $3$ on the outcome of unit $1$ is ``less'' than the effect of units $4$ and $5$. This may be because units $2$ and $3$ effect unit $1$ only through units $4$ and $5$. Such a structure can be captured by the introduction of a graph and placing units that have ``more'' effect on $1$ closer in the graph $G$. Continuing with the example, units $4$ and $5$ would be directly connected to $1$ in $G$ and $2$ and $3$ would be  
 Note that in general, $G$ can be asymmetric and even weighted. For simplicity of notation, we will assume for the rest of the paper that $G$ is symmetric and un-weighted, i.e. if $g_{ij}$ denotes the edge from unit $i$ to unit $j$, we will assume $g_{ij} = g_{ji}$. All these ideas apply to an asymmetric weighted graph with additional notation. 

 %\begin{remark}
 	%$G$ can be connected, in which case the interference neighborhoods may overlap. In other cases, $G$ may consist of disconnected components., for example, housing studies of 
  %\end{remark}
 %\begin{remark}
 
 %\end{remark}
 We now consider a few examples of graph induced interference neighborhood:
 \begin{enumerate}
 	\item \emph{1 hops interference}: $N_i = \{j \in V: g_{ij} = 1\}$ %{\color{red} We may need a different name to avoid confusion between the general neighborhood $N_i$ and the neighborhood of a graph}
 	\item \emph{2 hops interference}: $N_i = \{j \in V : \exists k \mbox{ such that } g_{ik} = 1 \mbox{ and } g_{kj} =1 \}$
 	\item \emph{$k$ hops interference}: $N_i = \{j \in V: \exists \mbox{ a path of length at most $k$ connecting $i$ and $j$ in } G\}$
 \end{enumerate}  
 %Many other definitions of exposure can be specified by the use of a graph $G$ associated with the units.  
 
  \begin{remark}
  	The interference graph is an abstract representation of the interference that may exist in the real world setting. In general $G$ may not be observable, random or may not even be well defined. How does one choose $G$? This is an important question and beyond the scope of this paper. But we will give some remarks. In many cases, $G$ may be clear from the study. For example, consider the setting of \emph{partial interference}. In this setting, the $n$ units can be partitioned into $m$ disjoint groups $K_1,\ldots, K_m$. Interference may happen within the groups but not between the groups, see for e.g. \citet{sobel2006randomized} or \cite{hudgens2012toward}. The interference graph in this case consists of a collection of $m$ disjoint cliques. In many other settings, one may observe a social graph which can serve as a good approximation for $G$, (e.g. Facebook). It may also be the case that we observe a network but posit that the interference may happen only along stronger social ties, for e.g. frequently contacted friends, as opposed to all friends in a social network. In such cases, the interference graph $G$ may be an induced subgraph of the social graph. One may also consider $G$ as random and posit a distribution over $G$. This leads to additional complexities that are beyond the scope of this paper.
  \end{remark}
  
 For the remainder of the paper, we will assume that the interference neighborhood $N_i$ for each unit $i$ is defined through a fixed graph $G$ on $n$ units. 
 \subsubsection{Exposure Models}
After defining the interference neighborhood, there are two modeling choices remaining for specifying potential outcomes and for making causal inference tractable (i.e. to ensure that the table of science as shown in Table \ref{tab:tableofScience} not too wide) - the so called \emph{exposure} model and the \emph{structural} model.  The \emph{exposure} model specifies how the treatment status of units in $Z_{N_i}$ effect the outcome of $i$. It defines the relevant levels of \emph{exposure} and how the treatment levels of the interference neighborhood get mapped to these levels. 

%It also specifies how the units in $\textbf{z}_{N_i}$ effect the exposure level of a unit $i$. 

Formally, the exposure model is specified by an \emph{exposure} function $f$ that maps $\textbf{z}_{N_i}$ to a range $\mathcal E_i$. The range of $f$ specifies the relevant exposure levels and the mapping $f$ specifies how the treatment patterns of $\textbf{z}_{N_i}$ map to different exposure levels. To this end, let us assume that the potential outcome function $Y_i(z_i, \textbf{z}_{N_i})$ depends on $\textbf{z}_{N_i}$ through a function 
$$f: \{\textbf{z}_{N_i} \} \rightarrow \mathcal{E}_i. $$
Let $n_i = |N_i|$ be the number of units in the interference neighborhood of $i$. The domain of $f$ is the set of all possible treatment assignments of the neighborhood of a unit $i$. %For instance, when the treatment is binary, the domain of $f$ consists of $2^{|N_i|}$ binary vectors, each of length $|N_i|$. 
The domain of $f$ has at most $2^{n_i}$ elements and is finite. Hence the range of $f$ is also finite. This is because for each treatment assignment $\textbf{z}_{N_i}$, $f$ can map to at most one exposure level. Let $K_i = |\mathcal{E}_i|$ denote the size of the range. Thus, for every unit $i$, there are $K_i$ different levels of exposure. Without loss of generality, we can write the range of $f$ as $\mathcal{E}_i = \{0, 1, \ldots, K_i-1\}$.

Given an exposure function $f$, let $e_i = f(\textbf{z}_{N_i}) \in \mathcal{E}_i$.
Thus we can write the potential outcome function for each unit $i$ as 
\begin{align}
Y_i(\textbf{z}) = Y_i(z_i,\textbf{z}_{N_i}) = Y_i(z_i,f(\textbf{z}_{N_i})) = Y_i(z_i,e_i)
\end{align}
 and $e_i$ takes values in $\{0, 1,\ldots, K_i-1\}$.

%We will focus on the case where the exposure levels are discrete, but the ideas apply to continuous exposure levels as well.

%If there are no restrictions on $f$, then $K = 2^{n_i}$. 
To specify an exposure model, one must specify the function $f$ and the levels of exposure $\{0,1,\ldots K_i-1\}$. The total number of exposure patterns $K_i$ depends on the choice of $f$ and $\textbf{z}_{N_i}$. When $f$ is a one to one mapping, there are a total of $2^{n_i}$ levels of exposure for each unit $i$. Clearly, an $f$ that is onto reduces the number of exposure levels and hence the total number of possible potential outcomes. %The \emph{exposure function} $f$ also specifies how the treatment status of nodes in the neighborhood $N_i$ of a unit $i$ influences the outcome of unit $i$. 
 
In the most general case, one can set $N_i = \textbf{z}_{-i}$ and $f$ to be a one to one function. In this case, $K_i= 2^{n-1}$ and there is no reduction in the number of potential outcomes. On the other extreme, when $N_i = \emptyset$, we are in the setting of no interference. Intermediate cases are more interesting and can be defined by a network interference graph $G$.  To ensure identifiability we need $\max_i K_i < 2^{n-1}$.
 
 \begin{remark}
 	Note that $e_i$ is a way of indexing the $K_i$ different types of possible exposure patterns or levels and the symbols $0,1,\ldots, K_i-1$ denote these exposure patterns as defined by the function $f$. The interpretation of the symbols is a choice of the definition of $f$. For example, statements such as ``a unit is exposed if $ 10 \%$ of it's neighbors are treated'' can be the modeled by mapping $\{0, \ldots, K_i-1\}$ to the appropriate fractions. 
 \end{remark}
 We will define two special symbols for two commonly used values of the exposure patterns: $e_i = 0$ is called no exposure and $e_i = 1$ is full exposure. The exposure function $f$ also specifies what it means for a unit to be fully exposed and not exposed. We give two examples:
 \begin{enumerate}
 	\item  $e_i =0$, when all elements of $\textbf{z}_{N_i}$ are 0 and $e_i=1$ when all elements are 1.
 	\item $e_i=0$ when all elements of $\textbf{z}_{N_i}$ are $0$ and $e_i=1$ when at least one element is $1$.
 \end{enumerate}

 \begin{remark}
 	A possible exposure function is one that maps $Z_{N_i}$ to the number of units in $N_i$ that are treated. One subtle issue with choosing such an exposure function is that the levels of exposure function depends on the maximum degree in the graph $G$. If $G$ is not a regular graph, i.e. the degree of each node is different, then the exposure levels of each unit is different, which may not be desirable depending on the application. These are subtle issues that need to be resolved and are out of the scope of our paper. 
 \end{remark}
 %We want to point out that $e_i$ is a simply a symbolic representation of different exposure conditions. Thus, we will use $e_i = 0$ to denote the case when all units in the neighborhood of $i$ are not treated. Similarly, $e_i=1$ denotes that all units in the neighborhood of unit $i$ are treated.
Let us consider a few  examples of exposure functions:

 \begin{enumerate}
 	\item \emph{Symmetric Exposure:} $ f(\textbf{z}_{N_i})$ is symmetric  in the indices of $\textbf{z}_{N_i}$
 	\item \emph{Linear and Additive Exposure:} $f(\textbf{z}_{N_i}) = \sum_{j \in N_i}h_j(z_j)$
 	\item \emph{Linear Exposure}: $f(\textbf{z}_{N_i}) = \sum_{j \in N_i} \textbf{z}$.
 \end{enumerate}
 It is also possible to define more complex exposure functions. For example, consider a setting when the interference neighborhood is specified through a graph $G$. The interference neighborhood of unit $i$ is the set of units in $G$ that have a connected by a path of size $2$ to $i$, i.e. through \emph{friends} and \emph{friends of friends} of unit $i$. The exposure function $f$ can be parametric that allows the potential outcomes to depend on $Z_{N_i}$ through a weighted combination of the number of treated friends in $G$ and the number of treated friends of treated friends.
 
% Of particular interest is the symmetric exposure function. The symmetric exposure condition reduces the number of possible exposure conditions from $2^{n_i}$ to $n_i$, as Proposition \ref{prop:symmetric} shows.
%
% \begin{proposition}[Symmetric Exposure]
%\label{prop:symmetric}
% 	Let $\pi$ be any permutation on $N_i \subset [n]\backslash \{i\}$. If $f(\textbf{z}_{N_i}) = f( \pi \textbf{z}_{N_i})$, then 
% 	$$f(\textbf{z}_{N_i}) = f(|\textbf{z}_{N_i}|).$$
% \end{proposition}
% \begin{proof}
% 	Consider a function $f(z_1, z_2, z_3)$ of three variables, such that each $z_i \in \{0,1\}$. Since $f(\cdot)$ is symmetric, we have
% 	
% 	$f(z_1, z_2, z_3) = f(z_{i_1}, z_{i_2}, z_{i_3})$ for any permutation $(i_1,i_2,i_3)$ of $\{1,2,3\}$.
% 	
% \end{proof}

\begin{remark}
Note that we have made an assumption that the exposure function $f$ is independent of the unit $i$, i.e. we do not allow $f$ to depend on $i$. For example, we do not allow exposure functions where unit $i$'s exposure depends on the number of treated friends and unit $j$'s exposure depends on both the number of treated friends and the number of treated friends of friends. However, the range of $f$ may depend on $i$.
 \end{remark}

\subsubsection{Structural Models}
  
\paragraph{Parametrization of Potential Outcomes under Neighborhood Interference}
Before defining a structural model, it is convenient to introduce a parametrization or a linear decomposition of potential outcomes into direct and indirect effects. This parametrized form of potential outcomes allows one to define and focus on various treatment effects of interests. We present one such parameterization. When $N_i$ is the interference neighborhood and $e_i = f(\textbf{z}_{N_i})$ is the exposure model, every unit $i$ has $2K_i$ potential outcomes, that can be parametrized by $2K_i$ parameters, as given in Proposition \ref{prop:param}.
   
  \begin{proposition}
  	\label{prop:param}
  	For each unit $i$, let $e_i = f(\textbf{z}_{N_i})$ where $N_i$ is the interference neighborhood. The potential outcomes can be parametrized as  
  	\begin{align}
  	 	\label{eq:param}
  	Y_i(\textbf{z}) =  A_i(z_i) + B_i(e_i) + z_iC_i(e_i)
  	\end{align}
  	where $e_i = f(\textbf{z}_{N_i}) \in \{0, \ldots, K_i-1\}$, and $B_i(0) = C_i(0) = 0$.
  \end{proposition}

\begin{remark}
	Equation \ref{eq:param} resembles a linear model for Potential Outcomes. However, it is not a linear model in the usual sense of linear regression. Unlike regression, which is a model of conditional expectation, there are no random variables. Moreover, the parameters are not linear, and they depend on $i$.
\end{remark}
This parametrization has a nice interpretation: The $A_i$ parameters represent the direct treatment effects or the part of the Potential outcome that depends only on a unit $i$'s treatment. The $B_i$ and  $C_i$ parameters represent the indirect or interference effects, i.e. the part of the potential outcome that depends on the exposure level. In particular, the $B_i$ parameters represent the additive interference effect and the $C_i$ parameters represent the interaction between the additive interference and the direct treatment effects, see also section \ref{sec:estimands}. Given this linear parametrization, we are now ready to specify structural models.

\paragraph{Structural Model} Up to this point, we have made no assumptions on how the potential outcomes relate to each other, we have only focused on reducing the number of potential outcomes. However, in some cases, we may also make additional modeling assumptions on how one potential outcome relates to another. Sometimes these assumptions serve as null hypothesis for treatment effects. These are called \emph{structural} modeling assumptions as they impose a structure on different potential outcomes. Given the parameterization in equation \ref{eq:param},  structural assumptions can be regarded as restrictions on the parameters of the potential outcomes. Without any structural assumptions, the parameterizations are functionally independent of each other. Structural assumptions make the parameters functionally dependent. Examples include, linear models, additive models and so on. Some examples of structural assumptions are stated below:

\begin{enumerate}
	\item \emph{Additivity}: $C_i(e_i) = 0 \forall e_i \in \{1, \ldots, K_i-1\}$.
	\item \emph{Constant effects:} $A_i(z_i) = A(z_i)$, $B_i(e_i) = B(e_i)$ and $C_i(e_i) = C(e_i)$.
	\item \emph{Linear Effects:} $B_i(e_i) = b_i e_i$
	\item \emph{Constant Additive Effects} $A_i(z_i) = A(z_i)$ and $C_i(e_i) = 0$. 
	\item \emph{Sharp Null:} $A_i(1) - A_i(0) = \beta \forall i$.
\end{enumerate}

\begin{remark} The interference neighborhood reduces the number of potential outcomes for each unit $i$ from $2^n$ to $2^{n_i}$.  An exposure model of the Potential Outcomes further reduces the number of potential outcomes for each unit $i$ from $2^{n_i}$ to a tractable number $K_i$. On the other hand, a \emph{structural model} of the Potential Outcomes specifies a relationship between the $K_i$ different potential outcomes by imposing constraints on the parameters. 
\end{remark}
% {\color{red} There is some confusion about exposure model and structural model. I think we need one more term to define the interference neighborhood} 
 \subsubsection{Some models of Potential Outcomes under interference}
Different choices of the interference neighborhood, the exposure function and the structural model lead to different models for the potential outcomes. In this section, we present specific choices that give rise to some models used in the paper. These examples illustrate how one can use the framework to reduce the number of potential outcomes and model them. We start with the parametrized model of potential outcomes modeled using the interference neighborhood $N_i$ and the exposure function $e_i=f(\textbf{z}_{N_i})$:
\begin{align}
\label{eq:NIparam}
Y_i(z_i,e_i) = A_i(z_i) + B_i(e_i) + z_iC_i(e_i) = \alpha_i + \beta_i z_i + B_i(e_i) + z_iC_i(e_i)
\end{align}

where $A_i(0) = \alpha_i$ and $A_i(1) = \alpha_i + \beta_i$.

\paragraph{Symmetric Exposure Models}
\paragraph{Model 1: Symmetric Exposure} Let $N_i$ denote the neighborhood of a unit $i$ as specified by the interference graph $G$, i.e.   
$$N_i = \{j: g_{ij} = 1\}$$
Next, let $f(\textbf{z}_{N_i}) = \sum_{j \in N_i} z_j$, then we get the following model:
\begin{align}
\label{eq:symmetricExposure}
Y_i(z_i,e_i) &= \alpha_i + \beta_i z_i + B_i(e_i) + z_iC_i(e_i)
\end{align}
where $e_i  = \sum_j g_{ij}z_j \in \{0,1,\ldots, d_i\}$ and $d_i$ is the degree of unit $i$ in the interference graph $G$. Under this model, each unit has $2(d_i+1)$ potential outcomes.

Starting with equation \ref{eq:NIparam}, we can make additional structural assumptions to get simpler models for the potential outcomes. We give two examples below.

\paragraph{Model 2:  Symmetric Linear exposure} In Model 1, Let $B(e_i) = \gamma_i e_i$ and $C_i(e_i) = \theta_i e_i$
\begin{align}
\label{eq:LinearSymmExposure}
Y_i(z_i,e_i) &= \alpha_i + \beta_i z_i + B_i(e_i) + z_iC_i(e_i)\nonumber \\
&= \alpha_i + \beta_i z_i + \gamma_i e_i  + \theta_i z_i e_i\nonumber \\
&= \alpha_i + \beta_i z_i + \gamma_i  \left(\sum_j g_{ij}z_j \right) +  \theta_i  z_i \left(\sum_j g_{ij}z_j \right)
\end{align}

\paragraph{Model 3: Symmetric  Additive  Linear exposure} In Model 2, Let $\theta_i =0$
\begin{align}
\label{eq:LinearAddSymmExposure}
Y_i(z_i,e_i)
&= \alpha_i + \beta_i z_i + \gamma_i e_i\nonumber \\
&= \alpha_i + \beta_i z_i + \gamma_i  \left(\sum_j g_{ij}z_j \right)
\end{align}

\paragraph{Binary Exposure Models}
The binary exposure model is te simplest exposure model that weakens the no-interference assumption. In these models, the range of the exposure function is always $\{0,1\}$, where $0$ is interpreted as not exposed and $1$ is interpreted as exposed. The definition of $f$ specifies which treatment levels get mapped to $0$ or $1$ and is chosen based on the application. In the binary exposure model, each unit has $4$ potential outcomes, in contrast to 2 potential outcomes per unit in the SUTVA case. The binary exposure model is the simplest possible model of potential outcomes when there is interference.
We present below a simple but natural choice of such a binary exposure function where a unit is said to be exposed if at least one of it's neighbor is treated. 
\paragraph{Model 4: Binary exposure} Let $N_i$ be $$N_i = \{j: g_{ij} = 1\}.$$ Let $f(\textbf{z}_{N_i}) = 0$ if all elements of $Z_{N_i}$ are $0$ and $1$ if at least one element of $Z_{N_i}$ is $1$. This gives us the so called \emph{two by two} potential outcomes model or the binary exposure model:
\begin{align}
\label{eq:BinaryExposure}
Y_i(z_i,e_i) = \alpha_i + \beta_i z_i + \gamma_i e_i + \theta_i z_i e_i
\end{align}
where $e_i = I\left(\sum_i{g_{ij}}z_j > 0\right) \in \{0,1\}$ and $z_i \in \{0,1\}$.

As in the previous case, one can impose structural assumptions on the binary exposure model to generate simpler models. 
\paragraph{Model 5: Additive Binary exposure} Let $\theta_i = 0$ in model 4, then we get the additive two by two model of potential outcomes
\begin{align}
\label{eq:addBinaryExposure}
Y_i(z_i,e_i) = \alpha_i + \beta_i z_i + \gamma_i e_i
\end{align}

 \subsection{Defining Causal Effects under Network Interference} 
 \label{sec:estimands}
Given an interference neighborhood $G$ and a corresponding exposure function $f(\textbf{z}_{N_i})$, causal effects or estimands are defined as contrasts between potential outcomes under distinct treatment and exposure assignments. Under interference, there are many non-equivalent ways of defining causal effects. The definition of the estimand depends on the question one is interested in answering. It is important to note that the estimands are defined only using the table of science, and they do not depend on the actual treatment assignment mechanism used to estimate them.

For a given exposure function $f(\cdot)$, each unit $i$ has $2K_i$ distinct potential outcomes denoted by $Y_i(z_i,e_i)$ where $z_i \in \{0,1\}$ and $e_i \in \{0,\ldots, K_i-1\}$. These potential outcomes can be assembled in the form of a $n \times m$ Table of Science as before. The number of columns of the Table of Science in Table \ref{tab:tableofScience} reduce from $2^n$ to  $2\cdot (\sum_{i=1}^n K_i)$ columns where $K_i$ is the number of exposure levels for each unit $i$. The columns of the table of science now correspond to the relevant treatment and exposure conditions $(z_i,e_i)$ as specified by $f$. For instance, under the binary exposure model, the table of science has $4$ columns and $n$ rows. This is the simplest setting of a Table of Science that relaxes the no-treatment interference assumption. Causal effects are functions of at least two columns of $\mathbb T$.  

Before we define causal effects, we need some additional notation. A fixed treatment assigned vector $\textbf{z}$ gets mapped to different treatment and exposure combination $(z_i,e_i)$ for each unit $i$. Moreover, different treatment assignment vectors can get mapped  to the same treatment and exposure combination. Formally, let $z_0$ and $e_0$ denote a generic treatment and exposure condition. Let 
 $$
 \Omega_i(z_0,e_0) = \{\textbf{z} \in \Omega: z_i = z_0, f(\textbf{z}_{N_i}) = e_0 \} 
 $$
 be the set of all treatment assignment vectors that give rise to treatment $z_0$ and exposure $e_0$ for a unit $i$.
%Given a fixed exposure and treatment combination $(z_0,e_0)$, we can consider the set of all treatment assignments that give rise to a unit $i$ being allocated to treatment $z_0$ and exposure level $e_0$. Let $\Omega_i(z_1,e_1) = \{\textbf{z} \in \Omega: z_i = z_1, e_i = e_1 \}$. Similarly, let $\Omega_i(z_0,e_0) = \{\textbf{z} \in \Omega: z_i = z_0, e_i = e_0 \}$.

 We will consider two classes of estimands: \emph{marginal} and \emph{average} causal effects: The \emph{marginal effects} are  defined as contrasts between two different randomized treatment policies. The \emph{average effects} are a contrast between two different types of potential outcomes. In some cases, both definitions can lead to the same estimand, but it is not true in general. 
 
\paragraph{Marginal Effects}Let us first consider the marginal effects that are defined as contrasts between two randomized treatment assignment mechanisms. We will refer to such treatment assignments as \emph{policies} to distinguish them from the actual treatment assignment mechanism used in the experiment. For instance, a policy can be to treat randomly chosen $10 \%$ of the units in the population, or to treat $5 \%$ of the units in the population and so on. Let $\phi$ and $\psi$ be two policies, i.e. $\phi(\textbf{Z})$ and $\psi(\textbf{Z})$ are two distributions over $\textbf{Z}$. Similar to \cite{hudgens2012toward}, we define conditional and marginal potential outcomes of a unit $i$ as expectations of potential outcomes under a treatment policy. Let $E_i$ denote the random exposure condition of unit $i$.
\begin{definition}[Conditional and Marginal Potential Outcomes]
	  \begin{align*}
	  \bar{Y}_i(z_i;\phi) &= \E_{\phi}\left[Y_i(Z_i,E_i)|Z_i=z_i\right] = \sum_{e_i} Y_i(z_i,e_i)\phi(E_i=e_i|Z_i=z_i) \\
	  \bar{Y}_i(\phi) &= \E_{\phi}\left[Y_i(Z_i,E_i)\right] = \sum_{e_i, z_i} Y_i(z_i,e_i)\phi(Z_i=z_i,E_i=e_i) 
	  \end{align*}
\end{definition}
Here $\phi(Z_i=z,E_i=e) = \sum_{\textbf{z} \in \Omega_i(z,e)}\phi(\textbf{Z} = \textbf{z})$ and $\Omega_i(z,e) = \{\textbf{z} \in \Omega: z_i = z, e_i = e \}$.  Given the conditional and marginal potential outcomes, various causal effects can be defined as follows:
  \begin{align*}
  \theta(\phi) &= \frac{1}{n}\sum_{i=1}^n \bar{Y}_i(1;\phi) - \bar{Y}_i(0;\phi) \\
  \theta(\phi;\psi) &= \frac{1}{n}\sum_{i=1}^n \bar{Y}_i(\phi) - \bar{Y}_i(\psi) \\
\theta(\phi;\psi,z)  &= \frac{1}{n} \sum_{i=1}^n \bar{Y}_i(z_i=z;\phi) - \bar{Y}_i(z_i=z;\psi)
  \end{align*}
   One can define total, direct and indirect causal effects using these definitions, and consider various decompositions among them.
  
\paragraph{Average Causal Effects} An alternate way to define causal effects is to consider contrasts between two fixed types of potential outcomes. Let us consider a generic causal estimand $\beta$ defined as a contrast between two different treatment and exposure combinations: $\tau_1 = (z_1,e_1)$ and $\tau_0 = (z_0,e_0)$:
\begin{align}
\beta = \frac{1}{n}\sum_{i=1}^n \left(Y_i(z_1,e_1) - Y_i(z_0,e_0)\right)
\end{align}

The most popular average causal effects are the Average treatment effects and the Average interference effects.  We will consider two types of average treatment effects: the direct treatment effect (DTE) and the total treatment effect (TTE) that are defined below. Recall that the exposure levels $e_i=0$ and $e_i=1$ are special values defined to represent the situations when a unit is not exposed or exposed respectively.  Direct effects are defined as contrasts between the conditions $\tau_1 = (1,0)$ and $\tau_0 = (0,0)$, i.e. when a unit is treated and not exposed vs when a unit is neither treated nor exposed.

\begin{definition}[Direct Treatment Effect]
 $$DTE  = \beta_{DE} = \frac{\sum_i Y_i(1,0)}{n} - \frac{\sum_i Y_i(0,0)}{n}$$
\end{definition}
Similarly, one can consider the total treatment that is a contrast between $\tau_1 = (1,1)$ and $\tau_0 = (0,0)$ when a unit is treated and exposed vs when a unit is neither treated nor exposed.

\begin{definition}[Total Treatment Effect]
 $$TTE = \beta_{TE} = \frac{\sum_i Y_i(1,1)}{n} - \frac{\sum_i Y_i(0,0)}{n} $$
\end{definition}
Similarly, one can also define average estimands that measure interference effects: 
\begin{definition}[Average Interference Effects]
 $$\gamma_1 = \frac{\sum_i Y_i(0,1)}{n} - \frac{\sum_i Y_i(0,0)}{n}$$
 
 $$\gamma_2 = \frac{\sum_i Y_i(1,1)}{n} - \frac{\sum_i Y_i(1,0)}{n}$$
\end{definition}

%{\color{red} Remark on how under SUTVA $\beta_1 = \beta_2 = 0$ etc}
 
 \paragraph{Relation between the estimands} As we saw, there are two ways to define causal effects. The marginal effects defined as a contrast between expectations of the potential outcomes under two different randomization policies, and the average effects defined  as a contrast between averages of potential outcomes. These two definitions are non-equivalent in general.
 However, under SUTVA, the marginal effects and the average effects reduce to the classical ATE, as the following proposition shows:
 \begin{proposition}
 	\label{prop:SUTVAATE}
 	Assume that $Y_i(\textbf{z}) = Y_i(z_i)$. Then we have
 	$\beta_{DE} = \beta_{TE} =  \theta(\phi) = \beta$ where
 	$\beta = \frac{1}{n}\sum_{i=1}^{n}\left(Y_i(1) - Y_i(0)\right).$
 \end{proposition}
  
  \begin{remark}
  	Under the no interference assumption, the direct and the total effect reduce to the ``usual'' classical version of ATE. Moreover, the average interference effects are $0$ under no-interference. But under interference, the direct and the total effects are different, and the average interference effects are not $0$. For example, consider the following linear model of the potential outcomes:
  	
  	$$Y_i(\textbf{z}) = \alpha_i + \beta_i z_i + \gamma_i \left(\sum_j g_{ij}z_j\right) + \delta_i z_i \left(\sum_j g_{ij}z_j\right) $$
  	Let $\beta = (1/n)\sum_i \beta_i)$.
  	Under this model, $DTE = \beta$ and $TTE = \beta + \frac{1}{n}\sum_{i=1}^n{(\gamma_i + \delta_i) d_i} $. where $d_i = \sum_j g_{ij} z_j$ is the number of treated neighbors of unit $i$ in $G$. However, if $\gamma_i = 0$ and $\delta_i = 0$, then $DTE = TTE$. Hence we have the following proposition:
  \end{remark}

  \begin{proposition}
  	\label{prop:ate1vsate2}
  	Consider the linear additive symmetric exposure model of Potential Outcomes given by equation \ref{eq:LinearAddSymmExposure}, and let $\beta_i = \beta$, $\gamma_i = \gamma$ and $\delta_i = \delta$ $\forall i$ and $\bar{d} = \frac{1}{n} \sum_i d_i$. We have
  	$DTE = \beta$, $TTE = \beta + (\gamma+\delta)\bar{d}$, $\gamma_1 =  \gamma \bar{d}$, and $\gamma_2 = (\delta+\gamma) \bar{d}$.
  \end{proposition}
  
 Note that one can obtain $TTE$ as a special case of $\theta(\phi,\psi)$, but this is not the case for $DTE$. 
 
 %We need the notion of a \emph{degenerate} policy, which captures the idea of a treatment assignment with a point mass distribution that reveals only a fixed potential outcome. 
 
% \begin{definition}[Degenerate Policy]
%A \emph{degenerate policy} is a distribution $\psi$ on $\textbf{Z}$ such that all its mass is concentrated on a set $\Omega_0 \subset \Omega$ and there exists $i$ in $\Omega_0$, where $\psi(Z_i=z) = 1$. A policy that is not degenerate is called non-degenerate. Under the distribution $\psi$, there exists at lease one unit $i$ that always gets assigned to the treatment status $z$. 
% \end{definition}
 
\begin{proposition}
 	\label{prop:relation}
Consider the following two degenerate policies $\phi_0$ and $\phi_1$
\begin{align*}
\phi_1(\textbf{Z}=\textbf{z}) = 
\begin{cases}
1  &\text{ if }  \textbf{z} = \textbf{1} \\
0 &\text{ otherwise}
\end{cases}
\end{align*}
\begin{align*}
\phi_0(\textbf{Z} = \textbf{z}) = 
\begin{cases}
1  &\text{ if }  \textbf{z} = \textbf{0} \\
0 &\text{ otherwise}
\end{cases}
\end{align*}
Then $\theta(\phi_1,\phi_0) = TTE$.
\end{proposition}
 
Obtaining $DTE$ as a special case of $\theta$ is not possible. This is because there is no single policy (degenerate or non-degenerate) that allows us to estimate $\sum_i Y_i(1,0)$.  We need to define $n$ different degenerate policies, 
  $\psi_i(\textbf{Z})$ where 
  $$\psi_i(\textbf{Z}) = 1 \text{ if } \textbf{z} \in \Omega_i(1,0), 0 \text{ otherwise}$$ 
  then 
  $$DTE = \frac{1}{n}\sum_i\E_{\psi_i} \left[ Y_i(z_i,e_i)\right] - \E_{\psi_0}\left[Y_i(z_i,e_i)\right]$$
 A similar analysis can be done for the interference effects $\gamma_1$ and $\gamma_2$. 
 
\begin{definition}[Irrelevant or nuisance potential outcomes] Let a causal estimand be a function of $\{ Y_i(z^j,e^j)\}_{j=1}^J$. We will call such potential outcomes as \emph{relevant}. Potential outcomes that are not relevant are \emph{nuisance} or irrelevant potential outcomes. %Potential outcomes that do not appear in the definition of an estimand are called nuisance or irrelevant potential outcomes.
\end{definition}
	
For instance, consider the average causal estimands defined as contrasts between fixed potential outcomes, i.e. the direct treatment effect and the total treatment effect. If we make no structural assumptions, then each causal effect is a function of only two columns in the Table of science, and the other columns are irrelevant. For e.g. potential outcomes of the form $Y_i(1,e_i), e_i \neq 0$ are not relevant for estimating $DTE$, since $DTE$ is a function of only $Y_i(1,0)$ and $Y_i(0,0)$. In this sense, if the goal is to only estimate $DTE$, then designs that generate other potential outcomes are wasteful.
	
Similarly, consider the marginal causal effects defined as contrast between expected potential outcomes under different policies. All the potential outcomes that have positive probability under the policy are relevant. On the other hand, potential outcomes that have $0$ probability under the policy are nuisance. If support of the policy is $\Omega$, the causal effect is a function of the entire table of science and all potential outcomes are relevant.

As we will see in Section \ref{sec:Bias}, average causal effects like the DTE and TTE are more difficult to estimate, and the nuisance potential outcomes form a source of bias when using difference-in-means estimators. We need new designs along with Horvitz-Thompson estimators for unbiased estimation of average estimands. On the other hand, it is simpler to estimate marginal estimands as long as the actual randomization follows the policies of interest. One can use difference of means estimators to estimate marginal estimands.

\subsection{Designs, Estimators and Strategies}
\label{sec:designs}
We will now consider the problem of estimation of causal effects and the existence of estimators. Our focus will be on randomization based inference where the potential outcomes are considered to be fixed and the only source of randomness is due to the random assignment of treatment vector $\textbf{z}$. Under this setting, a random assignment $\textbf{Z}$ is sampled according to $p(\textbf{Z})$ and the units are assigned to the treatment $\textbf{Z}$.  The outcome observed for each unit $i$ is denoted by $Y_i^{obs}$.

Consider an interference and exposure model of potential outcomes specified by a graph $G$ and an exposure function $f$. Let $\theta$ be a generic causal estimand defined as a function of $\mathbb T(G,f)$. 
\begin{definition}[Design]
	A design $p(\textbf{Z}|G, f)$ is a probability distribution supported over $\Omega$, the set of all possible treatment assignments for $n$ units.
\end{definition} 
In general, the design may depend on the interference graph $G$ and the exposure function $f$. We will suppress the dependence on $G$ and $f$ for simplicity. Designs that do not depend on the interference graph $G$ are called \emph{network-oblivious} designs. Such designs are preferred in the case when $G$ is unknown, however, these designs may not be optimal. 

Given a realization $\textbf{z}$ of a design $p(\textbf{Z})$, let $Y^{obs}(\mathbb{T}, \textbf{z})$ be the set of observed potential outcomes, where $\mathbb{T}$ is the Table of Science. 
\begin{definition}[Estimator]
	 An estimator $\hat \theta$ for $\theta$ is a function of the observed potential outcomes $Y^{obs}(\mathbb T, \textbf{z})$.
\end{definition}
We are now ready to define an \emph{estimation strategy}. An estimation strategy for estimating a causal effect is a combination of a design and an estimator to be used with that design.

\begin{definition}[Estimation Strategy]
	An estimation strategy or simply, a strategy for estimating $\theta$ is a pair $(p(\textbf{Z}), \hat \theta)$.
\end{definition}

Estimation strategies are evaluated based on their properties such as unbiasedness and variance. 
An estimation strategy is said to be unbiased for estimating $\theta$ if
$$
\mathbb E_{p(\textbf{Z})}[\hat \theta] = \theta.
$$
An estimator is said to be unbiased for $\theta$ if it is unbiased for any design $p(\textbf{Z})$.
For a given $\theta$, one goal is to construct strategies that are so called uniformly minimum variance unbiased (UMVU), see for e.g. \cite{sarndal2003model}. It is well known that such strategies don't exist, even with the no-interference assumption. We will focus only on the unbiasedness properties of an estimation strategy.

\subsection{Existence of estimators}
\label{sec:existenceofEstimators}
 We will now consider the assumption to ensure causal estimates are identified. We first start with the classic SUTVA assumption and present it's counterpart when there is treatment interference. When there is no interference, Stable Unit Treatment Value Assumption (SUTVA) has two parts: 
\begin{enumerate}
	\item No Interference: The potential outcome of unit $i$ depends only on the treatment of unit $i$.
	\item Consistency or Stability: There are no hidden versions of the treatment.
\end{enumerate}

The first part of SUTVA says that the potential outcome of a unit $i$ depends only on its treatment assignment. The second part says that there are no hidden versions of the treatment. Put in a different way, the stability assumption states that there is only one column corresponding to a treatment in the table of science. Under interference, we consider both these parts of SUTVA separately. The first part of SUTVA is relaxed to neighborhood interference by considering models of potential outcomes, and the second part of SUTVA is modified to the ``no hidden versions of treatment and exposure assumption''. 

%Our goal is to relax the first assumption of no-interference to some form of treatment interference by specifying models of potential outcomes. But the second assumption of stability is needed. The second assumption entails that there is only one column corresponding to a treatment in the table of science. i.e. given a treatment $\textbf{z}$, each unit can have only one potential outcome $Y(\textbf{z})$ 	

\paragraph{Neighborhood Interference}
The neighborhood interference assumption states that the potential outcome of unit $i$ depends on its treatment and the treatment status of it's interference neighborhood $N_i$ as specified by the exposure models, i.e.
$$Y_i(\textbf{z}) = Y_i(z_i,\textbf{z}_{N_i})  = Y_i(z_i,e_i)$$
where $e_i = f(\textbf{z}_{N_i})$ and $f$ is the exposure model.

\paragraph{Consistency or Stability} 
The consistency assumption states that the observed outcome of a unit $i$ is exactly equal to the unit's potential outcome under the assigned treatment and exposure combination, that is, there are no hidden versions of the treatment and exposure combination. 

\begin{align}
Y_i^{obs} = \sum_i{Y_i(z_i,e_i)I(Z_i = z_i, E_i=e_i)}
\end{align}

\paragraph{Unconfoundedness}
Unconfoundedness assumption states that the treatment assignment mechanism $p(\textbf{Z})$ does not depend on the potential outcomes $Y_i(z_i,e_i)$.

%The unconfoundedness assumption is implicit in the notation we have used for treatment assignment mechanism, i.e we have assumed that $p(\textbf{Z}|\mathbb{T}) = p(\textbf{Z})$ where $\mathbb{T}$ is the table of science given in Table \ref{tab:tableofScience}.

\paragraph{Positivity}
To state the positivity assumption, we need some preliminary definitions.
Recall that $e_i = f(\textbf{z}_{N_i})$. Let $z_0$ and $e_0$ denote a generic treatment and exposure condition. Recall that 
$$
\Omega_i(z_0,e_0) = \{\textbf{z} \in \Omega: z_i = z_0, f(\textbf{z}_{N_i}) = e_0 \} 
$$
is the set of all treatment vectors that give rise to $z_0$ and $e_0$ for unit $i$.
\begin{definition}[Propensity Scores]
	The propensity score for each unit $i$ for treatment and exposure pair $(z_0,e_0)$ denoted by $\pi_i(z_0,e_0)$ is defined as follows: 
	\begin{align}
	\pi_i(z_0,e_0) = \underset{\textbf{z} \in \Omega_i(z_0,e_0)}{\sum} \pr{\textbf{Z} = \textbf{z}}
	\end{align}
\end{definition}
Let the causal estimand be a function of $\{ Y_i(z^j,e^j)\}_{j=1}^J$. We will call such potential outcomes \emph{relevant}. Potential outcomes that are not relevant are \emph{nuisance}. Then we need
\begin{align}
0 < \pi_i(z = z^j,e = e^j) < 1 \forall i, j
\end{align}
The \emph{positivity} assumption implies that there is a positive probability of observing the relevant potential outcomes for each unit.
\begin{example}
	For example, $DTE$ is a function of $Y_i(0,0)$ and $Y_i(1,0)$. Hence the positivity condition requires that
	\begin{align}
	0 < \pi_i(z = 1,e = 0), \pi_i(z=0,e=0) < 1 \forall i
	\end{align}
	
\end{example}

If the positivity condition is not satisfied, the relevant potential outcomes are not observable and causal inference is not possible. In particular, 
\begin{theorem}
	\label{thm:existence}
	Let $p(\textbf{Z})$ be any design and let $\{ Y_i(z^j,e^j)\}_{j=1}^J$ be the set of relevant potential outcomes for any estimand $\theta$. Without any structural assumptions on the Potential Outcomes, unbiased estimators of $\theta$ under a design $p(\textbf{Z})$ exist iff 
	$$0 < \pi_i(z = z^j, e = e^j) < 1 \forall i=1, \ldots,n \text{ and } j = 1, \ldots, J$$ 
\end{theorem}

\begin{example}
	The positivity assumption depends not only on the design, but also on the interference graph and the exposure model. Consider an interference graph $G$ with a linear exposure model. If there is a unit $j$ with degree $n-1$, then either $\pi_j(0,0) = 0$ or $\pi_i(1,0)=0$ for all $i$. This is because the only way $Y_j(0,0)$ can be observed if all units are assigned to control, i.e. $\textbf{Z} = \textbf{0}$. However, under this assignment, $Y_i(1,0)$ is unobservable for any $i$. That is when degree is $n-1$, it is impossible to observe both $Y_i(1,0)$ and $Y_i(0,0)$. A solution is to allow for biased estimators or to exclude nodes with degree $n-1$ from the definition of causal effect. Note that in practice, such networks may be rare.
\end{example}

\subsection{Commonly used designs and estimators}
\label{sec:designandestimators}
In this section we will consider some classical designs and estimators that are used for estimating causal effects.
We will start by considering three classic designs used for estimating causal effects; these designs are commonly used when there is no treatment interference. Since these designs no not depend on the interference graph $G$ and the exposure function, they are network-oblivious. We will examine the applicability of these designs in estimating causal effects when there is interference.

\paragraph{Completely Randomized Design}
In a Completely Randomized Design (CRD), the treatment is assigned by fixing the number of treated and control units to $n_t$ and $n_c$ such that $n_t + n_c = n$, and choosing $n_t$ random units without replacement to be assigned to treatment, the rest of the units are assigned to control. The probability distribution $p(\textbf{Z})$ is hyper-geometric:
\begin{align*}
p(\textbf{Z} = \textbf{z})
\begin{cases}
&= \frac{1}{\binom{n}{n_t}}  \text{ if } |\textbf{z}| = n_t \\
& = 0 \text{ o.w}
\end{cases}
\end{align*}
Note that the entries of $\textbf{Z}$ are correlated.

\paragraph{Bernoulli Randomization}
In a Bernoulli trial, each unit is assigned to treatment independently with probability $p$. The total number of treated and control units are random. The probability distribution of a Bernoulli trial is:
\begin{align*}
p(\textbf{Z} = \textbf{z}) &= p^{|\textbf{z}|}(1-p)^{n-|\textbf{z}|}
\end{align*}
In the Bernoulli trial, there is a positive probability that all units get assigned to either the treatment or control, violating the positivity assumption. A simple way to avoid this is to consider a \emph{restricted} Bernoulli trial. Under a restricted Bernoulli trial, the number of treated unit and control units is always at least $1$. The probability distribution of a restricted Bernoulli trial is:
\begin{align*}
p(\textbf{Z} = \textbf{z}) &= \begin{cases}
\frac{p^{|\textbf{z}|}(1-p)^{n-|\textbf{z}|}}{1 - p^n - (1-p)^n} \text{ if } 0 < |z| < n \\
0 & \text{ otherwise }
\end{cases} 
\end{align*}

\paragraph{Cluster Randomization}
In a cluster randomized design, units are grouped to form clusters and these clusters are randomly assigned to the treatment or control condition, i.e. the randomization happens at the cluster level. The treatment status of a unit is equal to the treatment assigned to it's cluster. 

Formally, let the $n$ units be partitioned into $k=1, \ldots, K$ clusters. Let $n_1, \ldots, n_K$ be the number of units in each cluster where $n = \sum_{i=1}^K n_k$. Note that the $n_k$'s are fixed. Let $z_{k}$ denote the treatment assignment of cluster $k$ and let $c_i$ denote the cluster that unit $i$ belongs to. Thus we have $z_i = z_{c_i}$. The random assignment of clusters to treatment or control is done by a completely randomized design. Let $K_t$ and $K_c$  denote the number of treated and control clusters respectively. Let $n_t$ be the total number of treated nodes and $n_c$ be the total number of control nodes. Note that $n_t =\sum_k n_k z_k$ and $n_c = \sum_k n_k (1-z_k)$ and hence $n_t$ and $n_c$ are random. 

\begin{remark}
	One can consider a Bernoulli assignment of clusters to treatment and control. The Bernoulli assignment is not preferred as one has no control over the number of clusters assigned to treatment or control.
\end{remark}

Next, we will discuss two classes of estimators.
\paragraph{Difference-in-Means Estimators}
The difference-in-means estimators are simplest estimators. These estimators are defined as difference of means between two types of observed potential outcomes. We consider the following three difference-in-means estimators, the first one of which is the classic difference-in-means:
\begin{subequations}
	\label{eq:dom}
	\begin{align}
	\hat{\beta}_{naive} &= \frac{\sum_iY_i^{obs}Z_i}{\sum_i{Z_i}}- \frac{\sum_i Y_i^{obs}(1-Z_i)}{\sum_{i}{(1-Z_i)}} \label{diffofMeans}\\
	\hat{\beta_1} &= \frac{\sum_i Y_i^{obs} I(Z_i=1,E_i=0)}{\sum_i I(Z_i=1,E_i=0)} - \frac{\sum_i Y_i^{obs} I(Z_i=0,E_i=0)}{\sum_i I(Z_i=0,E_i=0)} \label{diffofMeans1}\\
	\hat{\beta_2} &= \frac{\sum_i Y_i^{obs} I(Z_i=1,E_i=1)}{\sum_i I(Z_i=1,E_i=1)} - \frac{\sum_i Y_i^{obs} I(Z_i=0,E_i=0)}{\sum_i I(Z_i=0,E_i=0)}  \label{diffofMeans2}
	% \hat{\gamma_1} &= \frac{\sum_i Y_i^{obs} I(Z_i=0,E_i=1)}{\sum_i I(Z_i=0,E_i=1)} - \frac{\sum_i
	% Y_i^{obs} I(Z_i=1,E_i=0)}{\sum_i I(Z_i=1,E_i=0)} \label{diffofMeans3}\\
	%\hat{\gamma_2} &= \frac{\sum_i Y_i^{obs} I(Z_i=1,E_i=1)}{\sum_i I(Z_i=1,E_i=1)} - \frac{\sum_i Y_i^{obs} I(Z_i=1,E_i=0)}{\sum_i I(Z_i=1,E_i=0)}  \label{diffofMeans4}
	%%\hat{\gamma_2} &= \frac{\sum_i Y_i^{obs} I(1,1)}{\sum_i I(1,1)} - \frac{\sum_i Y_i^{obs} I(1,0)}{\sum_i I(1,0)}  \label{diffofMeans4}
	\end{align} 
\end{subequations}

\paragraph{Linear Estimators}
One can also consider a larger class of estimators that are linear combination of the observed potential outcomes:
\begin{align*}
\hat{\theta} = \sum_{i=1}^n w_i(\textbf{z})Y_i^{obs}
\end{align*}
This class includes the difference-in-means estimators.
We will study the difference-in-means estimators in Section \ref{sec:Bias} and the linear weighted estimators in Section \ref{sec:lue}.

\section{Analytical insights for Difference-in-Means Estimators}
\label{sec:Bias}
In this section we study various estimation strategies that use a combination of difference-in-means estimators and classical designs for estimating causal effects when there is interference. We will focus on the nature and source of bias, if any, for estimating average and marginal causal effects.
The nature and source of bias depends on the estimand, the estimation strategy (i.e. the design and the estimator) and the model for potential outcomes.

In Section \ref{sec:dimMarginal} we show that the difference-in-means estimator is unbiased for estimating marginal effects.  This section is devoted to understanding the nature of bias when using difference-means estimators to estimate  the direct treatment effect (DTE). We will consider different models of potential outcomes  when estimating the direct effect using the difference of means estimators. The role of these models is to gain some analytical insights into the nature of bias and it's dependence on the modeling assumptions.

\subsection{Sources of bias in estimating the direct and the total effect}
\label{sec:biasSources}
%In this section we will consider the bias in estimating the non-degenerate policy based estimands using the difference of means estimator. We will mostly focus on estimating the direct effect $ATE_1$ or $\beta_1$ using the difference of means estimator. We will consider different models of potential outcomes to understand the nature of bias. The role of these models is to try and gain some analytical insights into the type of network statistics that are relevant to bias and variance of the estimators we consider in the design we consider. Having multiple models will help understand which network statistics are a consequence of specific model assumptions (eg, linearity) as opposed to suggestive of more general principles.

 Without making any \emph{structural} assumptions on the potential outcomes, (i.e. without any assumptions on how one potential outcome is related to another) the difference-in-means estimators given in equation \ref{eq:dom} have two different sources of bias for estimating the total effect and the direct effect:
 \begin{enumerate}
 	\item The first source of bias is due to unequal weights given to the potential outcomes, or equivalently, unequal probability of including in the sample for estimating the mean potential outcomes. For some designs, this source of bias can be eliminated.
 	\item The second source of bias is due to the inclusion of \emph{irrelevant} potential outcomes, i.e. potential outcomes other than those used in the definition of the corresponding causal effect. For example, when estimating DTE, these would be any other potential outcomes other than $Y_i(0,0)$ and $Y_i(1,0)$. 
 \end{enumerate}
 
 Proposition \ref{prop:biasGeneralCase} characterizes the bias of the naive estimator in estimating the direct treatment effect. For the bias of estimating $TTE$ using the naive estimator, see Proposition \ref{prop:biasGeneralCasebeta2} in Appendix. 
 
\begin{proposition}
	\label{prop:biasGeneralCase}
	Consider the parametrized potential outcomes given in equation \ref{eq:NIparam}:
	\begin{align*}
	Y_i(z_i,e_i) = A_i(z_i) + B_i(e_i) + z_iC_i(e_i)
	\end{align*}
	The direct effect $DTE$ is given by
	$$
	DTE = \frac{1}{n}\sum_{i=1}^n (A_i(1) - A_i(0)).
	$$
	For any design $p(\textbf{Z})$, the bias of the naive estimator $\hat{\beta}_{naive}$  (equation \ref{diffofMeans}) for estimating $DTE$ is:
	\begin{align*}
		b_1 &= E[\hat{\beta}_{naive}] - DTE \\
		&=\sum_i 
		\left(A_i(1)\left(\alpha_i(1) -\frac{1}{n}\right) - A_i(0)\left(\alpha_i(0) + \frac{1}{n}\right)\right)\\
		&+ \sum_i\sum_{e_i\neq 0} 
		B_i(e_i) \left(\alpha_i(1,e_i) - \alpha_i(0,e_i)\right) \\
		&+ \sum_i\sum_{e_i\neq 0} C_i(e_i)\alpha_i(1,e_i)
	\end{align*}
	where,
	\begin{align*}
	\alpha_i(z_i,e_i) = E\left[ \frac{I(Z_i=z_i,E_i=e_i)}{\sum_i{I(Z_i=z_i)}} \right] \text{ and } \alpha_i(z_i) = \sum_{e_i}{\alpha_i(z_i,e_i)} = E\left[ \frac{I(Z_i = z_i)}{\sum_i I(Z_i=z_i)} \right].
	\end{align*}
%	The bias of the difference in means estimator is:
%	\begin{align*}
%	b &= E[\hat{\beta}] - \beta_1 \\
%	&=\sum_i 
%	\left(Y_i(1,0)E\left[\alpha_i(1,0) -\frac{1}{n}\right] - Y_i(0,0)E\left[\alpha_i(0,0) + \frac{1}{n}\right]\right)\\
%	&+ \sum_i\sum_{e_i\neq 0} 
%	\left( Y_i(1,e_i)E[\alpha_i(1,e_i)] - Y_i(0,e_i)E[\alpha_i(0,e_i)]\right)
%	\end{align*}
\end{proposition}

Next Proposition \ref{prop:expDoMs} shows that the difference of means estimators $\hat \beta_1$ and $\hat \beta_2$  are also biased, but the source of bias is milder when compared to the naive estimator.
\begin{proposition}
	\label{prop:expDoMs}
	\begin{align*}
	E[\hat{\beta_1}] &= \sum_{i=1}^n A_i(1) \beta_i(1,0) - A_i(0)\beta_i(0,0) \\
	E[\hat{\beta_2}] &= \sum_{i=1}^n A_i(1) \beta_i(1,1) - A_i(0)\beta_i(0,0) + \sum_{i=1}^n B_i(1)\beta_i(1,1) + \sum_{i=1}^n C_i(1) \beta_i(1,1)
	\end{align*}
	where,
	\begin{align*}
	\beta_i(z_i,e_i) = E\left[ \frac{I(Z_i=z_i,E_i=e_i)}{\sum_i{I(Z_i=z_i,E_i=e_i)}} \right].
	\end{align*}	
\end{proposition}

Propositions \ref{prop:biasGeneralCase}, \ref{prop:expDoMs} and \ref{prop:biasGeneralCasebeta2} suggests that for any design $p(\textbf{Z})$, the biases in estimating the direct and the total treatment effects using difference-in-means estimators are controlled by the \emph{weighted exposure probabilities} $\alpha_i(z_i,e_i)$ under that design.  For instance, as Proposition \ref{prop:biasGeneralCase} shows, without making any structural assumptions on the potential outcomes $Y_i(z_i,e_i)$, the bias of $\hat{\beta}$ in estimating $DTE$ is $0$ only if $\alpha_i(1) = \alpha_i(0) = \frac{1}{n}$
and $\alpha_i(1,e_i) = \alpha_i(1,e_i) =0 \forall e_i \neq 0$. The first condition removes the first source of bias by placing equal weights on the relevant potential outcomes, and the second condition removes the second source of bias by placing $0$ weights on irrelevant potential outcomes. Similarly, as seen by Proposition \ref{prop:expDoMs}, the estimators $\hat{\beta_1}$ and $\hat{\beta_2}$ remove the second source of bias by eliminating irrelevant potential outcomes, but the first source of bias remains.

%{\color{red}The next two paragraphs are out of place! They need to be moved somewhere else or eliminated.
%	
%Given these sources of bias, unbiased estimation can be performed by using three strategies: (1) Construct new design unbiased estimators or (2) Construct new designs (3) Impose structural assumptions on the potential outcomes. All three strategies can be combined as well. 
%
%In Section \ref{sec:lue}, we explore linear unbiased and approximately unbiased estimators with and without structural assumptions on the potential outcomes for any design. In Section \ref{sec:newdesigns} we propose new designs and estimators along with the designs that eliminate the bias .} %The bias due to the first source can be reduced by choosing designs such that $\alpha_i(1)$ and $\alpha_i(0)$ are close to $\frac{1}{n}$. The bias due to the second term is $0$ if the weights on the irrelevant potential outcomes is $0$.

%In this section, we explore the nature of bias by imposing structural assumptions on the potential outcomes. 

\subsection{Characterization of bias under various models of Potential Outcomes}
\label{sec:biasPOModels}
The bias of the difference-in-means estimators depends on the weights $\alpha_i(z_j,e_j)$. These weights depend on the design and the  exposure model. To gain additional insight into the nature of the bias, we will make several modeling  assumptions on the potential outcomes. These assumptions allow us to computing the exposure weights analytically for commonly used designs. We will focus on the bias of estimating the direct effect using the naive estimator $\beta_{naive}$. We also ask the related question:  Does the bias in the difference-in-means estimator disappear if we make structural assumptions on the Potential Outcomes? 

\subsubsection{Symmetric Exposure Model}
%\begin{remark} 
%It is easy to see that the bias has two terms, the first term involves the \emph{correct} potential outcomes $Y_i(1,0)$ and $Y_i(0,0)$, but the weights $\alpha_i(1,0)$ and $\alpha_i(0,0)$ are incorrect under CRD and the Bernoulli designs. The second term involves \emph{irrelevant} potential outcomes $Y_i(z_i,e_i), e_i \neq 0$ with non-zero weights under CRD and Bernoulli designs.
%\end{remark}
We begin by considering the Symmetric exposure model given in equation \ref{eq:symmetricExposure} and computing the exposure weights for CRD and Bernoulli designs.

\begin{theorem}[Exposure Weights for Symmetric Exposure]
	\label{prop:exposureweights}
	Consider the symmetric exposure model of potential outcomes given in equation \ref{eq:symmetricExposure}. Under a CRD and a Bernoulli design, we have $\alpha_i(1)  = \alpha_i(0) = \frac{1}{n}$. On the other hand, under a cluster randomized design, $\alpha_i(1) \neq \alpha_i(0) \neq \frac{1}{n}$.
	%Under a CRD design, $\alpha_i(z_i,e_i) = \frac{1}{n}\pr{E_i = e_i|Z_i=z_i}$. 
	For a CRD design,
	\begin{align*}
	\alpha_i(1,e_i) & = \frac{1}{n} \frac{\binom{n_t-1}{e_i} \binom{n_c}{d_i-e_i}}{\binom{n-1}{d_i}} \text{ if } n_t \geq e_i +1 \text{ and } n_c \geq d_i-e_i, 0 \text{ otherwise} \\
	\alpha_i(0,e_i) & = \frac{1}{n} \frac{\binom{n_t}{e_i} \binom{n_c-1}{d_i-e_i}}{\binom{n-1}{d_i}} \text{ if } n_t \geq e-i \text{ and } n_c \geq d_i -e_i+1, 0 \text{ otherwise}
	\end{align*}		
	For a Bernoulli design, let $K$ be a restricted binomial random variable with support on $\{1,\ldots, n-1\}$ and $\pr{K=k} = \frac{\binom{n}{k}p^k(1-p)^{n-k}}{1 - (1-p)^n - p^n}$. Then,
	\begin{align*}
	\alpha_i(1,e_i) & = \frac{1}{n}\mathbb{E}_K \left[ 
	\frac{\binom{K-1}{e_i} \binom{n-K}{d_i-e_i} }{\binom{n-1}{d_i}} \right]\\
	\alpha_i(0,e_i) & = \frac{1}{n}\mathbb{E}_K \left[ 
	\frac{\binom{K}{e_i} \binom{n-K-1}{d_i-e_i} }{\binom{n-1}{d_i}} \right]
	\end{align*}	
\end{theorem}
Theorem \ref{prop:exposureweights} shows that the first source of bias gets eliminated under the CRD and the Bernoulli designs. Does the second source of bias disappear under these designs? Without any further assumptions, the answer is no. However, under additional assumptions, the second source of bias can go to $0$ asymptotically, or even be made exactly $0$. Examination of the second source of bias requires computing the weighted exposure probabilities $\alpha_i(z_i,e_i)$ under the CRD, Bernoulli designs and Cluster Randomized designs, which depend on the exposure model. Computing $\alpha_i(z_i,e_i)$ under the Bernoulli and cluster randomized designs is further complicated by the fact that, unlike the CRD, the denominator is a random variable that is correlated with numerator. Moreover, due to the overlapping neighborhoods, the correlation depends on a complicated manner on the graph $G$. Similar issues prevent us from obtaining explicit formula for $\beta_i(z_i,e_i)$. However, progress can be made by computing the bias directly under some structural assumptions. %This is the goal of the remainder of this section. {\color{red} Write this in a better way!}

\subsubsection{Additive Symmetric Exposure model} 
Let us consider the additive symmetric model given by equation \ref{eq:additiveLinear} below which is obtained by making the structural assumption $C_i(e_i) = 0$ in equation \ref{eq:symmetricExposure}. 
\begin{align}
\label{eq:additiveLinear}
Y_i(z_i,e_i) = \alpha_i + \beta_i z_i + B_i(e_i)
\end{align}
where $e_i \in \{0,1,\ldots, d_i\}$.
\begin{corollary}
	\label{prop:BiasNaiveAdditiveModel}
	Let $Y_i(z_i,e_i) = \alpha_i + \beta_i z_i + B_i(e_i)$, $e_i \in \{0,1,\ldots, d_i\}$. The bias in estimating $DTE$ using the difference-in-means estimator $\hat \beta_{naive}$  under the CRD and Bernoulli designs is 
	$$\sum_i \sum_{e_i \neq 0} B_i(e_i)\left[\alpha_i(1,e_i) - \alpha_i(0,e_i) \right]$$
	%where, under the CRD Design,
	%\begin{align*}
	%w_i(e_i)
	%&=\frac{1}{n}\left(\pr{E_i=e_i|Z_i=1} - \pr{E_i=e_i|Z_i=0}\right)
	%\end{align*}

\end{corollary}

%{\color{red} Do simulations to see what are $\alpha_i$'s and if $\alpha_i(1,e_i) - \alpha_i(0,e_i)$ is less than $0$.}

\subsubsection{Symmetric Additive Linear exposure}

Consider the symmetric additive linear model of the potential outcomes model specified by equation \ref{eq:LinearAddSymmExposure} and further assume constant interference effects, i.e. $\gamma_i = \gamma$. This gives us the following linear model of Potential outcomes:
\begin{align*}
Y_i(z_i,e_i)
&= \alpha_i + \beta_i z_i + \gamma  \left(\sum_j g_{ij}z_j \right)
\end{align*}

\begin{proposition}
\label{prop:biasSimpleLinearModelCRD}
Under a completely randomized design, we have,
	$$E[\hat{\beta}]  - DTE = - \gamma \frac{2m}{n(n-1)}.$$
where $m = \frac{1}{2}\sum_i{d_i}$ is the number of edges in the interference graph. 
\end{proposition}

\paragraph{Bernoulli Randomization}
\begin{proposition}
\label{prop:biasSimpleLinearModelBernoulli}
 Under a restricted Bernoulli trial, we have,
	$$E[\hat{\beta}]  - DTE = -\gamma \frac{2m}{n(n-1)}$$
\end{proposition}

\paragraph{Remark} Propositions \ref{prop:biasSimpleLinearModelCRD} and \ref{prop:biasSimpleLinearModelBernoulli} show that even under the structural assumptions of additivity, linearity and constant interference effect, there is always a bias due to the interference. The bias is independent of $n_c$ and $n_t$ in the CRD and $p$ in the Bernoulli designs. The bias is in the opposite direction of the interference effect, i.e. a positive interference leads to smaller estimate of the average treatment effect when compared to the true $\beta_1$. The bias scales as $O\left(\frac{m}{n^2}\right)$, hence for sparse and large network, asymptotically, the bias goes to 0.  Is it possible for the bias to be exactly $0$? The answer is yes, and further explained in the next section.

\subsubsection{Binary Exposure Model}
In this section we consider the binary exposure model given in equation \ref{eq:BinaryExposure} and study the bias of $\beta_{naive}$ for estimating DTE.

\paragraph{Completely Randomized Design}
\begin{proposition}
\label{prop:Bias2by2ModelCRD}
		Under model \ref{eq:BinaryExposure}, we have, for a CRD
	\begin{align*}
	\mathbb E\left[\hat{\beta}\right] - DTE &= -\frac{1}{n}\sum_i  \gamma_i  \frac{\binom{n_c-1}{d_i-1}}{\binom{n-1}{d_i}} + \frac{1}{n}\sum_i \theta_i \left(1 - \frac{\binom{n_c}{d_i}}{\binom{n-1}{d_i}}\right)
	\end{align*}
\end{proposition}
\paragraph{Bernoulli Randomization}
\begin{proposition}
	\label{prop:Bias2by2ModelBernoulli}
For a Bernoulli trial, we have
	\begin{align*}
	\mathbb E \left[\hat{\beta}\right] - DTE = - \sum_i\left( \frac{d_i \gamma_i (1-p)^{d_i}}{n(n-d_i)}\right)  + \sum_i \theta_i \left[ \frac{1}{n} - \frac{(1-p)^{d_i}}{n}\right]
	\end{align*}
\end{proposition}

Under the structural assumption of additivity we have $\theta_i=0$. In this case, by Proposition \ref{prop:Bias2by2ModelCRD} it follows that the bias of the difference-in-means estimator can be $0$ when $n_c < \min_i d_i$. Thus, we have the following result:
\begin{proposition}
	Consider the binary additive exposure model \ref{eq:addBinaryExposure}. Under the completely randomized design, if $n_c < \min_i d_i$, then the bias of the difference-in-means estimator $\hat \beta_{naive}$ in estimating $DTE$ is $0$.
\end{proposition}

\section{Linear Unbiased Estimators}
\label{sec:lue}
For any design $p(\textbf{Z}=z)$, we can construct unbiased estimators of causal effects by using standard techniques from the survey sampling literature. Let us consider a generic average causal estimand $\theta$ defined as a contrast between two different treatment and exposure combinations: $\tau_1 = (z_1,e_1)$ and $\tau_2 = (z_0,e_0)$:
\begin{align}
\theta = \frac{1}{n}\sum_{i=1}^n \left(Y_i(z_1,e_1) - Y_i(z_0,e_0)\right)
\end{align}
For example, $\theta = \beta_{DTE}$ if $(z_1,e_1) = (1,0)$ and $(z_0,e_0) = (0,0)$, and so on. Following \cite{godambe1955unified}, let us consider the most general class of linear weighted estimators for estimating $\theta$, i.e
\begin{align}
\label{eq:linearestimator}
\hat{\theta} = \sum_i w_i(\textbf{z})Y_i^{obs}
\end{align}
Here $w_i(\textbf{z})$ is the weight assigned to unit $i$.  Note that the weight assigned to unit $i$ depends on the treatment assigned to all the units in the finite population, i.e it depends on $\textbf{z}$.  The set of weights $w_i(\textbf{z})$ that lead to unbiased estimators of $\theta$ can be characterized as a solution to a system of equations that depend on the design, interference graph and the exposure model.

\begin{theorem}
	\label{thm:generalLinear}
	Consider an exposure model $e_i = f(\textbf{z}_{N_i})$ where $N_i$ is specified by an interference graph $G$. Assume that there are no structural assumptions on the potential outcomes. Let $\Omega_i(z_1,e_1) = \{\textbf{z}: z_i = z_1, e_i = e_1 \}$. Similarly, let $\Omega_i(z_0,e_0) = \{\textbf{z}: z_i = z_0, e_i = e_0 \}$.
	The estimator $\hat{\theta}$ in equation \ref{eq:linearestimator} is unbiased for $\theta = \frac{1}{n}\sum_i \left(Y_i(z_1,e_1) - Y_i(z_0,e_0)\right)$ if and only if 
	$0 < \pi_i(z_1,e_1)<1$ and $0<\pi_i(z_0,e_0) < 1$ and $w_i(\textbf{z})$ satisfy the following system of equations:
	\begin{align*}
	\underset{\textbf{z} \in \Omega_i(z_1,e_1) }{\sum} w_i(\textbf{z})p(\textbf{z})  &= \frac{1}{n}, \text{ } \forall i = 1, \ldots, n\\
	\underset{\textbf{z} \in \Omega_i(z_0,e_0) }{\sum} w_i(\textbf{z})p(\textbf{z})  &= -\frac{1}{n}, \text{ }\forall i = 1, \ldots, n\\
	\underset{\textbf{z} \in \Omega_i(z,e) }{\sum} w_i(\textbf{z})p(\textbf{z})  &= 0, \text{ }\forall (z,e) \neq (z_0,e_0) \text{ and } (z,e) \neq (z_1,e_1), i = 1,\ldots, n	
	\end{align*}
\end{theorem}
Recall that $\Omega_i(z,e)$ are the set of treatment allocations that reveal the potential outcome $Y_i(z,e)$ for unit $i$. Note that these sets depend on $i$ and are different for each unit. In general, there can be infinitely many solutions to the system of equations in Theorem \ref{thm:generalLinear} depending on the interference graph, exposure model, and the design. Hence there can be infinitely many unbiased estimators of $\theta$. For each unit $i$, let us consider the number of equations $p_i$ and the number of unknowns $m$. For each $i$, there are $m = |\Omega|$ unknown weights $w_i(\textbf{z}), \textbf{z} \in \Omega$ which depend on the support of the design. On the other hand, there are $p_i =  2\cdot K_i$ linearly independent equations in Theorem \ref{thm:generalLinear} which depend on the exposure model. Recall that $K_i$ is the number of levels of exposure for unit $i$. Hence linear weighted unbiased estimators don't exist if for each $i$, $m < p_i = 2K_i$. We have the following result:
\begin{proposition}
	Let $m = |\Omega|$ be the number of allocations and $K_i$ be the number of levels of the exposure model for each unit $i$. If for each $i$, $m > 2K_i$ and 	$0 < \pi_i(z_1,e_1), \pi_i(z_0,e_0)<1$, there are infinitely many linear unbiased estimators of $\theta$.
\end{proposition}
Table \ref{tab:equations} gives the values of $m$ and $p_i$ for some exposure models and designs. For instance, under a restricted Bernoulli design, there are $m = 2^n-2$ unknown weights, since  $\textbf{z} = 0$ and $\textbf{z} = 1$ is not allowed as it violates the positivity assumption. On the other hand, for a symmetric exposure model, there are $p_i = 2d_i$ equations where $d_i$ is the number of units in $N_i$. %If there is a unit with $n-1$ neighbors, the system has no solution and hence unbiased estimators dont exist. %A similar requirement holds for the symmetric exposure assumption - here $K = d_i$ and for each $i$, there are $2^{d_i+1}$ equations, and $2^n-2$ unknowns.  

\begin{table}[h]
	\centering
	\begin{tabular}{|c|c|c|c|}
		\hline  
									& Symmetric Exposure 	& Binary Exposure 	& General Exposure\\
		\hline
		\multirow{ 2}{*}{Bernoulli} & $m = 2^n-2$ 			&  $m = 2^n-2$		& $m = 2^n-2$\\ 
		  							& $p_i = 2d_i$ 		&  $p_i = 4$			& $p_i = 2^{d_i+1}$ \\ 
		  \hline
		 \multirow{ 2}{*}{CRD} 		& $m = \binom{n}{n_t}$  & $m = \binom{n}{n_t}$	& $m = \binom{n}{n_t}$  \\ 
		  							& $p_i = 2d_i$ 		& $p_i = 4$ 				& $p_i = 2^{d_i+1}$  \\ 
		\hline 
	\end{tabular} 
	\caption{Number of unknowns and equations for linear unbiased estimators}
	\label{tab:equations}
\end{table}

%Choosing an estimator from this class of estimators with minimum variance was studied by \cite{sussman2017elements}.

\subsection{Horvitz-Thompson Estimator}
\label{sec:HT}
If the weight of a unit $i$ is allowed to depend on $\textbf{z}$ only through $z_i$ and $e_i$, then we get a smaller class of linear estimators of the following form:
\begin{align}
\label{eq:linearestsmall}
\hat{\theta}_{2} = \sum_i{w_i(z_i,e_i) Y_i^{obs}}
\end{align}
The restriction on the weights is a form of sufficiency: instead of the weight depending on the entire vector $\textbf{z}$, it depends only on $(z_i,e_i)$. Since the potential outcomes are reduced from $Y_i(\textbf{z})$ to $Y_i(z_i,e_i)$, it is natural to consider such a reduction of the weights from $w_i(\textbf{z})$ to $w_i(z_i,e_i)$.

Theorem \ref{thm:HT} shows that under no structural assumptions on the potential outcomes, the only unbiased estimator of type $\hat \theta_2$ is the Horvitz-Thompson estimator. %In fact, as is well known in the survey sampling literature, the only unbiased estimator in the class $\hat{\theta}_2$ is the Horvitz-Thompson estimator.

%{\color{red} Need to mention this somewhere: The bias occurs due to the following fact: Due to interference, the units (and hence their potential outcomes) are no longer selected with equal probability. But they are samples with unequal probability. Hence the potential outcomes are also observed with unequal probabilities. This is what leads to bias. In the SUTVA case, and a crd, the samples are selected with equal probability, and the HT estimator reduces to the difference in means estimator. A simple way to obtain unbiased estimates with unequal probabilities is to use the HT estimator... But another way is to modify the design so that the units are again simple random samples. i.e we want $P(z_i=1,e_i=1)$ to be the same for each unit! }

\begin{theorem}
	\label{thm:HT}
	Consider the estimators of type $\hat \theta_2$ given by equation \ref{eq:linearestsmall}. Without any structural assumptions on the potential outcomes, the only unbiased estimator of $\theta$ in this class is the Horvitz-Thompson estimator $\hat \theta_{HT}$ where
	\begin{align*}
	w_i(z_i,e_i) = 
	\begin{cases}
	\frac{1}{n\pi_i(z_1,e_1)} \text{ if } (z_i,e_i) = (z_1,e_1) \\
	-\frac{1}{n\pi_i(z_0,e_0)} \text{ if } (z_i,e_i) = (z_0,e_0) \\
	0 \text{ otherwise }
	\end{cases}
	\end{align*}
\end{theorem}
The HT estimator eliminates both sources of bias mentioned in Section \ref{sec:Bias} by choosing the correct weights. In particular, the HT estimator assigns a weight of $0$ to nuisance potential outcomes, and a positive weight to relevant potential outcomes. The positive weight is inversely proportional to the probability of observing that potential outcome under the design $p(\textbf{Z})$. The HT estimator depends on the propensity scores $\pi_i(z_i,e_i)$. As mentioned before, these probabilities depend on the design and the exposure model. We compute an analytical formula of these probabilities for the CRD and the Bernoulli designs for different exposure models.

%\begin{theorem}
%	\label{thm:HT}
%	Let $(a,b)$ and $(c,d)$ be two different treatment regimes, such that not all $a,b,c,d$ are equal. Consider the causal effect $\theta$ defined as 
%	$$\theta([a,b], [c,d]) = \frac{\sum_i{Y_i(a,b) - Y_i(c,d) }}{n}$$
%	Consider the class of all linear weighted estimators
%	$$\hat \theta(w) = \sum_{i=1}^{n} w_i Y_i^{obs}$$
%	where the weights are fixed and allowed to depend on $z_i$ and $e_i$.
%	Under no additional restrictions on the potential outcomes, the only unbiased estimator in this class is the Horvitz-Thompson estimator, i.e. the estimator with the following weights:
%	\begin{align*}
%	w_i(a,b) &= \frac{1}{n\pr{Z_i=a,E_i=b}} \\
%	w_i(c,d) &= - \frac{1}{n\pr{Z_i=c,E_i=d} } \\
%	w_i(z_i,e_i) & = 0, \forall (z_i,e_i) \neq \{(a,b),(c,d)\}
%	\end{align*}
%	Hence the estimator is
%		\begin{align*}
%		\hat \theta_{HT} =\frac{1}{n} \left(
%		\sum_{i=1}^{n} Y_i(a,b)\frac{I(z_i=a,e_i=b)}{\pr{z_i=a,e_i=b}} - 
%		\sum_{i=1}^{n} Y_i(c,d)\frac{I(z_i=c,e_i=d)}{\pr{z_i=c,e_i=d}}
%		\right)
%		\end{align*}
%\end{theorem}

\begin{theorem}[Propensity Scores for Symmetric Exposure]
	\label{thm:exposureProbSymmetric}
	Consider the symmetric exposure function, $e_i = f(\textbf{Z}_{N_i}) = |\textbf{Z}_{N_i}|$, $e_i \in \{0,1,\ldots, d_i\}$.
	For a CRD Design,
	\begin{align*}
	\pr{Z_i=1,E_i=e_i} & = \frac{n_t}{n} \frac{\binom{n_t-1}{e_i} \binom{n_c}{d_i-e_i}}{\binom{n-1}{d_i}} \text{ if } n_t \geq e_i +1 \text{ and } n_c \geq d_i-e_i, 0 \text{ otherwise} \\
	\pr{Z_i=0,E_i = e_i} & = \frac{n_c}{n} \frac{\binom{n_t}{e_i} \binom{n_c-1}{d_i-e_i}}{\binom{n-1}{d_i}} \text{ if } n_t \geq e-i \text{ and } n_c \geq d_i -e_i+1, 0 \text{ otherwise}
	\end{align*}
	For a Bernoulli Design,
	\begin{align*}
	\pr{Z_i=1,E_i=e_i} & = \binom{d_i}{e_i}p^{e_i+1} (1-p)^{d_i-e_i}\\
	\pr{Z_i=0,E_i = e_i} & = \binom{d_i}{e_i}p^{e_i}(1-p)^{d_i-e_i+1}\\
	\end{align*}		
\end{theorem}	

\begin{theorem}[Propensity Scores for Binary Exposure]
	\label{thm:propensityscores}
	Consider the symmetric exposure function, $e_i = f(\textbf{Z}_{N_i}) = I(|\textbf{Z}_{N_i}|>1)$, $e_i \in \{0,1\}$.
	i.e a unit is exposed if at least 1 of its neighbor is treated.
	For a CRD,
	\begin{align*}
	\pr{Z_i = 1,E_i=1} &=  
	\begin{cases}
	0        																& \text{if } d_i = 0 \\
	\frac{n_t}{n}\left[1- \frac{\binom{n_c}{d_i}}{\binom{n-1}{d_i}}\right]	& \text{if } 0 < d_i \leq n_c\\
	\frac{n_t}{n}, 															&\text{if } d_i > n_c
	\end{cases} \\
	\pr{Z_i=1,E_i=0}  &= 
	\begin{cases}
	\frac{n_t}{n}       									& \text{if } d_i = 0 \\
	\frac{n_t}{n}\frac{\binom{n_c}{d_i}}{\binom{n-1}{d_i}}	& \text{if } 0 < d_i \leq n_c\\
	0, 														&\text{if } d_i > n_c
	\end{cases} \\
	\pr{Z_i=0,E_i=1}  &= 
	\begin{cases}
	0       														& \text{if } d_i = 0 \\
	\frac{n_c}{n}\left[1- \frac{\binom{n_c-1}{d_i}}{\binom{n-1}{d_i}}\right]	& \text{if } 0 < d_i \leq n_c -1\\
	\frac{n_c}{n}, 																			&\text{if } d_i > n_c -1
	\end{cases} \\		
	\pr{Z_i=0,E_i=0}  &=
	\begin{cases}
	\frac{n_c}{n}       													& \text{if } d_i = 0 \\
	\frac{n_c}{n}\frac{\binom{n_c-1}{d_i}}{\binom{n-1}{d_i}}	& \text{if } 0 < d_i \leq n_c -1\\
	0 												&\text{if } d_i > n_c -1
	\end{cases}
	\end{align*}
	Similarly, for a Bernoulli trial with probability of success $p$, we have
		\begin{align*}
		\pr{Z_i = 1,E_i = 1}  &= p(1-(1-p)^{d_i})\\
		\pr{Z_i = 1,E_i = 0}  &= p(1-p)^{d_i}	     		\\
		\pr{Z_i = 0,E_i = 1}  &= (1-p)(1-(1-p)^{d_i}) \\
		\pr{Z_i = 0,E_i = 0}  &= (1-p)^{d_i+1}		 
		\end{align*}
	Under a cluster randomized design, let $u_i$ be the number of clusters of unit $i$ and it's neighbors. Assume $n_k > 0, \forall k = 1, \ldots, K$.
		\begin{align*}
		\pr{z_i = 1,e_i = 1}  &= \frac{K_t}{K}\\
		\pr{z_i = 1,e_i = 0}  &= 0 \\ %\frac{K_t}{K} \prod_{i=1}^{u_i} \frac{K_c-i}{K-i+1} \mbox{ if }, u_i > 1, 0 \mbox{ if } u_i = 1	     		\\
		\pr{z_i = 0,e_i = 1}  &= 0 \mbox{ if } u_i = 1, \frac{K_c}{K}\left[1 - \prod_{i=1}^{u_i-1}\frac{K_c-u_i}{K-u_i}\right] \mbox{ if } u_i > 1\\
		\pr{z_i = 0,e_i = 0}  &= \frac{K_c}{K} \mbox{ if } u_i =1, \frac{K_c}{K}\left[\prod_{i=1}^{u_i-1}\frac{K_c-i}{K-i}\right] \mbox{ if } u_i > 1 \\
		&=\prod_{i=1}^{u_i} \frac{K_c-i+1}{K-i+1}		 
		\end{align*}
		
\end{theorem}

\begin{remark}
	The weights of the HT estimator depend on the exposure model and the interference graph $G$. In cases where the interference graph $G$ is not known, the HT estimator may not be usable. 
\end{remark}
\begin{remark}
The weights of the HT estimator depend only on the exposure model and the design. They do not depend on the structural model. For example, in the linear model of Potential Outcomes, the HT weights do not depend on the linearity of the model, but only on the exposure neighborhood.
\end{remark}

\begin{remark}
For a clustered randomized design, we cannot estimate $\beta_{DE}$ using HT estimators because some propensity scores are $0$.
\end{remark}

\subsection{Inadmissibility of the Horvitz-Thompson estimator}
\label{sec:inadmissible}
It is clear that the class of linear weighted unbiased estimators as given in equation \ref{eq:linearestimator} is quite large. Choosing a single estimator from this class is not possible. This is due to the fact that uniformly minimum variance estimators of $\theta$ don't exist. This can be shown by following a classical proof of Godambe who shows that uniformly minimum variance unbiased estimators of the finite sample population totals do not exist, see \cite{godambe1955unified}. On the other hand, if we restrict ourselves to a smaller class of linear unbiased estimators $\theta_2$ by requiring the weights to depend only on the treatment $z_i$ and the exposure $e_i$ of unit $i$, the HT estimator is the only unbiased estimator and hence is the minimum variance unbiased estimator. 

 %However, as we saw previously, the Horvitz-Thompson estimator is the only unbiased estimator in this smaller class. 

%At this point, we can consider several different options for choosing estimators. One strategy is to choose an estimator that minimizes a prior expected variance, see Sussaman and Airoldi. An alternate strategy is to consider the set of admissible estimators.

It is natural to ask if the HT estimator satisfies some optimality properties in a larger class of estimators. We study the admissibility of the HT estimator with respect to the mean squared error, in the class of all estimators for estimating a causal parameter $\theta$ under interference. The mean squared error of an estimator is defined as 
$$MSE (\hat \theta) = \mathbb E_{p(\textbf{Z})} [(\hat \theta - \theta)^2]
$$
\begin{definition}
	An estimator $\hat \theta_1$ is \emph{inadmissible} with respect to mean squared error if there exists an estimator $\hat \theta_2$ such that
	$MSE(\hat \theta_2) < MSE(\hat \theta_1)$  for all $\theta$.
\end{definition} For finite population inference, the admissibility of the HT estimator for estimating a finite population total in the class of all unbiased estimators is well known, see \cite{godambe1965admissibility}.

%{\color{red} Conjecture: The HT estimator is inadmissible for estimating $\beta$ when the potential outcomes are bounded}

We show that the Horvitz-Thompson estimator is inadmissible under the class of all estimators with respect to the mean squared error for a special class of designs called the \emph{non-constant} designs. A \emph{non-constant} design is a design where the number of units allocated to the treatment and exposure combinations of interest are random. 

\begin{definition}[Non-Constant Designs]
	\label{def:nonconstant}
	Consider a generic estimand $\theta$ given in equation \ref{eq:param} that is a contrast between treatment and exposure combinations $\tau_0 = (z_0,e_0)$ and $\tau_1 = (z_1,e_1)$.  Let $X_0 = \sum_{i=1}^n I(Z_i=z_0, E_i=e_0)$ and $X_1 = \sum_{i=1}^n I(Z_i=z_1, E_i=e_1)$. 
	A design $\mathbb P$ is a non-constant design for an estimand $\theta$ if $X_0$ and $X_1$ are random.
\end{definition}
\begin{theorem}[Inadmissibility of HT]
	\label{thm:HTinadmissible}
	Let $\mathbb P$ be any non-constant design as given in Definition \ref{def:nonconstant}. Consider the class of all estimators of $\theta$ with respect to the design $\mathbb P$. The Horvitz-Thompson estimator is inadmissible with respect to the mean squared error in this class. 
\end{theorem} 

It is can be verified that under interference, most commonly used designs such as Bernoulli design, CRD, and cluster randomized designs are non-constant. This is because the these designs control the treatment condition, but the exposure is indirectly assigned and hence the number of units under $\tau_0$ and $\tau_1$ are random. Thus, the consequence of this is that the H-T estimator is inadmissible for estimating average causal effects under interference for these designs. 

\subsection{Improving the Horvitz-Thompson Estimator}
\label{sec:improveHT}
The HT estimator is the only unbiased estimator in the class of linear unbiased estimators when the weights are not allowed to depend on the sample. However, the Horvitz-Thompson estimator is inadmissible with respect to the mean squared error in a larger class of estimators. In fact, the mean squared error of the HT estimator can be quite large, see Section \ref{sec:simulations}.  There are three general directions in which the HT estimator can be improved.
\begin{enumerate}
	\item \emph{Generalized Linear Estimators:} Allow the weights to depend on the sample and/or auxiliary information.
	\item \emph{Model dependent Unbiased Estimation:} Make structural assumptions on the Potential outcomes and seek unbiased estimators under the model assumptions.
	\item \emph{Model assisted  Estimation:} Make structural assumptions on the Potential Outcomes and seek model assisted HT estimators that are mildly biased.
\end{enumerate}
\subsubsection{Generalized Linear Estimators} The Horvitz-Thompson estimator can be improved by taking into account auxiliary information that depends on the labels of the sample. Such estimators fall into the class $\hat \theta_1$ that we have considered before and are called generalized linear estimators in the survey sampling literature, see for instance \cite{basu2011essay}. We present the so called \emph{generalized difference} estimator which can be an improvement to the HT estimator that is still unbiased. This estimator takes into account auxiliary information on the potential outcomes for each unit and in some cases can have smaller variance that the HT estimator.

Following \cite{basu2011essay}, let $a_1, \ldots, a_n$ and $b_1, \ldots, b_n$ be auxiliary information available for each unit $i$. Then the following difference estimator is an unbiased estimator of $\theta$:

\begin{align}
\hat \theta_{D}  = \frac{1}{N} \left( \sum_i (Y_i(z_1,e_1) -a_i) \frac{I_i(z_1,e_1)}{\pi_i(z_1,e_1)} - \sum_i (Y_i(z_0,e_0) -b_i) \frac{I_i(z_0,e_0)}{\pi_i(z_0,e_0)} + \sum_i \left(a_i -b_i\right) \right)
\end{align}
Here, $a_i$ and $b_i$ can be thought of as a priori information about the potential outcomes $Y_i(z_1,e_1)$ and $Y_i(z_0,e_0)$.  This estimator is a special case of the following generalized difference estimator. To define the generalized difference estimator, let
\begin{align*}
\hat{\bar{Y}}(z_1,e_1) &= \frac{1}{N} \left(\sum_i \frac{Y_i(z_1,e_1)I(Z_i=z_1,E_i=e_1)}{\pi_i(Z_i=z_1,E_1)} + \lambda_1\left(\sum_i \frac{a_iI(Z_i=z_1,E_i=e_1)}{\pi(Z_i=z_1,E_i=e_1)} - \sum_i a_i\right) \right)\\
\hat{\bar{Y}}(z_0,e_0) &= \frac{1}{N} \left(\sum_i \frac{Y_i(z_0,e_0)I(Z_i=z_0,E_i=e_0)}{\pi_i(Z_i=z_0,E_0)} + \lambda_2\left(\sum_i \frac{b_i I(Z_i=z_0,E_i=e_0)}{\pi(Z_i=z_0,E_i=e_0)} - \sum_i b_i\right) \right)
\end{align*}
where $\lambda_1$ and $\lambda_2$ are fixed numbers.
Then we have,
$$\hat \beta_{GD} = \hat{\bar{Y}}(z_1,e_1) - \hat{\bar{Y}}(z_0,e_0)$$
If we set $\lambda_1 = \lambda_2 = -1$, $\hat \theta_{GD} = \hat \theta_{D}$.

\subsubsection{Model dependent Unbiased Estimation}
In the model dependent unbiased estimation, one assumes a structural model for the potential outcomes and constructs unbiased estimators with respect to the model. For instance, consider the estimation of the direct treatment effect. Let us assume the additive model of potential outcomes with $C_i(e_i) = 0 \forall i$, as given in equation \ref{eq:additiveLinear}. Then any linear weighted estimator of the form
$$ \hat{\theta} = \sum_i w_i(z_i,e_i) Y^{obs}$$
is an unbiased estimator of the direct treatment effect, where the weights are given by the following system:
%Let $\pi_i(Z_i = z_i,E_i=e_i)$ denote the probability distribution of the design, where $Z_i$ and $E_i$ denote the random treatment and exposure assigned to unit $i$. Let $w_i(z_i,e_i)$ denote the weight assigned to unit $i$ under the fixed allocation $\{z_i,e_i\}$. 
For each unit $i$, consider the linear system,
\begin{align}
w_i(1,0)\pi_i(1,0) + \ldots +w_i(1,d_i)\pi_i(1,d_i) &= \frac{1}{n} \nonumber \\
w_i(0,e_i)\pi_i(0,e_i) + w_i(1,e_i)\pi_i(1,e_i) &= 0, e_i = \{0,\ldots, d_i\} 
\end{align}
The unbiasedness of the estimator depends on the structural assumptions of the potential outcomes, which are not known in general. Instead of using modeling assumptions, one may use the model to assist in the designing estimators.
\subsubsection{Model assisted Estimation}
Model assisted approach is very popular in the survey sampling literature, see for e.g. \cite{sarndal2003model}. Assume that for each unit $i$, there exists a auxiliary information $X_i$, and the goal is to estimate the finite population total $\sum_i Y_i$. In the model assisted estimation, a working model between $Y_i$ and $X_i$ is assumed. This may introduce a mild bias in the estimate as a trade off for a reduction in the variance.  Such model assisted estimators are called Generalized Regression Estimators or GREG estimators or calibration estimators in the survey literature.

The model assisted approach fits naturally in causal inference with interference. As in the case with survey sampling, model assisted estimators can be constructed by assuming a ``working'' model of potential outcomes. There are several natural models of potential outcomes that one can consider. Moreover, a natural auxiliary variable associated with each unit is the exposure level $e_i$. When considering causal estimands with interference, the model assisted approach offers a subtle advantage. It allows one to include units with \emph{nuisance} potential outcomes in the estimator - The model relates the nuisance potential outcomes to the relevant potential outcomes, thus allowing us to use both in the estimator. For instance, consider estimation of the direct treatment effect using the HT estimator. Without any underlying model, the only units that appear in the estimator are those whose  observed exposure is $(z_i=0,e_i=0)$ and $(z_i=1,e_i=0)$. Any unit with a different exposure does not appear in the estimator. However, the model assisted approach allows us to include information from the units whose observed outcome is a nuisance potential outcome, thereby increasing the effective sample size.

For example, consider the symmetric linear model for the potential outcomes:
\begin{align}
Y_i(z_i,e_i) = \alpha + \beta z_i + \gamma e_i + \delta z_i e_i
\end{align}
where $e_i \in \{0,1,\ldots,d_i\}$ denotes the number of treated units in the interference neighborhood of unit $i$. Let $\hat \alpha$, $\hat \beta$, $\hat \gamma$ and $\hat \delta$ be the weighted least squares estimates of $\alpha, \beta, \gamma,$ and $\delta$, where the weights for each unit $i$ is $w_i = \frac{1}{\pi(Z_i,E_i)}$. For any exposure $(z,e)$, let $\hat{Y}(z,e)$ be the estimated potential outcomes using the least squares fit. Let $\epsilon_i(z,e) = Y_i(z,e) - \hat{Y}_i(z,e)$.
The GREG estimator is defined as 
\begin{align}
\hat{\beta}_{greg} = \hat{\bar{Y}}_g(z_1,e_1) - \hat{\bar{Y}}_g(z_0,e_0)
\end{align}
where,
\begin{align*}
\hat{\bar{Y}}_g(z_1,e_1) &= \frac{1}{N}\sum_{i=1}^N \hat{Y}_i(z_1,e_1) + \frac{1}{N} \sum_i \frac{\hat{\epsilon}_i(z_1,e_1)I(Z_i=z_1,E_i=e_1)}{\pi_i(Z_i=z_1,E_1)} \\
\hat{\bar{Y}}_g(z_0,e_0) &= \frac{1}{N}\sum_{i=1}^N \hat{Y}_i(z_0,e_0) + \frac{1}{N} \sum_i \frac{\hat{\epsilon}_i(z_0,e_0)I(Z_i=z_0,E_i=e_0)}{\pi_i(Z_i=z_0,E_0)}
\end{align*}

%\subsection{Stratification by degree of the naive estimator}

%Help reduce / remove bias and reduce variance in the naive estimator for the ATE.

%%% %%% %%%
%%% %%% %%%
%%% %%% %%%

\section{New Designs for estimating ATE}
\label{sec:newdesigns}

The CRD and the Bernoulli designs are oblivious to the interference structure. There are at least two issues in using such designs for performing causal inference with interference. The first issue is that the experimenter has no control over which potential outcomes are revealed. The second issue is that the observed potential outcomes have an unequal probability of being revealed, which can lead to increased variance and bias in estimation.

In fact, a careful analysis reveals that these issues are two sides of the same coin - the exposure condition $e_i$ is only indirectly assigned. The experimenter has an indirect control over the exposure probability $\pi_i(z_i,e_i)$. Hence in some cases, relevant potential outcomes are not observed since $\pi_i(z,e)$ can be $0$, in other cases, this probability is non-uniform.

For a concrete example, consider the symmetric exposure model. Each unit has $2(d_i+1)$ potential outcomes - depending on the treatment status of unit $i$ and the number of treated units in the interference neighborhood. Here $d_i$ is the number of units in $i's$ interference neighborhood.  Consider the case when we are interested in estimating the Direct treatment effect. In this setting, the only relevant potential outcomes for each unit are $Y_i(1,0)$ and $Y_i(0,0)$. Any unit that has at least one treated neighbor is thus not included in the estimator. A naive design for a dense interference graph can lead to situations where all units only reveal nuisance potential outcomes and thus can be wasteful. Thus, we must consider new designs for estimating causal effects under interference. One more subtle issue is that a design that is optimal for estimating one type of estimand may be far from optimal for a different estimand. For example, the cluster randomized design cannot be used to estimate the Direct Treatment effect, as seen from the results of Theorem \ref{thm:propensityscores}. However, as we will see in Section \ref{sec:simulations}, simulations suggest that the cluster randomized design (along with any estimator) has the least mean sqaured error for estimating the Total Treatment effect.

\subsection{Re-randomization for estimating $\beta_1$ and $\beta_2$}
\label{sec:rerandomize}
A simple solution to avoid bad designs where nuisance potential outcomes are revealed is to re-randomize until a desired number of units fall under the treatment and exposure assignments that reveal relevant potential outcomes. For instance, consider estimating DTE using a CRD design. Let $\textbf{z}$ be a realized treatment vector and let $n(1,0)$ be the  number of units $i$ that reveal $Y_i(1,0)$. Similarly define $n(0,0)$. In general, $n(1,0) + n(0,0) < n$ and in fact, both these numbers can be $0$ for dense graphs. The re-randomization strategy would be to do a rejection sampling until $n(1,0)$ and $n(0,0)$ are larger than a given threshold. Estimation is done by using the HT estimator.

Clearly, the re-randomization approach can be very slow. An alternate strategy is to consider new designs where we maximize the number of units that reveal relevant potential outcomes. We discuss such a design to estimate the Direct Treatment Effect and the Total Treatment Effect next.

\subsection{The Independent Set Design for estimating Direct Effect}
\label{sec:independentset}
Consider the problem of estimating the direct treatment effect under the symmetric exposure model when an interference graph $G$ is known. In estimating DTE, the only relevant potential outcomes are when a unit's neighborhood in $G$ is untreated. We can construct such a design by using the concept of an independent set. An independent set $\mathbb{I}$ is a set of vertices in the graph such that no two vertex in $\mathbb I$ share an edge. A maximal independent set is a set that is not a subset of any other independent set. Independent sets have been well studied in the graph theory literature, and constructing maximal independent set is NP hard. Fortunately, for the independent set design, it is sufficient to construct a random independent set of size $k$. In fact, a design based on maximal independent sets may violate the positivity condition need to ensure estimation.

The independent set design iteratively selects nodes to be included in the independent set, also called as the Ego nodes. The nodes not in the independent set are called alters. At each step $i$, a random node is selected to be included in $\mathcal I$. Once a node is selected, the node and it's neighbors are deleted. This process is repeated until there are no more nodes remaining. Let $k$ be the number of units in the independent set. Randomization is performed by randomly assigning $k_t$ nodes to treatment.

\begin{enumerate}
	\item Let $G$ be the interference graph on $n$ nodes. Set $G_0 = G$.
	\item Let the independent set $\mathcal I = \emptyset$.
	\begin{enumerate}
			\item At step $t = 1,\ldots$, choose a unit $i$ randomly from $G_{t-1}$
			\item Insert $i$ in the independent set $\mathcal I$. 
			\item Let $G_t$ be the graph obtained by deleting $i$ and it's neighbors from $G_{t-1}$
			\item If $G_t$ is empty, stop.
	\end{enumerate}
	\item Choose $k_t$ units in $\mathcal I$ and assign them to control. 
\end{enumerate}
 
 The units in the independent set are called \emph{egos} and the units outside are called \emph{alters}. The independent set design ensures that every ego is either assigned to $(z_i=1,e_i=0)$ or $(z_i=0,e_i=0)$ condition.  Only the units in the ego set are chosen to estimate the causal effect and the alter units act as buffer units to prevent interference. Hence, it is beneficial to maximize the number of units in the ego set. %This can be done by choosing units with smallest degree instead of choosing units at random. In fact, greedy algorithm with the smallest degree can be shown to find an approximately maximal independent set for bounded degree graphs \cite{halldorsson1997greed}.
 
 Note that every unit has a positive probability of being in $\mathcal I$ in the greedy algorithm. This may not be the case in other variants of the independent set algorithm, for e.g. where one starts with the unit with smallest degree. In this case, the causal estimate is not unbiased for the $n$ units, but it could be unbiased for those units that have a positive probability of being included in the independent $\mathcal I$. One way to solve this problem is to choose a random unit with probability $p$ and a unit with the smallest degree with probability $1-p$. 
 
It is important to note that we still need to take into account the unequal probability of revealing the potential outcomes, and hence we need to use a Horvitz-Thompson or its variant to obtain unbiased or approximately unbiased estimator. This requires the knowledge of the propensity scores. Unlike the CRD and Bernoulli designs, there are no simple expressions for the propensity scores in the random independent set design. However, they can be computed using Monte Carlo simulation.

\subsection{Cluster Randomized Design for estimating Total Treatment Effect}
\label{sec:CRD}
 A drawback of the independent set design is that it cannot be used to estimate the total treatment effect. In fact, any design that reveals the relevant potential outcomes for estimating the total treatment effect will reveal nuisance potential outcomes for estimating the direct treatment effect and vice versa.
 
 Here we consider a cluster based design for estimating $TTE$, see also \cite{ugander2013graph}. Assume that the interference neighborhood depends only on the immediate neighbors as defined by the interference graph $G$. Consider partitioning the graph $G$ into $1,\ldots, K$ clusters. In the cluster randomized design, we select $n_k$ clusters and label them with treatment. The remaining clusters are labeled with control. The nodes in each cluster are assigned to the treatment status indicated by their labels. This design attempts to increase the number of relevant potential outcomes $Y_i(1,1)$ and $Y_i(0,0)$ for estimating $TTE$ and was introduced in \cite{ugander2013graph}. Unbiased Estimation is done by using the Horvitz-Thompson estimator.
 
% \subsection{Two stage randomizations for estimating $\theta$}
% \label{sec:designTheta}
% Clearly, one cannot estimate $\theta(\psi;\phi)$ by using a single stage randomization, as we cannot apply both the policies $\psi$ and $\phi$ at the same time. Instead, we can resort to a two stage randomization as done in \cite{hudgens2012toward}. Assume that the graph is divided into clusters $C_1, \ldots, C_k$. In stage 1, we randomly pick $k_1$ clusters and label them to strategy $\psi$ and the remaining are labeled to strategy $\phi$. In stage 2, the clusters are assigned treatment according to the strategy that was labeled. Along the lines of Proportion \ref{prop:unbiasedMean}, one can construct difference of means estimators of $\theta(\psi,\phi)$ and show that they are unbiased. 
  
%\subsection{Designs that attain novel notions of network balance}
%
%Help reduce / remove bias in the naive estimator for the ATE.

%
%(If I read your table correctly, we also have the variance for the IPW estimator under Bernoulli Randomization in the two by two model. May be worth saying we calculated it to compare to the formula in Proposition 14 to see if there were any major differences and/or insights we could get out of the comparison. Anything worth noting here?)

\section{Discussion}
We systematically investigated the problem of estimating causal effects when there is treatment interference. When there is arbitrary treatment interference, the number of potential outcomes explodes, rendering causal inference impossible. A starting point to resolve this issue is to posit models of potential outcomes that aim at reducing the total number of potential outcomes. These models are specified by using an interference graph $G$ and an exposure model. Using the exposure model, the potential outcomes for each unit can be decomposed into direct effects, interference effects and interaction between the two. Relying on this nonparametric linear decomposition of potential outcomes, we proposed two classes of causal estimands - the marginal effects and the average effects. These classes contain many of the popular estimands considered in the literature such as the direct treatment effect and the total treatment effect. 

Focusing on the direct treatment effect, we showed that the classical designs and difference-in-means estimators can be biased. The nature and magnitude of the bias depends on the interference graph and the exposure model, both of which may be unknown. The bias remains even after making strong linearity assumptions on the potential outcomes; however the bias is mild when the potential outcomes are linear and additive and the interference graph is sparse. On the other hand, the Horvitz-Thompson estimator is always unbiased, as long as the correct propensity scores are used. In practice, the Horvitz-Thompson estimator performs quite poorly in terms of the mean squared error due to high variance. Moreover, the weights used in the Horvitz-Thompson estimator depend on the interference graph and the exposure model which are not known in general.

%Our contribution is to build a fundamental framework for causal inference under interference. We envision this framework to serve as a starting point (and not the end goal) to formalize and pose inference problems. For e.g. we assumed the existence of an interference graph $G$. In applications, this graph is seldom known. So the inference problem is to how to estimate the graph $G$? How does the estimators of Causal Effects change when $G$ is estimated. One can ask questions about the robustness of estimators of causal effects under mis-specifications of exposure neighborhoods. One can also ask how to test different forms of exposure neighborhoods and so on.
%We  did not consider with testing the assumed form of interference, this is an important question and deserves an independent investigation. Another important direction is to develop estimators that are either robust, adaptive or even agnostic to the form of interference.

A central open issue is to design estimation strategies when the interference graph and the exposure model are not known. One possibility is to consider estimators and designs that are robust to the interference graph and the exposure model, another would be to learn the interference graph and the exposure model from the data. An important related question that deserves further investigation is testing the assumed form of interference

%\newpage
\phantomsection
%\addcontentsline{toc}{section}{References}
\bibliographystyle{plainnat}
\bibliography{references}

\appendix

\section{Numerical results}
\label{sec:simulations}
In this section, we carry out several simulation studies to illustrate the theoretical claims. In the first set of experiments, we study the bias of the difference-in-means estimators for a simple model. In the second set of experiments, we evaluate various estimation strategies for estimating DTE and TTE.

\subsection{Bias of the naive Estimator}
\label{sec:simulationnaiveEstimator}
In this section, we illustrate the bias of the difference of means estimator for estimating the direct effect, as a function of the interference. We consider the Completely Randomized design. The potential outcomes are modeled using the additive binary exposure model \ref{eq:addBinaryExposure}. We use an erdos renyi model to generate the interference graph on $n=100$ nodes with the probability of an egde between any two nodes $p=0.05$. The bias is estimated using the results in Proposition \ref{prop:Bias2by2ModelCRD}. 

Figure \ref{fig:binaryexposureBias} shows the bias of the difference in means estimator for estimating $DTE$ for a completely randomized design.   The results show that the bias increases with the interference effect and is negative when the interference is positive, and vice versa. Also, the bias goes down as the number of treated unit increases. In fact, Proportion  \ref{prop:Bias2by2ModelCRD} reveals that the bias is exactly $0$ whenever $n_c < \min_i d_i$. This is because when $n_c < \min_i d_i$, the propensity scores reduce to the weights of a completely randomized design, see Theorem \ref{thm:propensityscores}. Due to this, the bias of the naive estimator goes away. Note that this is true only in the additive model, i.e. when there is no interaction between the interference effect and the treatment of unit $i$ ($C_i(e)=0$).

\begin{figure}[h]
	\centering
	\includegraphics[width=0.5\textwidth]{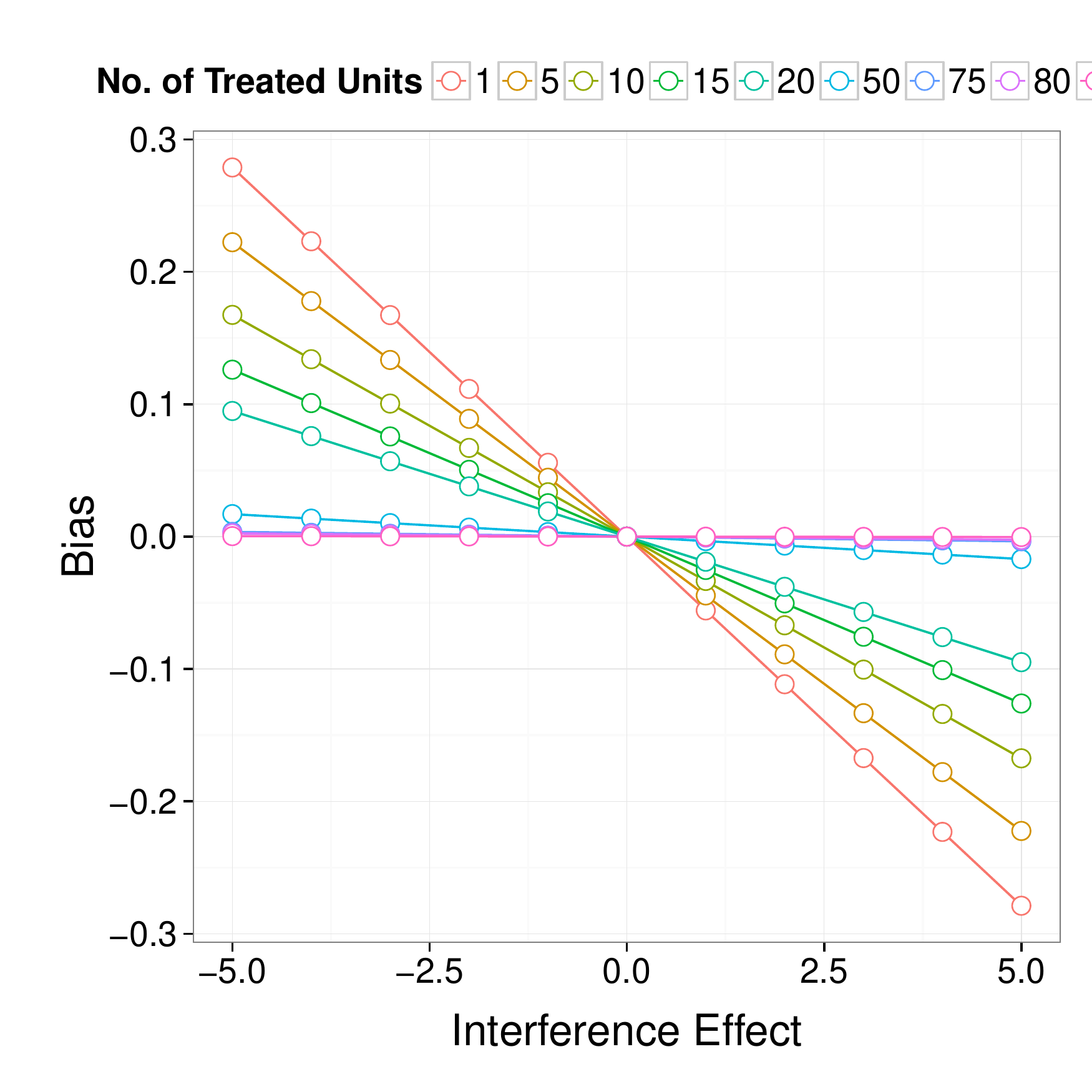}
	\caption{Bias of the Difference in means estimator in estimating $\beta_1$ as a function of interference effect and the number of treated units}
	\label{fig:binaryexposureBias}
\end{figure}
 %%% %%% %%%
 %%% %%% %%%
 %%% %%% %%%
\subsection{Estimation of Direct and Total Effects}
\label{sec:DTE}
In this section, we will perform a simulation study for estimating the direct effect $ATE_1$ and total effect $ATE_2$. There are 5 different factors that play an important role in the simulation study. These factors are the \emph{exposure model}, \emph{network model}, \emph{potential outcomes model}, \emph{estimand}, \emph{design}, and finally the \emph{estimator}. The various possible settings of these factors are listed below:

\begin{enumerate}
	\item \textit{Exposure model} - Binary Exposure, Symmetric Exposure
	\item \textit{Network model} - Erdos Renyi, Barabasi Albert, Small World Networks
	\item \textit{Potential Outcomes model} - Linear, Correlated.
	\item \textit{Estimand} - Direct Effect and Total Effect 
	\item \textit{Design} - CRD, Bernoulli, Independent Design, Cluster Randomized
	\item \textit{Estimator} - Naive, Difference of Means, Horvitz-Thompson, Ratio, GREG
\end{enumerate}

Including every possible combination of the factors into the design of the simulation would lead to a large number of combinations. Moreover, not every possible combination of the factors is possible. For example, the Independent Set Design can be used only to estimate the direct effect. Similarly, the cluster randomized design is a good design only for the total effect.  To reduce the total number of experiments, we will make some design choices. 

We will focus only on the binary exposure model. We choose $n=200$ and simulate the interference graph from three different models: An Erdos renyi model with $p=0.01$, a Barabasi Albert Model with minimum degree $2$ and attractiveness parameter $\rho = 0.1$, and the small world network with neighborhood size $1$. These choice of parameters lead to three different kinds of degree distributions: Erdos Renyi graphs are low degree graphs (for e.g., $d_{\min} = 0, d_{\max} = 8, d_{med} = 2$), the Barabasi Albert graphs show a power law behavior ($d_{\min} = 2, d_{\max} = 16, d_{med} = 2$) whereas the small world network produces almost regular graphs ($d_{\min} = 1, d_{\max} = 3, d_{med} = 2$).

For the binary exposure model, the potential outcome of each unit $i$ can be parameterized by $4$ parameters as below:
$$
Y_i(z_i,e_i) = \alpha_i + \beta_i z_i + \gamma_i e_i + \delta_i e_i z_i
$$
We consider two different models for the potential outcomes: In the \emph{uncorrelated} model, the parameters of the potential outcomes $\alpha$, $\beta$, $\gamma$ and $\delta$ are generated from independent distributions as given below: 
\begin{align*}
 \alpha_i &= N(\mu=1,\sigma=0.1)\\
 \beta_i &= Unif(0,1) \\
 \gamma_i &= Unif(0,1) \\
 \delta_i &= N(2,0.1)
\end{align*}
In the \emph{correlated} model, the parameters  generated by specifying a conditional distribution in terms of two covariates $x$ and $y$ (which can be interpreted as age and gender respectively). More specifically, the correlated potential outcomes model is given by
\begin{align*}
x_i &\sim \text{ Log Normal}(\log \mu = 3,\log \sigma = 0.5) \\
y_i &\sim \text{ Bernoulli}(p=0.4) \\
\alpha_i &= 1 + 15\log(x) - 0.5y + \text{ Normal}(\mu=0,\sigma = 1+y_i
	|\log(age)|)\\
\beta_i &= -2 -0.8x_i + 0.8y_i + \text{ Normal}(\mu=0,\sigma=2)\\
\gamma_i &= 3 + 4\log(x_i) + \text{Normal}(\mu=0,\sigma=0.1|\alpha_i|)\\
\delta_i &= 2\log(x_i) + \text{gammaa}(2,2)
\end{align*}
We considered 4 different designs, and for each design, we considered $5$ estimators: The naive estimator that takes the difference between treated and untreated units. The DofM estimator that compares the average of the relevant units, the HT estimator, ratio estimator, and the GREG estimator. For the GREG estimator, we use a linear regression model regressing on the treatment status $z_i$ and the exposure status $e_i$ of each unit and the covariates, if any. 

We will present results only for the CRD, cluster randomized design and the independent set design. The results for the Bernoulli design are very similar to the CRD design. In particular, if $p$ is the probability of assigning a unit to treatment, by letting $n_t = np$, the CRD design can approximate the Bernoulli design when $n$ is large. This is because for large $n$, one can show that the exposure probabilities for both the designs are very close, since the Binomial distribution is concentrated around its mean. 

\begin{figure}[!tbp]
	\begin{subfigure}{\textwidth}
		\centering
		\includegraphics[page=1,scale=0.25]{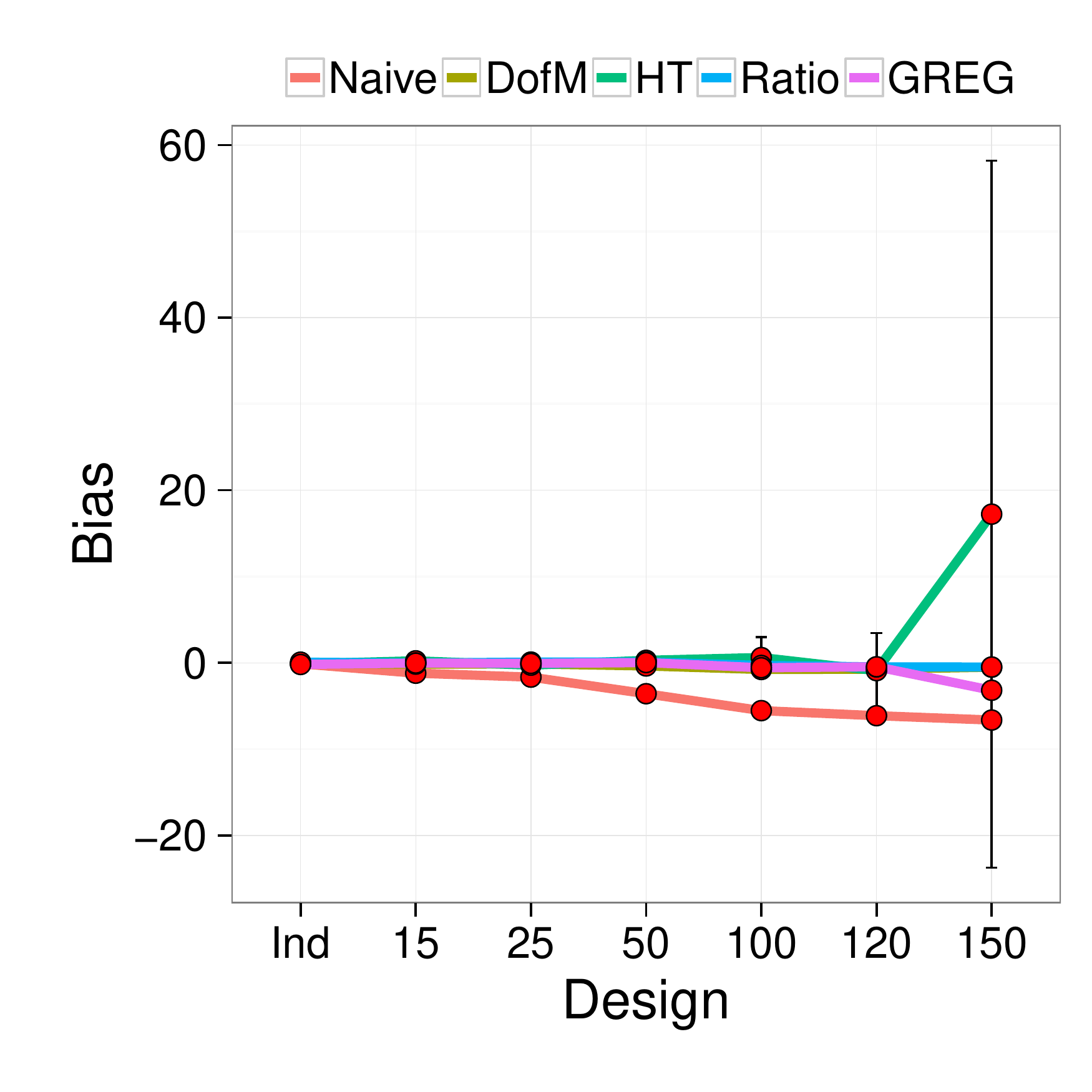}
		\includegraphics[page=2,scale=0.25]{../figures/LinePlotsATE1-erdos-covariates}
		\includegraphics[page=3,scale=0.25]{../figures/LinePlotsATE1-erdos-covariates}
		\caption{Erdos Renyi graphs with Correlated Potential Outcomes}
		\label{fig:sfig1}
	\end{subfigure}%
	
	\begin{subfigure}{\textwidth}
		\centering
		\includegraphics[page=1,scale=0.25]{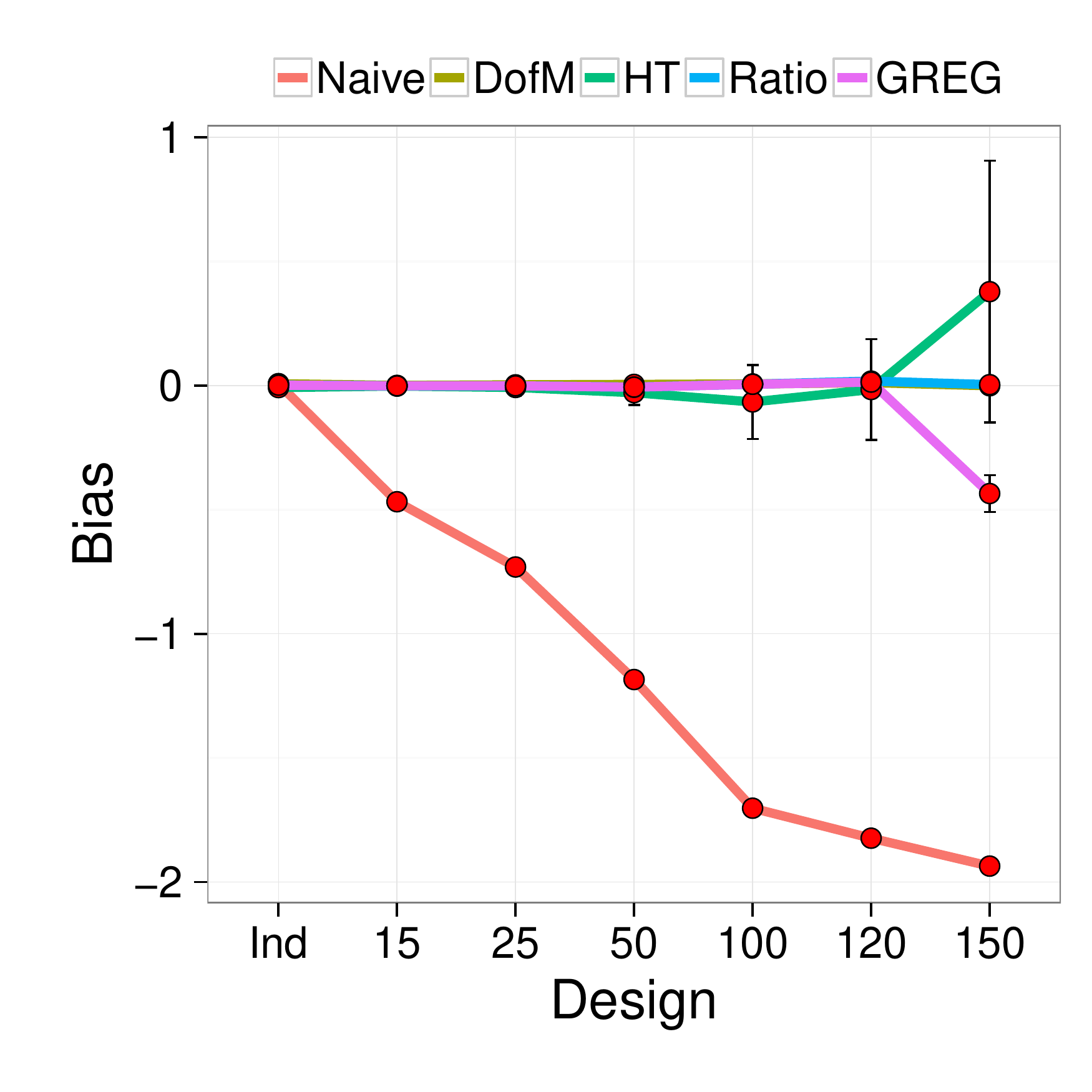}
		\includegraphics[page=2,scale=0.25]{../figures/LinePlotsATE1-barabasi-linear}
		\includegraphics[page=3,scale=0.25]{../figures/LinePlotsATE1-barabasi-linear}
		\caption{Barabasi graphs with Linear Potential Outcomes}
		\label{fig:sfig1}
	\end{subfigure}%	
	
	\begin{subfigure}{\textwidth}
		\centering
		\includegraphics[page=1,scale=0.25]{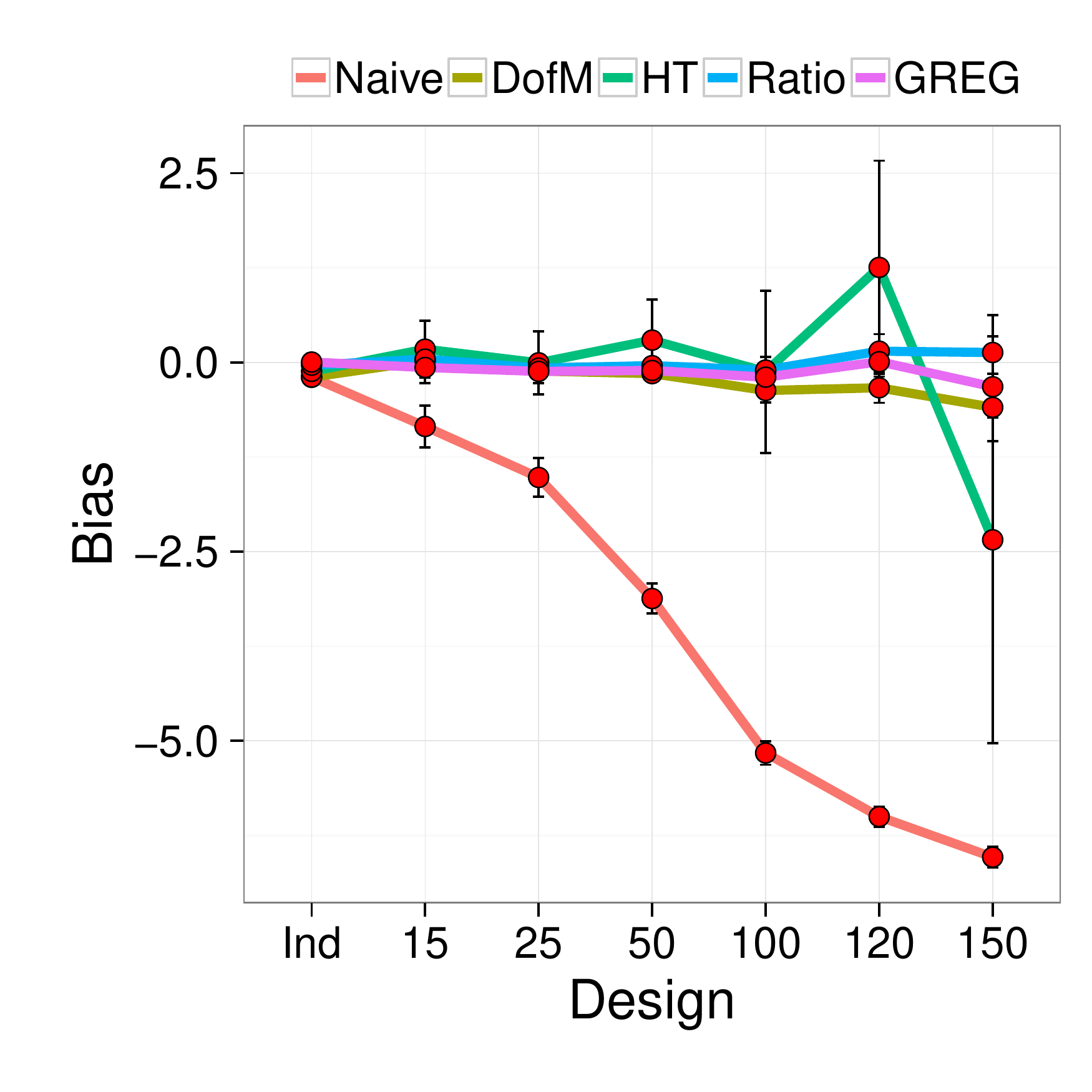}
		\includegraphics[page=2,scale=0.25]{../figures/LinePlotsATE1-smallworld-covariates}
		\includegraphics[page=3,scale=0.25]{../figures/LinePlotsATE1-smallworld-covariates}
		\caption{Small world graphs with Correlated Potential Outcomes}
		\label{fig:sfig1}
	\end{subfigure}%
	
	\caption{Estimation of Direct Effect: Bias, Variance and MSE plots of various estimators, designs for three different graph models. }
	\label{fig:directEffect}
\end{figure}

 \begin{figure}[!tbp]
 	\begin{subfigure}{\textwidth}
 		\centering
 		\includegraphics[page=1,scale=0.25]{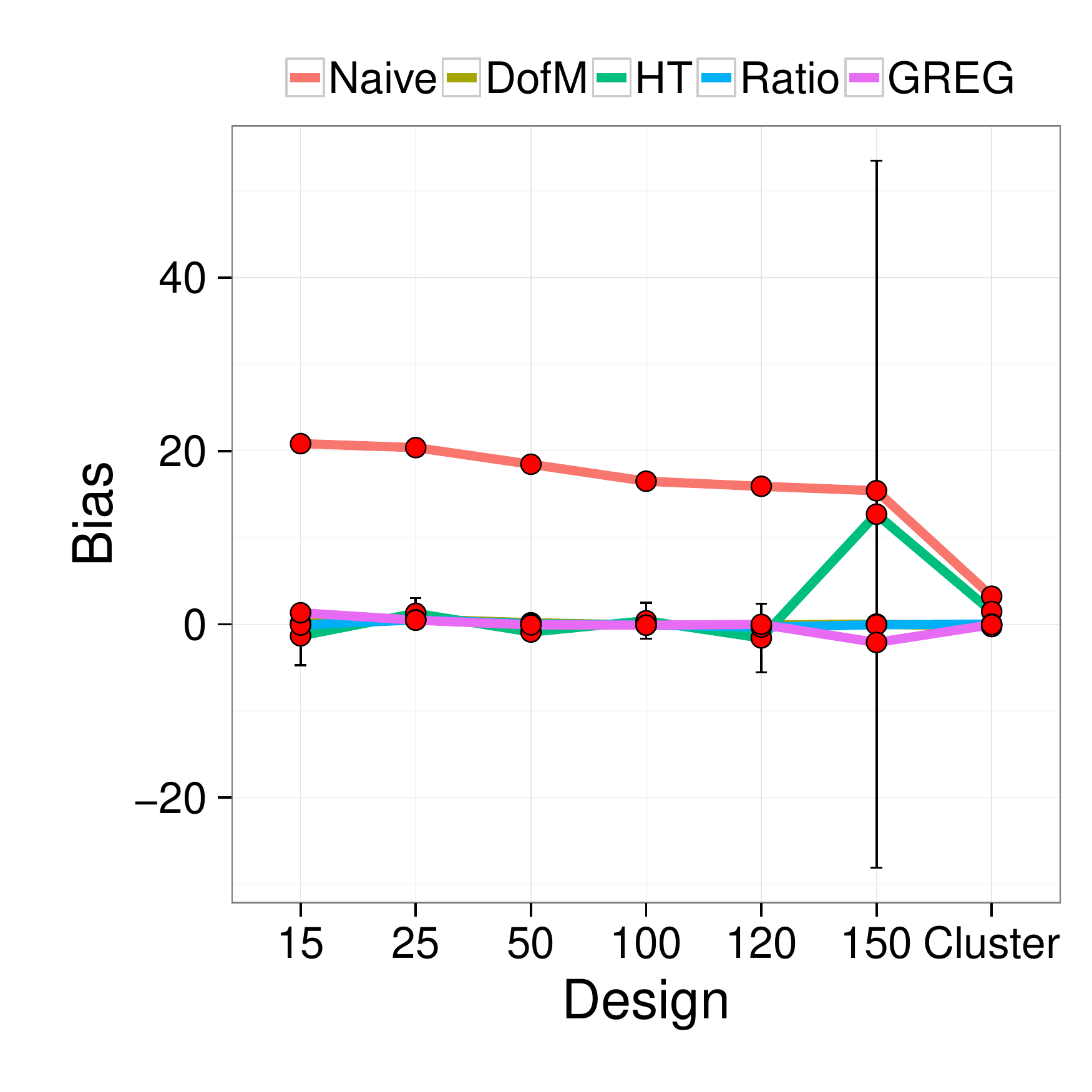}
 		\includegraphics[page=2,scale=0.25]{../figures/LinePlotsATE2-erdos-covariates}
 		\includegraphics[page=3,scale=0.25]{../figures/LinePlotsATE2-erdos-covariates}
 		\caption{Erdos Renyi graphs with Correlated Potential Outcomes}
 		\label{fig:ATE2fig1}
 	\end{subfigure}%
 	
 	\begin{subfigure}{\textwidth}
 		\centering
 		\includegraphics[page=1,scale=0.25]{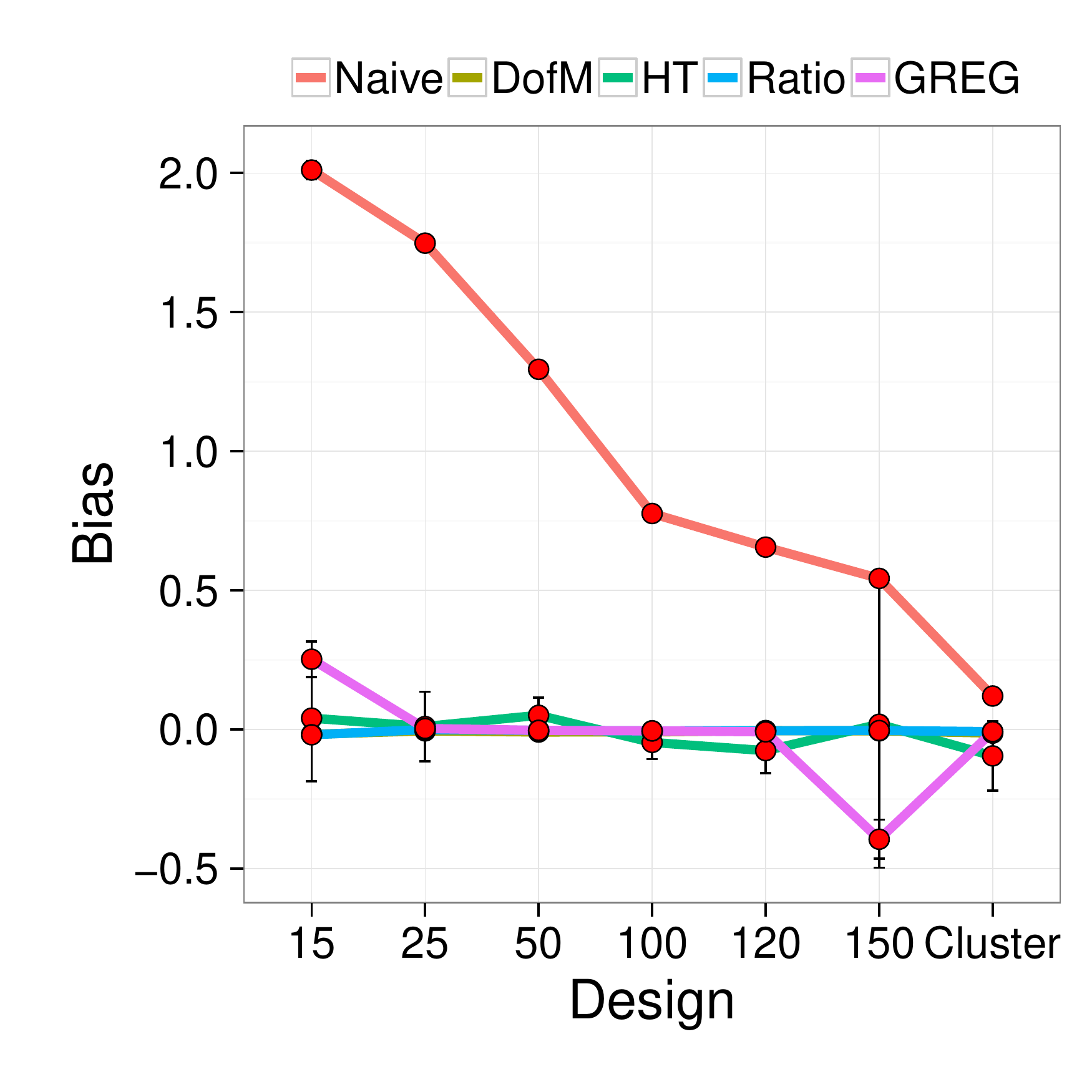}
 		\includegraphics[page=2,scale=0.25]{../figures/LinePlotsATE2-barabasi-linear}
 		\includegraphics[page=3,scale=0.25]{../figures/LinePlotsATE2-barabasi-linear}
 		\caption{Barabasi graphs with Uncorrelated Potential Outcomes}
 		\label{fig:ATE2fig2}
 	\end{subfigure}%	
 	
 	\begin{subfigure}{\textwidth}
 		\centering
 		\includegraphics[page=1,scale=0.25]{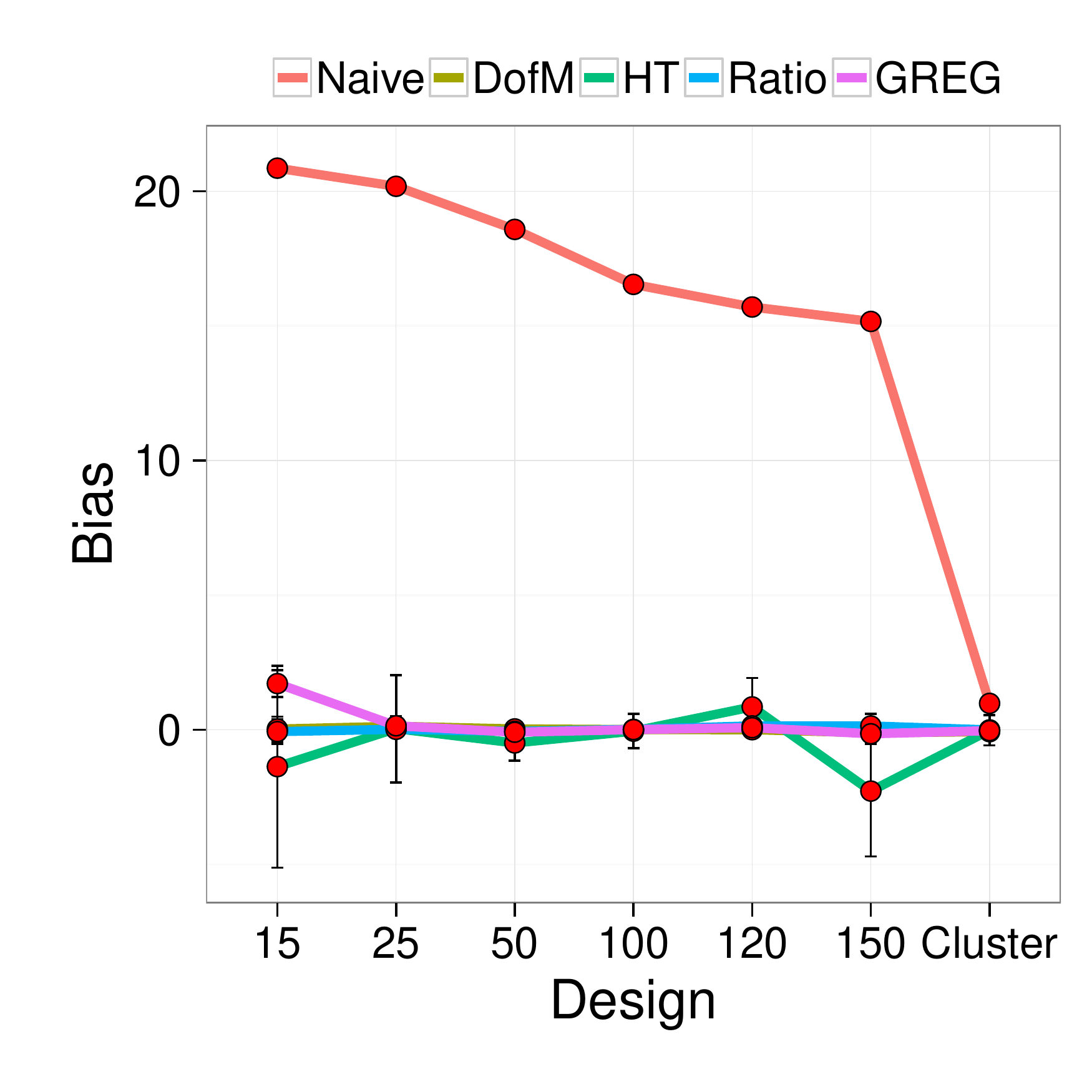}
 		\includegraphics[page=2,scale=0.25]{../figures/LinePlotsATE2-smallworld-covariates}
 		\includegraphics[page=3,scale=0.25]{../figures/LinePlotsATE2-smallworld-covariates}
 		\caption{Small world graphs with Linear Potential Outcomes}
 		\label{fig:ATE2fig3}
 	\end{subfigure}%
 	
 	\caption{Estimation of Total Effect: Bias, Variance and MSE plots of various estimators and designs for two different models of Potential Outcomes. }
 	\label{fig:TotalEffect}
 \end{figure}

Figures \ref{fig:directEffect} and \ref{fig:TotalEffect} show the plots of the bias, variance and the mean squared error for estimating the direct effect and the total effect respectively. The x-axis of the plots shows the number of treated units for CRD designs, along with the Independent set design for the direct effect (Figure \ref{fig:directEffect}) and the cluster randomized design for the total effect (Figure \ref{fig:TotalEffect}). The findings of the simulation study can be summarized below:
\begin{enumerate}
	\item Designs that are optimal (minimize the mean square error) for estimating the Total Effect are sub-optimal (in fact far from optimal) for estimating the direct effect. In all the experiments, the Independent Set Design was optimal for the Direct Effect. On the other hand, the cluster randomized design was the optimal design for estimating the Total Effect . Similarly, for the CRD, when estimating the direct effect, setting $n_t$ to a smaller value is more optimal, whereas when estimating the total effect, setting $n_t$ to a larger value is more optimal. These results are in line with the intuition. The direct effect is a function of $Y_i(1,0)$ whereas the indirect effect is a function of $Y_i(1,1)$. Hence the exposure conditions for estimating the direct effect are nuisance for estimating the total effect and vice versa. In the CRD, when $n_t$ is small, more units are allocated to the condition $(Z_i=1,E_i=0)$ as opposed to $(Z_i=1,E_i=1)$. Similarly, when $n_t$ is large, more units are allocated to the condition $(Z_i=1,E_i=1)$. 
	\item The HT estimator is unbiased in theory, but in practice can be very unstable. This is demonstrated by the large Monte Carlo standard errors of the estimate of the bias. The HT estimator has the largest variance, no matter what the design is. Hence we do not recommend the HT estimator. 
	\item The difference of means estimators - naive and DofM - are biased in theory. But surprisingly, they can perform quite well in practice. The theory developed in Section \ref{sec:Bias} offers an explanation for this behavior: Recall that the naive estimator ignores the network interference and compares average outcome of treated and control units whereas the DofM estimator compares the average outcome of the relevant units. The naive estimator has two sources of bias: Irrelevant potential outcomes and incorrect weights. The bias of the naive estimator is small whenever the design reduces the number of irrelevant potential outcomes (For e.g. the Independent set design and a CRD design with a small $n_t$). Similarly, the difference of means estimator that uses the relevant potential outcomes has only one source of bias: Incorrect weights. For large $n$ and sparse graphs, this source of bias also goes to $0$. Moreover, the difference of means estimator has smaller variance. Even though it is biased, it beats the HT estimator in terms of the Mean squared error. 
	\item The Ratio and GREG estimators behave similarly. They are both biased in theory, but in practice the bias is very small. Moreover, their variance is much smaller than the HT estimator. In some cases, the GREG estimator has higher variance than the ratio estimator. This happens when the exposure probabilities are very small which  can cause the predictions of the regression model to be unstable.
	\item The exposure probabilities and hence the bias and variance of the estimators depend on the degree of the network. The simulation suggests that the choice of the design parameters (e.g. the type of clustering in the cluster randomized design) depends on the degree of the graph. Indeed as the theoretical results show, the existence of unbiased estimators is tied to the relation between the minimum degree and the variance is a function of the exposure probabilities and hence the degree.   
\end{enumerate}
%In particular, we will consider 3 designs: CRD, Bernoulli and the Independent Set Design, 4 estimators: Two difference of means estimators, 3 Horvitz-Thompson based estimators: These are the naive difference of means, difference of means,  Horvitz-Thompson, Ratio and Greg estimators. For each combination of design and estimator, we will consider three families of graphs : Erdos Renyi, Barabasi Albert and a real network. Finally, we will consider two different potential outcome models: Uncorrelated Potential outcomes and Heterkoskedastic Potential Outcomes. 

%
%\subsection{Lessons learnt while Estimating Total Effect and Direct Effect}
%\begin{enumerate}
%	\item The estimate of the bias of the HT estimator is really bad. In theory the bias of the HT estimator is 0. In practice, we estimate the bias by a monte carlo sample. When the number of treated units are high, this estimated bias is far from zero and we need very large monte carlo sample sizes. For example, when estimating $ATE_1$ using the Bernoulli design, if we let $p=0.1$, the estimated bias of the HT estimator is close to $0$. But if $p=0.3$, the estimated bias of HT estimator is large.
%	\item If we have isloated nodes, estimating $ATE_2$ in an unbiased manner is not possible.
%\end{enumerate}

\section{Variance of Estimators}
\label{sec:var}

The problem of estimating variance of the estimators of causal effects is central to inference. However, given the complexity of the number of estimands, estimators, and designs, it is non-trivial to identify good estimators of variance. In this section, we illustrate the complexity of estimating variance for estimands under interference by deriving the formula of population variance for some simple estimators under different models. In particular, we consider the variance of the difference-in-means estimator under the linear model and the binary exposure model, and the variance of the Horvitz-Thompson estimator.

\subsection{Sources of Variation of the Horvitz-Thompson Estimator}
\label{sec:varHT}
Let $\pi_{ij}((z_1,e_1), (z_0,e_0))$ denote the joint exposure probability of unit $i$ being in condition $(z_1,e_1)$ and unit $j$ being in condition $(z_0,e_0)$. The variance of the HT estimator can be calculated by using the standard formula in the survey sampling literature.

\begin{theorem}[Variance of Horvitz-Thompson]
	\label{thm:varHT}
	The variance of the Horvitz-Thompson Estimator is given by
	\begin{align*}
	n^2Var_{HT} { }=&  
	\sum_i \frac{\pi_i(z_1,e_1) \left(1-\pi_i(z_1,e_1) \right)}{\pi^2_i(z_1,e_1)}  Y_i^2(z_1,e_1) \\
	&+ \underset{i \neq j}{\sum \sum}\frac{\left(
		\pi_{ij}(z_1,e_1) - \pi_i(z_1,e_1)\pi_j(z_1,e_1) 
		\right)}{\pi_i(z_1,e_1) \pi_j(z_1,e_1)}  \left( Y_i(z_1,e_1) Y_j(z_1,e_1)\right) \\
		&+ \sum_i \frac{\pi_i(z_0,e_0) \left(1-\pi_i(z_0,e_0) \right)}{\pi^2_i(z_0,e_0)}  Y_i^2(z_0,e_0) \\
		&+ \underset{i \neq j}{\sum \sum}\frac{\left(
			\pi_{ij}(z_0,e_0) - \pi_i(z_0,e_0)\pi_j(z_0,e_0) 
			\right)}{\pi_i(z_0,e_0) \pi_j(z_0,e_0)}  \left( Y_i(z_0,e_0) Y_j(z_0,e_0)\right) \\
			&-2\left( \underset{i \neq j}{\sum \sum}Y_i(z_0,e_0) Y_j(z_0,e_0) \frac{\pi_{ij}((z_1,e_1),(z_0,e_0))}{\pi_i(z_1,e_1) \pi_j(z_0,e_0)} - \sum_i Y_i(z_1,e_1)Y_j(z_0,e_0)
			\right)
			\end{align*}
			\end{theorem}
			
			\subsection{Sources of Variation in estimating the direct effect using difference of means estimator}
			\label{sec:varDOF}
			Here the calculations get quickly out of hand, thus we focus on the case of CRD design and a linear additive model under symmetric and binary exposure conditions.
			
			\subsubsection{Symmetric Additive linear model}
			
			\begin{proposition}[Variance of naive estimator under linear model]
				\label{prop:VarSimpleLinearModelCRD}
				Let $d_i$ be the degree of node $i$, and $m$ be the number of edges in the network.
				
				\begin{align*}
				Var(\hat{\beta}_{naive}) &= \sigma^2 \left( \frac{1}{n_t} + \frac{1}{n_c} \right) + \gamma^2 \left(c_1 m + c_2 m^2 + c_3 \sum_{i}{d_i^2} + c_4 \sum_{i \neq j}{d_i d_j}\right)\\
				\end{align*}
				where 
				\begin{align*}
				c_1 &= \frac{4n}{(n-1)(n-2)(n-3)}\left(1 - \frac{1}{n_t}\right)\left(1 - \frac{1}{n_c}\right)\\
				c_2 & = \left(
				\frac{8(n_t-1)(6n_t - 3n + 3n^2 -5nn_t)}{n(n-1)^2(n-2)(n-3)n_tn_c}
				\right)\\
				c_3 &= \left(
				\frac{4n(n_t-1)(n_t-2)}{n_tn_c(n-1)(n-2)(n-3)} + \frac{n_t}{n_c n^2} - \frac{4(n_t-1)}{n_c(n-1)(n-2)}		
				\right)\\
				c_4 & = -\frac{n_t}{n_c n^2 (n-1)}
				\end{align*}
				
				If $n_c = n_t = \frac{n}{2}$, then 
				$$Var(\hat{\tau}) = \frac{4 \sigma^2}{n}+ \gamma^2 O\left(\frac{\sum_i{d_i^2}}{n^2} + \frac{m^2}{n^2} + \frac{m}{n^2} - \frac{\sum_{ij}{d_i d_j}}{n^3} \right)$$
				\end{proposition}

				%Proposition 12. Variance of IPW estimator under linear model = .. NEED TO CALCULATE
				
				\subsubsection{Binary Additive exposure model}
				
				\begin{proposition}[Variance of naive estimator under two by two model]
					\label{prop:Var2by2Model}
					Let $\rho_i = P(z_i=1,e_i=1)$, $\pi_i = P(e_i=1)$, $\rho_{ij} = P(z_i=1,e_i=1, z_j=1,e_j=1)$, $\pi_{ij} = P(e_i=1,e_j=1)$. Under equation \ref{eq:addBinaryExposure}, assuming $\beta_i = \beta$, we have, for a CRD
					\begin{align*}
					Var\left[\hat{\tau}\right] &=
					\frac{n^2}{n_t^2 n_c^2} \left[ \sum_i \alpha_i^2 \frac{n_tn_c}{n} + \sum_{i \neq j} \alpha_i \alpha_j \frac{n_t}{n^2}\right] \\
					&+ \frac{n^2}{n_t^2 n_c^2} \left[ \sum_i \gamma_i^2 \rho_i(1-\rho_i) + \sum_{i \neq j} \gamma_i \gamma_j (\rho_{ij} - \rho_i \rho_j)\right] \\
					&+ \frac{1}{n^2} \left[ \sum_i \gamma_i^2 \pi_i(1-\pi_i) + \sum_{i \neq j} \gamma_i \gamma_j (\pi_{ij} - \pi_i \pi_j)\right] \\
					&+ \frac{n^2}{n_t^2 n_c^2} \left[ \sum_i \alpha_i \gamma_i \rho_i\frac{n_c}{n} + \sum_{i \neq j} \alpha_i \gamma_j \frac{n_t^2}{n^2}\left( \E\{e_j|z_j=z_i=1\} - \E\{e_j|z_j=1\}\right)\right] \\
					&- \frac{1}{n_t n_c} \left[ \sum_i \gamma_i^2 \rho_i(1-\pi_i) + \sum_{i \neq j} \gamma_i \gamma_j\left( \E\{e_je_iz_i\} - \rho_i \pi_j\right)\right] \\	
					&- \frac{1}{n_t n_c} \left[ \sum_i \alpha_i \gamma_i \left(\rho_i-\frac{n_t}{n}\pi_i\right) + \sum_{i \neq j} \alpha_i \gamma_j\left( \E\{z_ie_j\} - \frac{n_t}{n}\pi_j\right)\right]
					\end{align*}
					\end{proposition}
					
\section{Unbiasedness of difference-in-means estimators for estimating marginal estimands}
\label{sec:dimMarginal}
 When the treatment assignment strategy $p(\textbf{Z})$ is equal to the policy $\psi$ that defines the estimand, the difference in means estimator is unbiased for estimating $\theta(\psi)$. 
 \begin{proposition}
 	\label{prop:unbiasedMean}
 	Let $\psi$ be a restricted Bernoulli or a CRD policy and let the treatment assignment mechanism $p(\textbf{Z})$ be equal to the policy $\psi$. 
 	Then $E[\hat \beta_{naive}] = \theta(\psi)$.
 	\end{proposition} 
\section{Bias of the difference in means estimator for $TTE$}
\begin{proposition}
	\label{prop:biasGeneralCasebeta2}
	Consider the parametrized Potential outcomes given in equation \ref{eq:NIparam}:
	\begin{align*}
	Y_i(z_i,e_i) = A_i(z_i) + B_i(e_i) + z_iC_i(e_i)
	\end{align*}
	For any design $p(\textbf{Z} = \textbf{z})$, the bias of the difference in means estimator $\hat{\beta}$  (equation \ref{diffofMeans}) for estimating $\beta_2$ is:
	\begin{align*}
	b_2 &= E[\hat{\beta}] - \beta_2 \\
	&=\sum_i 
	\left(A_i(1)\left(\alpha_i(1) -\frac{1}{n}\right) - A_i(0)\left(\alpha_i(0) + \frac{1}{n}\right)\right) + \sum_i B_i(1)\left(\alpha_i(1,1) - \alpha_i(1,0) -\frac{1}{n} \right) \\
	&+\sum_i C_i(1) \left(\alpha_i(1,1) - \frac{1}{n} \right)\\
	&+ \sum_i\sum_{e_i\neq \{0,1\}} 
	B_i(e_i) \left(\alpha_i(1,e_i) - \alpha_i(0,e_i)\right) \\
	&+ \sum_i\sum_{e_i\neq \{0,1\} } C_i(e_i)\alpha_i(1,e_i)
	\end{align*}
	
	where,
	\begin{align*}
	\alpha_i(z_i,e_i) = E\left[ \frac{I(Z_i=z_i,E_i=e_i)}{\sum_i{I(Z_i=z_i)}} \right] \text{ and } \alpha_i(z_i) = E\left[ \frac{I(Z_i = z_i)}{\sum_i I(Z_i=z_i)} \right].
	\end{align*}
\end{proposition}

\section{Bias of the naive estimator for the direct effect under Cluster Randomized Design}
\label{sec:clusterRandomized}
Even without interference, the difference of means estimator is biased in cluster randomization. A simple reason is that the number of nodes in each cluster is not fixed and random.
Consider the following simple linear model of potential outcomes:
\begin{align}
\label{eq:simpleLinearModel}
Y_i = \alpha_i +  \beta_i z_{i} + \gamma \left( \sum_{ij}{g_{ij}z_j} \right)
\end{align}
For estimating the bias in the clustered randomized trial, let $c_k$ be the covariance between $Z_k$, the treatment status of cluster $k$ and $\frac{Z_k}{n_t}$. Similarly, let $d_k$ be the covariance between $1-Z_k$ and $\frac{1-Z_k}{n_c}$.

\begin{proposition}
	\label{prop:biasSimpleLinearModelCluster}
	Consider the simple linear model of the potential outcomes model specified by equation \ref{eq:simpleLinearModel}. Under a cluster randomized design, we have,
	\begin{align*}
	E[\hat{\beta}_{naive}] - \beta_{DE} 
	&=\gamma - \frac{K}{K_t} \sum_k \bar{\beta}_k n_k^2c_k 
	+ \frac{K}{K_t}\sum_k \bar{\alpha}_k n_k \left( d_k - c_k \right)
	\end{align*}
	where $\bar{\beta}_k$ is the average of $\beta_i$ for all nodes in cluster $k$. Similarly, $\bar{\alpha}_k$ is the average of $\alpha_i$ in cluster $k$.
\end{proposition}

\section{Proofs}
\subsection{Proof of Proposition \ref{prop:impossible}}
\begin{proof}
	Since there are no further assumptions on the potential outcomes, only one entry of the table of science is observable due to the fundamental problem of causal inference.  As causal effects are defined as contrasts between two distinct treatment assignments, they are unidentifiable as only one entry of the Table \ref{tab:tableofScience} is observed.
\end{proof}

\subsection{Proof of Proposition \ref{prop:param}}
	\begin{proof}
	\begin{align*}
	Y_i(\textbf{z}) &= Y_i(z_i,\textbf{z}_{N_i}) \\
	&= Y_i(z_i,f(\textbf{z}_{N_i}))\\
	&= Y_i(z_i,e_i) \\
	&=  A_i(z_i) + B_i(e_i) + z_iC_i(e_i) 
	\end{align*}
	Let us show that this is a linear map with full rank. For each unit $i$, there are a total of $2K_i$ distinct potential outcomes. In the linear parametrization, there are also a total of $2 + K_i - 1 + K_i - 1 = 2K_i$ parameters. These are: $\{A_i(0), A_i(1)\}, \{B_i(1), \ldots, B_i(K-1)\}, \{C_i(1), \ldots, C_i(K-1)\}$. For each $i$, the inverse map is given by:
	\begin{align*}
	A_i(z_i) & = Y_i(z_i,0) \\
	B_i(e_i) & = Y_i(0,e_i) - Y_i(0,0) \\
	C_i(e_i) & = Y_i(1,e_i) - Y_i(1,0) - Y_i(0,e_i) + Y_i(0,0)
	\end{align*}
	Hence the map has full rank.
	Also note that $B_i(0) = Y_i(0,0) - Y_i(0,0) = 0$ and
	$C_i(0) = Y_i(1,0) - Y_i(1,0) - Y_i(0,0) + Y_i(0,0) = 0 $.
\end{proof}

 \subsection{Proof of Proposition \ref{prop:SUTVAATE}}
 \begin{proof}
 	Note that $\bar{Y_i}(z_i;\phi) = \sum_{e_i} Y_i(z_i,e_i) \phi(E_i=e_i|Z_i=z_i) = Y_i(z_i)$.
 	Hence $\theta(\phi) = \frac{1}{n}\left(Y_i(1) - Y_i(0)\right)$.
 	The other results also follow from the definition.
 \end{proof}

\subsection{Proof of Proposition \ref{prop:ate1vsate2}}
	From Proposition \ref{prop:param} , we ave the following representation of the Potential Outcomes:
		$$Y_i(z_i,e_i) = \alpha_i + \beta_i z_i + B_i(e_i)  + z_i C(e_i)$$
		where $B_i(0) = 0$ and $C_i(0) = 0$.
	By the definition of the direct treatment effect, we have
	\begin{align*}
	\beta_{DE} &= \frac{1}{n}\left(\sum_{i=1}^n Y_i(1,0) - Y_i(0,0)\right) \\
	&=\frac{1}{n}\left(\sum_{i=1}^n \alpha_i + \beta_i - \alpha_i\right) \\
	&= \frac{1}{n}\sum_{i=1}^n\beta_i = \beta
	\end{align*}	
	Similarly, for the total treatment effect we have,
	\begin{align*}
		\beta_{TE} &= \frac{1}{n}\left(\sum_{i=1}^n Y_i(1,1) - Y_i(0,0)\right) \\
		&=\frac{1}{n}\left(\sum_{i=1}^n \alpha_i + \beta_i + B_i(1) + C_i(1)- \alpha_i\right)
	\end{align*}	
	The results follow by substituting $B_i(1)	= \gamma$ and $C_i(1) = 0$.

\subsection{Proof of Proposition \ref{prop:relation}}
	The result follows from the definition of the policies.
	
\subsection{Proof of Theorem \ref{thm:existence}}
\begin{proof}
	Let $\hat \theta$ be an unbiased estimator of $\theta$ under the design $p(\textbf{Z})$.
	Assume to the contrary that there exists a unit $i$ and a relevant potential outcome $j$ such that $\pi_i(z_j,e_j)$ is $0$. This implies that unit $i$'s potential outcome $Y_i(z_j,e_j)$ is never observed under design $p(\textbf{Z})$. Since $\hat \theta$ is a function of the observed potential outcomes, and there are no structural assumptions on the potential outcomes, the expectation of $\hat \theta$ is free from $Y_i(z_j,e_j)$. However, since $Y_i(z_j,e_j)$ is a relevant potential outcome, it appears in the definition of $\theta$. Hence $\hat \theta$ cannot be unbiased.
	
	Similarly, assume that there exists a unit $i$ and a relevant potential outcome $j$ such that $\pi_i(z_j,e_j)  = 1$. This implies that under design $p(\textbf{Z})$ we only observe $Y_i(z_j,e_j)$ for unit $i$. Since causal effects are defined as contrasts between two distinct potential outcomes, there exists a potential outcome $Y_i(z_{j'},e_{j'})$ that appears in the definition of $\theta$ but $\pi_i(z_{j'},e_{j'}) = 0$, which brings us back to the first case.  
	
	%Using godambe idea, see Singh's textbook or Huaren's paper
\end{proof}

\subsection{Proof of Proposition \ref{prop:unbiasedMean}}
\begin{proof}
	Note that $\frac{Y_i^{obs}Z_i}{\sum_i{Z_i}}$ is an unbiased estimator of $\frac{1}{N}\bar{Y}_i(1;\psi)$:
	\begin{align*}
	\E\left[\frac{Y_i^{obs}Z_i}{\sum_iZ_i}\right] &= \E\left[\frac{Y_i(Z_i,E_i)Z_i}{\sum_i Z_i}\right] \\
	&= \E_K \left[ \frac{1}{K} \E_{\psi}\left[ Y_i(Z_i,E_i)Z_i | \sum_i Z_i = K\right]\right] \\
	&=\E_k \left[ \frac{1}{K}\E_{\psi} \left[ \sum_{z_i,e_i} Y_i(z_i,e_i)Z_i I(Z_i=1,E_i=e_i)\right] \right]\\
	&= \E_K \left[ \frac{1}{K} \frac{K}{N} \sum_{e_i} Y_i(1,e_i) \pr{E_i=e_i|Z_i=1}\right] \\ 	 	
	&= \frac{1}{N}\sum_{e_i} Y_i(1,e_i) \psi(E_i=e_i|Z_i=1) \\
	&= \frac{1}{N}\bar{Y}_i(1;\psi)
	\end{align*}
	Similarly, one can show that $\frac{Y_i^{obs}(1-Z_i)}{\sum_i(1-Z_i)}$ is an unbiased estimator of $\frac{1}{N}\bar{Y_i}(0;\psi)$ which completes the proof.
%	
%	One can also show that $Y_i^{obs}$ is an unbiased estimator of $\bar{Y}_i(\psi)$:
%	\begin{align*}
%	\E\left[Y_i^{obs}\right] &= \E\left[Y_i(Z_i,E_i)\right] \\
%	&= \E_{\psi}\left[\sum_i(z_i,e_i)Y_i(z_i,e_i)\pr{Z_i=z_i,E_i=e_i}\right] \\
%	&= \sum_{z_i,e_i} Y_i(z_i,e_i) \psi(Z_i=z_i,E_i=e_i) \\
%	&= \bar{Y}_i(\psi)
%	\end{align*}
	
\end{proof}

\subsection{Proof of Proposition \ref{prop:biasGeneralCase}}
\begin{proof}
	\begin{align*}
	\hat{\beta} &= 
	\frac{\sum_iY_i^{obs}Z_i}{\sum_i{Z_i}}- \frac{\sum_i Y_i^{obs}(1-Z_i)}{\sum_{i}{1-Z_i}}\\
	&=\frac{\sum_iY_i(1,E_i)I(Z_i=1)}{\sum_i{I(Z_i=1)}} - \frac{\sum_i Y_i(0,E_i)I(Z_i=0)}{\sum_{i}{I(Z_i=0)}}\\
	&=\frac{\sum_i\sum_{e_i}Y_i(1,e_i)I(Z_i=1,E_i=e_i)}{\sum_i{I(Z_i=1)}} - \frac{\sum_i\sum_{e_i} Y_i(0,e_i)I(Z_i=0,E_i=e_i)}{\sum_{i}{I(Z_i=0)}}\\
	&=\sum_i\sum_{e_i}\left[ Y_i(1,e_i)\alpha_i(1,e_i) -
	Y_i(0,e_i)\alpha_i(0,e_i) \right] \\
	&= \sum_i \left[Y_i(1,0)\alpha_i(1,0) - Y_i(0,0)\alpha_i(0,0)\right]
	+\sum_i\sum_{e_i\neq 0} \left[ Y_i(1,0)\alpha_i(1,e_i) - Y_i(0,e_i)\alpha_i(0,e_i)\right]
	\end{align*}
	where, 
	$$\alpha_i(z_i,e_i) = \frac{I(Z_i=z_i,E_i=e_i)}{\sum_i{I(Z_i=z_i)}}$$
	Now, the bias is
	\begin{align*}
	b &= E[\hat{\beta}] - \beta_{DE} \\
	&=\sum_i 
	\left(Y_i(1,0)E\left[\alpha_i(1,0) -\frac{1}{n}\right] - Y_i(0,0)E\left[\alpha_i(0,0) + \frac{1}{n}\right]\right)\\
	&+ \sum_i\sum_{e_i\neq 0} 
	\left( Y_i(1,e_i)E[\alpha_i(1,e_i)] - Y_i(0,e_i)E[\alpha_i(0,e_i)]\right)
	\end{align*}
	
\end{proof}

\subsection{Proof of Proposition \ref{prop:expDoMs}}
\begin{proof}
	Note that from the non-parametric decomposition of the Potential outcomes given in Proposition \ref{prop:param}, we have,
	\begin{align*}
	Y_i^{obs} I(Z_i=1,E_i=0) &= A_i(1)I(Z_i=1,E_i=0) \\
	Y_i^{obs} I(Z_i=0,E_i=0) &= A_i(0)I(Z_i=0,E_i=0) \mbox{ and } \\
	Y_i^{obs} I(Z_i=1,E_i=1) &= (A_i(1)+B_i(1) + C_i(1))I(Z_i=1,E_i=1) \\
	\end{align*}
	Substituting these in the definition of $\beta_1$ and $\beta_2$ gives the result.
	
\end{proof}

\subsection{Proof of Theorem \ref{prop:exposureweights}}
The theorem follows from results of Theorem \ref{thm:propensityscores}
\subsection{Proof of Proposition \ref{prop:biasSimpleLinearModelCRD}}
\begin{proof}
	Note that 
	\begin{align*}
	\frac{\sum_i{y_i^{obs}z_i}}{n_t} = \frac{\sum_i\alpha_i z_i}{n_t} + \frac{1}{n_t} \beta \sum_{i}{z_i^2} + \frac{1}{n_t}\gamma \sum_i\sum_j g_{ij}z_j z_i %+ \frac{1}{n_t}\sum_i \epsilon_i z_i
	\end{align*}
	Similarly,
	\begin{align*}
	\frac{\sum_i{y_i^{obs}(1-z_i)}}{n_c} = \frac{\sum_i\alpha_i(1-z_i)}{n_c} + \frac{1}{n_c} \beta \sum_{i}{z_i (1-z_i)} + \frac{1}{n_c}\gamma \sum_i\sum_j g_{ij}z_j (1-z_i) %+ \frac{1}{n_c}\sum_i \epsilon_i (1-z_i)
	\end{align*}
	
	Using the facts $z_i^2 = z_i$, $\sum_{i}{z_i^2} = n_t$, and $z_i(1-z_i)=0$, we get,
	\begin{align*}
	\hat{\beta}_{naive} = \beta_{DE} + \gamma \left( \frac{\sum_i\sum_j g_{ij}z_j z_i}{n_t} - \frac{\sum_i\sum_j g_{ij}z_j (1-z_i)}{n_c}\right) + \frac{\sum_i \alpha_i z_i}{n_t} - \frac{\sum_i \alpha_i (1-z_i)}{n_c} 
	\end{align*}
	Note the following expectations - $E[z_iz_j] = P(z_i=1,z_j=1) = \frac{n_t}{n}\frac{n_t-1}{n-1}$ and $E[(1-z_i)z_j] = P(z_i=0,z_j=1) = \frac{n_t}{n}\frac{n_c}{n-1}$ and that $E[z_i] = \frac{n_t}{n}$. Using these facts and taking expectations, we get,
	\begin{align*}
	E[\hat{\beta}_{naive}] = \beta_{DE} - \gamma \frac{2m}{n(n-1)} 
	\end{align*} 
\end{proof}

\subsection{Proof of Proposition \ref{prop:biasSimpleLinearModelBernoulli}}
\begin{proof}
	Note that 
	\begin{align*}
	\frac{\sum_i{y_i^{obs}z_i}}{\sum_iz_i} = \alpha + \beta \frac{\sum_{i}{z_i^2}}{\sum_i z_i} + \gamma \sum_i\sum_j g_{ij}\frac{z_j z_i}{\sum_i{z_i}} + \sum_i  \frac{ \epsilon_i z_i}{\sum_iz_i}
	\end{align*}
	Similarly,
	\begin{align*}
	\frac{\sum_i{y_i^{obs}(1-z_i)}}{\sum_i (1 - z_i)} = \alpha + \beta \sum_{i}{\frac{z_i (1-z_i)}{\sum_i(1 - z_i)}} + \gamma \sum_i\sum_j g_{ij}\frac{z_j (1-z_i)}{\sum_i(1-z_i)} + \sum_i \frac{\epsilon_i (1-z_i)}{\sum_i(1-z_i)}
	\end{align*}
	
	Using the facts that $z_i^2 = z_i$, and $z_i(1-z_i)=0$, we get,
	\begin{align*}
	\hat{\beta}_{naive} = \beta + \gamma \left( \sum_i\sum_j g_{ij} \left(\frac{z_j z_i}{\sum_i{z_i}} - \frac{z_j (1-z_i)}{\sum_i{(1-z_i)}} \right)\right) + \sum_i \left( \frac{\epsilon_i z_i}{\sum_i z_i} -\frac{\epsilon_i (1-z_i)}{\sum_i(1-z_i)} \right)
	\end{align*}
	Note the following expectations - $E[e_iz_i] = E[e_i]E[z_i] = 0 $ and similarly, $E[e_i(1-z_i)] = 0$ as $\epsilon_i \Perp z_i$. Also note that $g_{ij} = 0$ when $i=j$, hence we need to only focus on $i \neq j$ for the calculation of the remaining expectations.
	Now let $X = \sum_i{z_i}$ and consider,
	\begin{align*}
	E\left[\frac{z_j z_i}{\sum_i{z_i}}\right] =  E\left[ \frac{1}{X} E\left[z_iz_j   \mid X \right] \right]  = E\left[ \frac{1}{X} \frac{X}{n} \frac{X-1}{n-1} \right] = E\left[\frac{X-1}{n(n-1)}\right] = \frac{\frac{np-np^n}{1-(1-p)^n -p^n}-1}{n(n-1)}
	\end{align*}
	In the last two equalities, we have used the fact that $X$ is a restricted binomial distribution with probability of success $p$, and $X \in \{1,\ldots, n-1\}$. Similarly, let $Y = \sum(1-z_i)$ be a restricted binomial with probability of success $1-p$.
	\begin{align*}
	E\left[\frac{z_j (1-z_i)}{\sum_i{(1-z_i)}}\right] &=  E\left[ \frac{1}{Y} E\left[z_j(1-z_j)   \mid Y \right] \right]  = E\left[ \frac{1}{Y} \frac{n-Y}{n} \frac{Y}{n-1} \right] \\
	&= E\left[\frac{n-Y}{n(n-1)}\right]
	= \frac{n - \frac{n(1-p)-n(1-p)^n}{1-(1-p)^n -p^n}}{n(n-1)}
	\end{align*}	
	
	Finally, let $c = \frac{1}{1-(1-p)^n-p^n}$ and note that,
	\begin{align*}
	E[X-1] - E[n-Y] = c(np-np^n)-1 - n + c(n(1-p) - n(1-p)^n) = -1
	\end{align*}
%	\begin{align*}
%	E\left[\frac{z_j (1-z_i)}{\sum_i{(1-z_i)}}\right] = \frac{1}{n(n-1)} \frac{n(1-p)}{1-(1-p)^n - p^n}	
%	\end{align*}
%	
%	\begin{align*}
%	E\left[\frac{z_j z_i}{\sum_i{z_i}}\right] =  E\left[ z_iz_j E\left[ \frac{1}{\sum_i{z_i}} \mid z_i, z_j\right] \right] = p^2E\left[ \frac{1}{2+ X} \right] = \frac{np -1 + (1-p)^n}{n(n-1)}
%	\end{align*}
%	where $X \sim Bin(n-2,p)$.
%	To see the last computation, note that
%	\begin{align*}
%	p^2E\left[\frac{1}{2+X}\right] &= p^2\sum_{x=0}^{n-2}\frac{1}{x+2}\frac{(n-2)!}{x!(n-2-x)!}p^x(1-p)^{n-2-x}\\
%	& = \frac{(n-2)!}{n!}\sum_{x=0}^{n-2} \frac{n!}{(x+2)!(n-2-x)!}p^{x+2}(1-p)^{n-2-x}(x+1)\\
%	& = \frac{1}{n(n-1)}{(np -1 + (1-p)^n)}
%	\end{align*}
%	
%	A similar calculation yields the fact,
%	\begin{align*}
%	E\left[\frac{z_j (1-z_i)}{\sum_i{(1-z_i)}}\right]  = p(1-p)E\left[ \frac{1}{n - 1 - X} \right] = \frac{np - np^n}{n(n-1)}	
%	\end{align*}
%	where $X \sim Bin(n-2,p)$.
	The result follows by plugging these expectations in the expression of $E[	\hat{\beta}_{naive}]$ using the fact that $\sum_i\sum_j g_{ij} = 2m$.
	\end{proof}

\subsection{Proof of Proposition \ref{prop:biasSimpleLinearModelCluster}}
	\begin{proof}
		Note that, as before, we have
		\begin{align*}
		\frac{\sum_i{y_i^{obs}z_i}}{\sum_iz_i} = \frac{\sum_i\alpha_i z_i}{\sum_iz_i} + \frac{ \sum_{i}{\beta_i z_i^2}}{\sum_iz_i}  + \gamma \frac{ \sum_i\sum_j z_j z_i}{\sum_iz_i} 
		\end{align*}
		Similarly,
		\begin{align*}
		\frac{\sum_i{y_i^{obs}(1-z_i)}}{\sum_i (1 - z_i)} = \frac{\sum_i\alpha_i(1-z_i)}{\sum_i (1 - z_i)} + \frac{ \sum_{i}{\beta_i z_i (1-z_i)}}{\sum_i (1 - z_i)}  + \gamma \frac{\sum_i\sum_j z_j (1-z_i)}{\sum_i (1 - z_i)}  
		\end{align*}
		
		Using the facts that $z_i^2 = z_i$, and $z_i(1-z_i)=0$, we get,
		\begin{align}
		\hat{\tau} = \frac{\sum_i\beta_i z_i}{\sum_iz_i} +  \gamma \underset{i \neq j}{\sum \sum} 
		\left(\frac{z_j z_i}{\sum_i{z_i}} - \frac{z_j (1-z_i)}{\sum_i{(1-z_i)}} \right) 
		+ \sum_i \left( \frac{\alpha_i z_i}{\sum_i z_i} -\frac{\alpha_i (1-z_i)}{\sum_i(1-z_i)} \right) \label{eq:beforeExP}
		\end{align}
		However, now we have the case that $\sum_iz_i = \sum_{k=1}^K n_k z_k$ and $\sum_i{1-z_i} = \sum_{k=1}^K{n_k (1-z_k)}$.
		Let $\bar{\beta}_k$ be the average of $\beta_i$ for all nodes in cluster $k$. Similarly, we define $\bar{\alpha}_k$. Thus, we get,
		\begin{align}
		\hat{\tau} = \frac{\sum_k \bar{\beta}_k n_k z_k}{\sum_k n_k z_k} +  \gamma  \underset{i \neq j}{\sum \sum}  
		\left(\frac{z_j z_i}{\sum_i{z_i}} - \frac{z_j (1-z_i)}{\sum_i{(1-z_i)}} \right) 
		+ \sum_k \left( \frac{\bar{\alpha}_k n_k z_k}{\sum_k n_k z_k} -\frac{\bar{\alpha}_k n_k (1-z_k)}{\sum_k n_k (1-z_k)} \right) \label{eq:beforeExP}
		\end{align}
		Note the following:
		\begin{align*}
		\underset{i \neq j }{\sum \sum}\left( \frac{z_iz_j}{\sum_i z_i} - \frac{(1-z_i)z_j}{\sum_i (1-z_i)} \right) = -1
		\end{align*}
		This is because,
		\begin{align*}
		\frac{\sum_i \sum_j z_i z_j}{\sum_i z_i} = \frac{\sum_j z_j \sum_i z_i}{\sum_i z_i } =  \sum_i z_i
		\end{align*}
		and,
		\begin{align*}
		\frac{\sum_i \sum_j (1-z_i) z_j}{\sum_i (1 - z_i)} = \frac{\sum_j z_j \sum_i (1-z_i)}{\sum_i (1- z_i) } =  \sum_i z_i.
		\end{align*}
		Hence, we have,
		\begin{align*}
		0 &= \sum_i \sum_j \left( \frac{z_iz_j}{\sum_i z_i} - \frac{(1-z_i)z_j}{\sum_i (1-z_i)} \right) \\ 
		&= 
		\underset{i \neq j }{\sum}\left( \frac{z_iz_j}{\sum_i z_i} - \frac{(1-z_i)z_j}{\sum_i (1-z_i)} \right) + \underset{i = j }{\sum \sum}\left( \frac{z_iz_j}{\sum_i z_i} - \frac{(1-z_i)z_j}{\sum_i (1-z_i)} \right) \\
		&= \underset{i \neq j }{\sum}\left( \frac{z_iz_j}{\sum_i z_i} - \frac{(1-z_i)z_j}{\sum_i (1-z_i)} \right) + \underset{i = j }{\sum \sum}\left( \frac{z_i^2}{\sum_i z_i} - \frac{(1-z_i)z_i}{\sum_i (1-z_i)} \right)\\
		&= \underset{i \neq j }{\sum}\left( \frac{z_iz_j}{\sum_i z_i} - \frac{(1-z_i)z_j}{\sum_i (1-z_i)} \right) + 1
		\end{align*}
		Thus, to evaluate the bias, we need to compute the expectations of the following terms:
		\begin{align*}
		E\left[ \frac{n_k z_k}{\sum_k n_k z_k} \right] \mbox{ and } E\left[ \frac{n_k (1-z_k)}{\sum_k n_k (1-z_k)} \right]
		\end{align*}
		We will use the trivial identity 
		$$E\left[\frac{U}{V}\right] = \frac{1}{E[V]}\left[E[U] - Cov\left(\frac{U}{V},V\right)\right]$$
		Thus, we get,
		\begin{align*}
		E\left[\frac{n_kz_k}{\sum_k{n_kz_k}}\right] &= \frac{1}{\sum_k E\left[n_k z_k\right]} \left[ E[n_k z_k - Cov\left(\frac{n_kz_k}{\sum_k n_k z_k}, n_kz_k\right)\right]\\
		&=\frac{K}{n K_t}\left[\frac{n_kK_t}{K} - n_k^2 Cov\left(\frac{z_k}{\sum_k n_k z_k}, z_k\right) \right]\\
		&=\frac{n_k}{n} - \frac{n_k^2K}{K_t} Cov\left(\frac{z_k}{\sum_k n_k z_k}, z_k\right) \\
		&=\frac{n_k}{n} - \frac{n_k^2K}{K_t} c_k 
		\end{align*}
		Similarly, one can show that
		\begin{align*}
		E\left[\frac{n_k(1-z_k)}{\sum_k{n_k(1-z_k)}}\right]
		&=\frac{n_k}{n} - \frac{n_k^2K}{K_c} Cov\left(\frac{1-z_k}{\sum_k n_k (1-z_k)}, 1-z_k\right) \\
		&=\frac{n_k}{n} - \frac{n_k^2K}{K_c} d_k 
		\end{align*}
		
		Thus, we get
		\begin{align*}
		E[\hat{\tau}] &= 
		\sum_k \bar{\beta}_k \left(\frac{ n_k}{n} - \frac{n_k^2Kc_k}{K_t}\right) -
		\gamma
		+ \sum_k \bar{\alpha}_k \left( \frac{n_k K (d_k - c_k)}{K_t} \right)\\
		&=\bar{\beta} - \gamma - \frac{K}{K_t} \sum_k \bar{\beta}_k n_k^2c_k 
		+ \frac{K}{K_t}\sum_k \bar{\alpha}_k n_k \left( d_k - c_k \right)
		\end{align*}
		
	\end{proof}

\subsection{Proof of Propositions \ref{prop:Bias2by2ModelCRD} and \ref{prop:Bias2by2ModelBernoulli} }
\begin{proof}
	From the definition of $\hat{\beta}_{naive}$ and using the facts $z_i^2 = z_i$ and $z_i(1-z_i) = 0$, we have,
	\begin{align*}
	\frac{\sum_i{y_i^{obs}z_i}}{\sum_iz_i} = \frac{\sum_i\alpha_i z_i}{\sum_iz_i} + \frac{ \sum_{i}{\beta_i z_i}}{\sum_iz_i}  + \frac{ \sum_i (\gamma_{i} + \theta_i) e_i z_i }{\sum_iz_i}
	\end{align*}
	and,
	\begin{align*}
	\frac{\sum_i{y_i^{obs}(1-z_i)}}{\sum_i (1-z_i)} = \frac{\sum_i\alpha_i (1-z_i)}{\sum_i (1-z_i)} + \frac{ \sum_i \gamma_{i}e_i (1-z_i)}{\sum_i (1-z_i)}
	\end{align*}
	Under both the CRD and Bernoulli trial, we get, (e.g. using Proposition \ref{prop:iidratio})
	\begin{align*}
	E[\hat{\beta}] = \bar{\beta} + \sum_i\left((\gamma_i + \theta_i) E\left[\frac{e_iz_i}{\sum_iz_i}\right] - \gamma_i E\left[\frac{e_i(1-z_i)}{\sum_i (1-z_i)}\right] \right)
	\end{align*}
	Under CRD, $\sum_iz_i = n_t$ and $\sum_i(1-z_i) = n_c$. Moreover, $E[z_i] = n_t$ and $E[1-z_i] = n_c$. Let $P(z_i=1,e_i=1) = \pi_i(1,1)$ and $P(z_i=0,e_i=1) = \pi_i(0,1)$, then we have,
		\begin{align*}
		E[\hat{\beta}] &= \beta_1 + \sum_i  \gamma_i \left( \frac{\pi_{i}(1,1)}{n_t} - \frac{\pi_i(1,0)}{n_c} \right) + \sum_i \theta_i \frac{\pi_{i}(1,1)}{n_t} \\
		\mathbb E[\hat{\beta}] &= \beta_1 + \frac{1}{n}\sum_i  \gamma_i \left( \frac{\binom{n_c-1}{d_i} - \binom{n_c}{d_i}}{\binom{n-1}{d_i}} \right) + \sum_i \theta_i \frac{\pi_{i}(1,1)}{n_t}\\
		\mathbb E[\hat{\beta}] &= \beta_1 - \frac{1}{n}\sum_i  \gamma_i \left( \frac{\binom{n_c-1}{d_i-1}}{\binom{n-1}{d_i}} \right) + \frac{1}{n}\sum_i \theta_i \left(1 - \frac{\binom{n_c}{d_i}}{\binom{n-1}{d_i}}\right)
		\end{align*}
	Now let us compute the bias for a Bernoulli trial. We need to compute the following expectations:
	\begin{align*}
	E\left[\frac{e_iz_i}{\sum_iz_i}\right] \mbox{ and } \left[\frac{e_i(1-z_i)}{\sum_i (1-z_i)}\right] 
	\end{align*}
	\begin{align*}
	E\left[\frac{e_i z_i}{\sum_i z_i} \right] &= E\left[ E\left[ \frac{e_i z_i}{\sum_i z_i} \bigg| \sum_i{z_i} = k \right] \right] = E\left[ \frac{1}{\sum_i z_i} P\left(e_i=1,z_i=1\bigg|\sum_i z_i =k \right)\right] \\
	&= 
	E_k \left[ 
	\frac{1}{k} \frac{k}{n}\left[ 1 - \frac{n_c (n_c-1)\ldots (n_c-d_i+1)}{(n-1)(n-2) \ldots (n-d_i)}\right] \right], \mbox{ where } n_c = n-k \\
	&=\frac{1}{n} - E_{k} \left[ \frac{n_c^{(d_i)} }{n^{(d_i)} (n-d_i)} \right] = \frac{1}{n} - E_{k} \left[ \frac{(n-k)^{(d_i)} }{n^{(d_i)} (n-d_i)} \right]\\
	&=\frac{1}{n} - \frac{n^{(d_i)} (1-p)^{d_i}}{n^{(d_i)} (n-d_i)} = \frac{1}{n} - \frac{(1-p)^{d_i}}{n-d_i}
	\end{align*}
	The last equation follows from the easy to show fact that if $X \sim Bin(n,p)$, then $E\left[X^{(r)}\right] = n^{(r)}p^r$, and that $n-k \sim Bin(n,1-p)$.
	Using a similar argument, one can show that 
	$$E\left[\frac{e_i(1-z_i)}{\sum_i (1-z_i)}\right]  = \frac{1}{n} - \frac{(1-p)^{d_i}}{n}$$
	Thus, we get,
	\begin{align*}
	E[\hat{\tau}] &= \bar{\beta} + \sum_i\left((\gamma_i + \theta_i) E\left[\frac{e_iz_i}{\sum_iz_i}\right] - \gamma_i E\left[\frac{e_i(1-z_i)}{\sum_i (1-z_i)}\right] \right)\\
	&= \bar{\beta} + \sum_i\left((\gamma_i + \theta_i) \left[ \frac{1}{n} - \frac{(1-p)^{d_i}}{n-d_i}\right] - \gamma_i \left[ \frac{1}{n} - \frac{(1-p)^{d_i}}{n}\right] \right) \\
	&= \bar{\beta} - \sum_i\left( \frac{d_i \gamma_i (1-p)^{d_i}}{n(n-d_i)}\right)  + \sum_i \theta_i \left[ \frac{1}{n} - \frac{(1-p)^{d_i}}{n}\right] \\
	\end{align*}
	
\end{proof}

\subsection{Proof of Theorem \ref{thm:generalLinear}}
\begin{proof}
	Note that by the consistency assumption, we have,
	$$
	\hat \theta_1 = \sum_{i=1}^n Y_i^{obs}w_i(\textbf{z}) = \sum_{i=1}^n \sum_{z,e} Y_i(z,e)I(z_i=z,e_i=e)w_i(\textbf{z}) 
	$$
	
	\begin{align*}
	\mathbb E [\hat \theta_1 ] 
	{ }= & \sum_{i=1}^n \underset{z,e}{\sum} Y_i(z,e) 
	\left(
	\sum_{\textbf{z} \in \Omega} I(z_i=z,e_i=e) w_i(\textbf{z})p(\textbf{z})
	\right) \\
	&\sum_{i=1}^n \sum_{\textbf{z} \in \Omega}  w_i(\textbf{z}) Y_i(z_1,e_1) I(z_i=z_1,e_i=e_1)p(\textbf{z})  \\
	&+\sum_{i=1}^n \sum_{\textbf{z} \in \Omega}  w_i(\textbf{z}) Y_i(z_0,e_0) I(z_i=z_0,e_i=e_0)p(\textbf{z})  \\
	&+\sum_{i=1}^n \underset{(z,e) \neq (z_1,e_1), (z_0,e_0)}{\sum}\sum_{\textbf{z} \in \Omega}  w_i(\textbf{z}) Y_i(z_1,e_1) I(z_i=z,e_i=e)p(\textbf{z})	\\
	=&\sum_{i=1}^n  \sum_{\textbf{z} \in \Omega_i(z_1,e_1)} w_i(\textbf{z}) Y_i(z_1,e_1) p(\textbf{z}) 
	+ \sum_{i=1}^n \sum_{\textbf{z} \in \Omega_i(z_0,e_0)}  w_i(\textbf{z}) Y_i(z_0,e_0) p(\textbf{z}) \\
	&+ \sum_{i=1}^n \underset{(z,e) \neq (z_1,e_1), (z_0,e_0)}{\sum}\sum_{\textbf{z} \in \Omega_i(z,e)}  w_i(\textbf{z}) Y_i(z_1,e_1) p(\textbf{z})	\\
	=&\sum_{i=1}^n  Y_i(z_1,e_1) \left(\sum_{\textbf{z} \in \Omega_i(z_1,e_1)} w_i(\textbf{z})  p(\textbf{z})\right) 
	+ \sum_{i=1}^n Y_i(z_0,e_0) \left(\sum_{\textbf{z} \in \Omega_i(z_0,e_0)}  w_i(\textbf{z})  p(\textbf{z})\right) \\
	&+ \sum_{i=1}^n \underset{(z,e) \neq (z_1,e_1), (z_0,e_0)}{\sum} Y_i(z,e) \left(\sum_{\textbf{z} \in \Omega_i(z,e)}  w_i(\textbf{z})  p(\textbf{z})\right)	\\
	=& \sum_{i=1}^n Y_i(z_1,e_1) \left(\frac{1}{n}\right) - \sum_{i=1}^n Y_i(z_0,e_0) \left(\frac{1}{n}\right)
	\end{align*}
	where the last line is required for unbiasedness. Since this is an identity in $Y_i(z,e)$, we have
	\begin{align*}
	\forall i=1, \ldots, n, \sum_{\textbf{z} \in \Omega_i(z_1,e_1)} w_i(\textbf{z})  p(\textbf{z}) &= \frac{1}{n} \\
	\forall i=1, \ldots, n, \sum_{\textbf{z} \in \Omega_i(z_0,e_0)} w_i(\textbf{z})  p(\textbf{z}) &= -\frac{1}{n} \\
	\forall i, \forall (z,e) \neq (z_1,e_1), (z_0,e_0), \sum_{\textbf{z} \in \Omega_i(z,e)} w_i(\textbf{z})  p(\textbf{z}) &= 0
	\end{align*} 
	Let us now show that $0 < \pi_i(z_1,e_1) < 1$ is necessary for unbiasedness. Suppose there exists a $j$ such that $\pi_j(z_1,e_1) = \sum_{\textbf{z} \in \Omega_j(z_1,e_1)} p(\textbf{z}) = 0$, then $p(\textbf{z}) = 0$ $\forall$ $\textbf{z} \in \Omega_j(z_1,e_1)$. This means that $E[\hat \theta]$ is free of $Y_j(z_1,e_1)$ irrespective of $w_j(\textbf{z})$, see line 4 of the previous equation, and hence cannot be equal to $\sum_{i=1}^n Y_i(z_1,e_1)$. Similarly, suppose there exists a $j$ such that $\pi_j(z_1,e_1) = \sum_{\textbf{z} \in \Omega_j(z_1,e_1)} p(\textbf{z}) = 1$. This implies that $\Omega_j(z_1,e_1) = \Omega$. Since for fixed $j$, the sets $\Omega_j(z,e)$ are disjoint, we have $p(\textbf{z}) = 0$ for any $\textbf{z} \in \Omega_j(z_0,e_0)$. Hence by the previous argument, $E[\hat \theta_1]$ will be free of $Y_j(z_0,e_0)$ and therefore cannot be unbiased. A similar argument will show the necessity of $0< \pi_i(z_0,e_0)$.
\end{proof}

\subsection{Proof of Theorem \ref{thm:HT}}
\begin{proof}
	Note that $\hat \theta_2$ given by \ref{eq:linearestsmall} is contained in the class of estimators given by $\hat \theta$ \ref{eq:linearestimator}, since $w(\textbf{z}) = w(z,e)$. Using the results from Theorem \ref{thm:generalLinear}, we have $\hat \theta_2$ is unbiased iff for each $i=1, \ldots, n$
	\begin{align*}
	\underset{\textbf{z} \in \tau_i(z_1,e_1) }{\sum} w_i(\textbf{z})p(\textbf{z})  &= \frac{1}{n}\\
	\implies \underset{\textbf{z} \in \tau_i(z_1,e_1) }{\sum} w_i(z,e)p(\textbf{z})  &= \frac{1}{n}\\
	\implies w_i(z_1,e_1) \underset{\textbf{z} \in \tau_i(z_1,e_1) }{\sum} p(\textbf{z})  &= \frac{1}{n}, \text{ } \\
	\implies w_i(z_1,e_1) \pi_i(z_1,e_1)  &= \frac{1}{n}\\
	\implies w_i(z_1,e_1)   &= \frac{1}{n\pi_i(z_1,e_1)}
	\end{align*}
	A similar argument shows that $w_i(z_0,e_0) = \frac{1}{n \pi_i(z_0,e_0)}$ and $w_i(z,e) = 0$ for all $(z,e) \neq (z_1,e_1) and (z_0,e_0)$.
\end{proof}

\subsection{Proof of Theorem \ref{thm:exposureProbSymmetric}}
\begin{proof}
	For a CRD design, we have,
	\begin{align*}
	\alpha_i(1,e_i) 
	&= \mathbb{E} \left[ \frac{I(Z_i=1,E_i=e_i)}{\sum_i{Z_i}} \right] \\
	&= \frac{1}{n_t}\pr{I(Z_i=1,E_i=e_i} \\
	&= \frac{1}{n_t}\frac{n_t}{n} \frac{\binom{n_t-1}{e_i} \binom{n_c}{d_i-e_i}}{\binom{n-1}{d_i}} \text{ if } n_t \geq e_i +1 \text{ and } n_c \geq d_i-e_i, 0 \text{ otherwise} \\
	&=\frac{1}{n} \frac{\binom{n_t-1}{e_i} \binom{n_c}{d_i-e_i}}{\binom{n-1}{d_i}} \text{ if } n_t \geq e_i +1 \text{ and } n_c \geq d_i-e_i, 0 \text{ otherwise}
	\end{align*}
	For a Bernoulli trial, we have,	
	\begin{align*}
	\alpha_i(1,e_i) 
	&= \mathbb{E} \left[\frac{I(Z_i=1,E_i=e_i)}{\sum_i Z_i}\right]  \\
	&= \mathbb{E}_k\left[ \mathbb{E} \left[ \frac{I(Z_i=1,E_i=e_i)}{\sum_i Z_i} \bigg| \sum_i{Z_i} = k \right] \right] \\
	&= \mathbb{E} \left[ \frac{1}{\sum_i Z_i} \mathbb{P}\left(Z_i=1,E_i=e_i\bigg|\sum_i Z_i =k \right)\right] \\
	&= 
	\mathbb{E}_K \left[ 
	\frac{1}{K} \frac{K}{n} \frac{\binom{K-1}{e_i} \binom{n-K}{d_i-e_i} }{\binom{n-1}{d_i}} \right], \mbox{ where } \sum_i{Z_i} = K \\
	&= 
	\frac{1}{n}\mathbb{E}_K \left[ 
	\frac{\binom{K-1}{e_i} \binom{n-K}{d_i-e_i} }{\binom{n-1}{d_i}} \right]
	\end{align*}
	where $K$ is a restricted binomial random variable with support on $\{1,\ldots, N-1\}$ and $\pr{K=k} = \frac{\binom{n}{k}p^k(1-p)^{n-k}}{1 - (1-p)^n - p^n}$.
	A similar proof holds for the other cases.
\end{proof}

\subsection{Proof of Theorem \ref{thm:HTinadmissible}}
To prove Theorem \ref{thm:HTinadmissible}, we first need an intermediate Lemma proved below. This Lemma essentially states that given any unbiased estimator of $\theta$ whose minimum variance is strictly greater than $0$, one can always construct a new estimator that has lower mean squared error than the unbiased estimator. Lemma \ref{lemma:HT} follows from a result of \cite{godambe1965admissibility}.
\begin{lemma}
	\label{lemma:HT}
	Let $\mathbb P$ be any design and let $\hat \theta$ be an unbiased estimator of a generic causal effect $\theta$ under the design $\mathbb P$. Suppose $\min_{\mathbb T} Var_{\mathbb P}(\hat \theta) > 0$ where $\mathbb T$ is the table of science. Then there exists an estimator $\hat \theta_1$ such that $MSE[\hat \theta] < MSE[\hat \theta_1]$ for all $\theta$.
	\end{lemma}
	\begin{proof}
		Let $0 < k \leq 1$ be a constant to be specified later and let $\hat \theta_1 = (1-k) \hat \theta$.
		Then we have
		\begin{align}
		MSE(\hat \theta_1) &= \E \left((1-k)\hat \theta - \theta\right)^2 \nonumber \\
		&= MSE(\hat \theta) + k^2\left(Var(\hat \theta) + \theta^2 \right)- 2k Var(\hat \theta) \label{eq:k}
		\end{align}
		Note that the MSE is a function of the design $\mathbb P$ and the unknown but fixed potential outcomes $\{Y_i(z_i,e_i)\}_{i=1}^n$ given by the entries of Table of science $\mathbb T$. In fact, since $\hat \theta$ is an unbiased estimator, one can show that it is a function of only the relevant potential outcomes, i.e. $\{Y_i(z_1,e_1)\}_{i=1}^n$ and $\{Y_i(z_0,e_0)\}_{i=1}^n$. Now if $k > 0$ and 
		$$ k\left( Var(\hat \theta) + \theta^2\right) < 2 Var(\hat \theta),$$
		for all $\theta$, then $MSE(\hat \theta_1) < MSE(\hat \theta)$.
		We need to show that such a $k$ exists. To see that this is true, let
		\begin{align*}
		k_0 = \underset{\mathbb T}{\min} \frac{2Var(\hat \theta)}{Var(\hat \theta) + \theta^2}
		\end{align*}
		Since $\min_{\mathbb T} Var (\hat \theta) > 0$, we have $k_0 > 0$. Let $ k = \min(k_0,1)$. Hence we have $0 < k \leq 1$.
		
		When $k_0 < 1$, $k= k_0$ and by definition of $k_0$ , $MSE(\hat \theta_1) < MSE(\hat \theta)$.
		
		When $k_0 \geq 1$, $k=1$, and $\hat \theta_1 = 0$. But $k_0 \geq 1$ implies that $2 Var(\hat \theta) > Var(\hat \theta) + \theta^2$ or $Var(\hat \theta) \geq \theta^2$. Using this fact and substituting $k=1$ in equation \ref{eq:k}, one can see that $MSE(\hat \theta_1) < MSE(\hat \theta)$. Note that in such a case, the variance of $\hat \theta$ is so large that a constant estimator is able to beat it. This happens when $Var (\hat \theta) > \theta^2$, making estimation impossible.	
		\end{proof}
		
		\begin{proof}[Proof of Theorem \ref{thm:HTinadmissible}]
			
			To show that the Horvitz-Thompson estimator is inadmissible, from Lemma \ref{lemma:HT}, it suffices to show that the variance of the HT estimator can never be zero for a non-constant design $\mathbb P$.
			Let $X_i = I(Z_i = z_0, E_i=e_0)$ and $Y_i = I(Z_i=z_1,E_i=e_1)$ and $p_i = \mathbb E (X_i)$ and $q_i = \mathbb E (Y_i)$. Let us assume to the contrary that the variance of $\hat \theta_{HT} = 0$ for some $\theta$ for a non-constant design $\mathbb P$.
			The variance of $\hat \theta_{HT}$ is $0$ iff 
			\begin{align}
			\hat \theta_{HT} &= \mathbb E \left(\hat \theta_{HT} \right) = \theta \mbox{ a.s. } \mathbb P \nonumber \\
			\iff 	\sum_{i=1}^n \left(Y_i(z_1,e_1)\frac{X_i}{p_i} - Y_i(z_0,e_0)\frac{Y_i}{q_i} \right) &= \sum_{i=1}^n \left(Y_i(z_1,e_1) - Y_i(z_0,e_0)\right) \mbox{ a.s. } \mathbb P \nonumber\\
			\iff \sum_{i=1}^n \frac{Y_i(z_1,e_1)}{p_i}\left(X_i - p_i\right) &= \sum_{i=1}^n \frac{Y_i(z_0,e_0)}{q_i}\left(Y_i-q_i\right) \mbox{ a.s. } \mathbb P\label{eq:pozero}
			\end{align}
			Since the potential outcomes are fixed, they cannot be functions of random variables. Hence equation \ref{eq:pozero} holds only if either
			\begin{enumerate}
				\item $Y_i(z_1,e_1) = c_1 p_i$, $Y_i(z_0,e_0) = c_2 q_i$ for all $i$ and $\sum_i (X_i - Y_i) = c\sum_i(p_i-q_i)$ almost surely for some constants $c_1, c_2$ and $c$ (or)
				\item $Y_i(z_1,e_1)$ and $Y_i(z_0,e_0)$ are all $0$.
				\end{enumerate}
				Ignoring the trivial solutions of equation \ref{eq:pozero}, the variance of the Horvitz-Thompson estimator is $0$ only if
				$X_0 - r_1 Y_0 = r_2$ almost surely for some constants $r_1$ and $r_2$. Now since $X_0 + Y_0 \leq n$, this implies $Cov(X_0, Y_0) \leq 0$. Hence 
				\begin{align*}
				0 &= Var(X_0 - Y_0) = Var(X_0) + Var(Y_0) - 2Cov(X_0,Y_0)
				\end{align*}
				which is true if and only if $Var(X_0) = 0$ and $Var(Y_0) = 0$. This implies that $X_0$ and $Y_0$ are constant, which contradicts the assumption that $\mathbb P$ is a non-constant design.
\end{proof}

\subsection{Proof of Proposition \ref{prop:VarSimpleLinearModelCRD}}
 \begin{proof}
 	Recall from the previous lemma, that 
 	\begin{align*}
 	\hat{\beta}_{naive} &= \beta + \gamma \left( \frac{\sum_i\sum_j g_{ij}z_j z_i}{n_t} - \frac{\sum_i\sum_j g_{ij}z_j (1-z_i)}{n_c}\right) + \frac{1}{n_t}\sum_i \epsilon_i z_i - \frac{1}{n_c}\sum_i \epsilon_i (1-z_i)\\
 	&= \beta + \gamma \left( \frac{n}{n_cn_t}\sum_i\sum_j g_{ij}z_j z_i - \frac{1}{n_c}\sum_i\sum_j g_{ij}z_i \right) + \sum_i \epsilon_i \frac{nz_i - n_t}{n_t n_c}	
 	\end{align*}
 	
 	%Under the assumption of $n_c = n_t$, we get
 	%\begin{align*}
 	%\hat{\tau} = \beta + 
 	%\frac{2\gamma}{n} \left(
 	%\sum_i\sum_j g_{ij}z_j(2z_i - 1)\right) + \frac{2}{n}\sum_i %\epsilon_i (2z_i -1) 
 	%\end{align*}
 	Let $t_1 = 	\gamma \left( \frac{n}{n_cn_t}\sum_i\sum_j g_{ij}z_j z_i - \frac{1}{n_c}\sum_i\sum_j g_{ij}z_i \right)$ and 
 	$t_2 = \sum_i \epsilon_i \frac{nz_i - n_t}{n_t n_c} $ and 
 	We will compute the variance of each of these terms separately. Note that the covariance between these two terms is $0$, since $e_i \perp z_i$ and $E[e_i]=0$, as seen below.
 	\begin{align*}
 	&Cov \left(\frac{n}{n_t n_c}z_iz_j -\frac{z_i}{n_c},\epsilon_i(nz_i-n_t)\right) \\
 	&= E\left[\left(\frac{n}{n_t n_c}z_iz_j -\frac{z_i}{n_c} \right)\epsilon_i(nz_i-n_t)\right] - E\left[\frac{n}{n_t n_c}z_iz_j -\frac{z_i}{n_c}\right] E[\epsilon_i(nz_i-n_t)] = 0
 	\end{align*}
 	The variance of $t_2$ is 
 	\begin{align*}
 	Var(t_2) &= Var\left(\sum_i \epsilon_i \frac{nz_i - n_t}{n_t n_c}\right)\\
 	&= \sum_i\sum_j Cov\left( \epsilon_i \frac{nz_i - n_t}{n_t n_c}, \epsilon_j \frac{nz_j - n_t}{n_t n_c}\right)\\
 	&= n E\left[e_i^2 \left(\frac{nz_i - n_t}{n_t n_c}\right)^2 \right]
 	&(\mbox{Since when $i \neq j$, the covariance is 0})\\
 	&=\frac{n\sigma^2}{n_t^2 n_c^2}\left[ (n-n_t)^2\frac{n_t}{n} + n_t^2 \frac{n_c}{n}\right]\\
 	&= \sigma^2\left(\frac{1}{n_t} + \frac{1}{n_c}\right)
 	\end{align*}
 	
 	To calculate the variance of $t_1$, note that
 	\begin{align*}
 	\frac{1}{\gamma^2}Var(t_1) &=  \frac{n^2}{n_t^2n_c^2}Var \left(\sum_i\sum_j g_{ij}z_jz_i \right) + \frac{1}{n_c^2}Var \left(\sum_i\sum_j g_{ij}z_i \right)  \\
 	&- 2\frac{n}{n_c^2n_t} Cov\left(\sum_i\sum_j g_{ij}z_jz_i,\sum_i\sum_j g_{ij}z_i \right)
 	\end{align*}
 	
 	Each of these variances are calculated in the propositions below. Using these propositions, we get
 	
 	\begin{align*}
 	&\frac{1}{\gamma^2}Var(t_1) \\
 	&=  \frac{4n}{n_t n_c }\frac{(n_t-1)}{(n-1)(n-2)(n-3)}\left( m(n_c-1) + 2m^2\frac{(3n + 3n_t -2n n_t -3)}{n(n-1)} + (n_t-2)\sum_i d_i^2\right)\\
 	& + \frac{1}{n_c}\frac{n_t}{n^2} \left(\sum_i{d_i^2} - \frac{\sum_{i\neq j}{d_id_j}}{n-1}\right)  \\
 	&- 2\frac{1}{n_c}\frac{2 (n_t-1)}{(n-1)(n-2)} \left(\sum_i d_i^2 - \frac{4m^2}{n}\right)
 	\end{align*}

 	%bilinear form using the following notation. Let $x$ be a vector of length $n$ with $i^{th}$ entry equal to $2z_i-1$ and $y$ be a vector of length $n$ with $i^{th}$ entry equal to $z_i$, let $G$ be the fixed adjacency matrix of the network, then we have
 	%$$Var(t_1) = \frac{4\gamma^2}{n^2}Var(x'Gy)$$
 	%Moreover, $E[x] = 0$, $E[y] = \{1/2\}_{i=1}^n$. Let $\Sigma = Cov(x,y)$, then $Var(x) = 2\Sigma$, $Var(y) = \Sigma/2$ with entries $(\Sigma)_{ii} = \frac{1}{2}$ and $(\Sigma)_{ij} = -\frac{1}{2(n-1)}$
 	
 	%By the formula of variance of a bilinear form, we have
 	%\begin{align*}
 	%Var(x'Gy) &= tr(G Cov(x,y)'G Cov(x,y)') +  tr(G Var(y) G' Var(x)) + \\
 	%& E[x]'GVar(y) G' E[x] + E[y]'G'Var(x) G E[y] + 2 E[x]'GCov(x,y)'G E[y]
 	%\end{align*}
 	
 	%Let $d$ be the vector of degree sequence of $G$. Fortunately, for us $E[x] = 0$, hence we get
 	%\begin{align*}
 	%Var(x'Gy) &= tr(G \Sigma' G \Sigma') +  tr(G (\Sigma/2) G' (2\Sigma)) +  E[y]'G'2 \Sigma G E[y]\\
 	%& = 2tr(G \Sigma' G \Sigma') + 2E[y]'G' \Sigma G E[y]\\
 	%& = 2tr(G \Sigma' G \Sigma') + d' \Sigma (d/2)
 	%\end{align*}
 	%Note  that $$d' \Sigma (d/2) = \sum_i{\frac{d_i^2}{2}} -\frac{1}{4(n-1)} \sum_{i\neq j}d_i d_j$$ and
 	%$$ 2tr(G \Sigma' G \Sigma') = \frac{1}{4(n-1)^2}\left( (1-2n) \sum_i d_i^2 + 2n^2m + \sum_{i \neq j} d_i d_j \right)$$
 \end{proof}
 
 \begin{proposition}
 	\begin{align*}
 	\frac{1}{4} Var\left( \sum_i\sum_j g_{ij}z_jz_i \right) = &  \frac{n_tn_c (n_t-1)}{n(n-1)(n-2)(n-3)}\left( m(n_c-1) + 2m^2\frac{(3n + 3n_t -2n n_t -3)}{n(n-1)} + (n_t-2)\sum_i d_i^2\right)
 	\end{align*}
 	
 \end{proposition}
 \begin{proof}
 	\begin{align*}
 	Var\left( \sum_i\sum_j g_{ij}z_jz_i \right) &= \sum_i\sum_j \sum_k \sum_l g_{ij} g_{kl} Cov(z_iz_j, z_k z_l)
 	\end{align*}
 	
 	We will consider several cases. To keep notation simple, let $(x)_n = (x)(x-1)(x-2)\ldots(x-(n-1))$ be the falling factorial. 
 	\paragraph{Case 1} $k=i, l=j$ and $k=j, l=i$
 	\begin{align*}
 	2\sum_{ij} g_{ij}^2 Var(z_iz_j) = \left(\frac{n_t-1}{n-1}\frac{n_t}{n} - \left(\frac{n_t}{n}\right)^2\right) \sum_{ij}g_{ij}  = 4m \left(\frac{n_t-1}{n-1}\frac{n_t}{n} - \left(\frac{(n_t)_2}{(n)_2}\right)^2\right)
 	\end{align*}
 	\paragraph{Case 2} $k \neq (i,j), l \neq (i,j), k \neq l$
 	\begin{align*}
 	Cov(z_iz_j, z_kz_l) = E[z_iz_jz_kz_l] - E[z_iz_j]E[z_kz_l]=\frac{(n_t)_4}{(n)_4} - \left(\frac{(n_t)_2}{(n)_2}\right)^2 \forall k\neq i, l \neq j, k \neq l
 	\end{align*}
 	
 	\begin{align*}
 	\sum_{i}\sum_{j}\sum_{k\neq i,j}\sum_{l\neq i,j}g_{ij}g_{kl} = \sum_{i}\sum_{j}g_{ij}(2m - 2d_i - 2d_j + 2g_{ij}) = 4(m^2 + m - \sum_i{d_i^2} ) 
 	\end{align*}
 	\paragraph{Case 3} $k=i, l \neq j$, $k=j, l \neq i$, $l=i, k \neq j$ and $l=j, k \neq i$
 	Let $\mathcal{K}_{ij}$ be the index set of tuples $(k,l)$ satisfying the conditions mentioned above, for a fixed $(i,j)$. Then for any $(k,l) \in \mathcal{K}_{ij}$, we have,
 	\begin{align*}
 	Cov(z_iz_j, z_kz_l) = E[z_i^2z_j]-E[z_iz_j]E[z_iz_l] = \frac{(n_t)_3}{(n)_3} - \left(\frac{(n_t)_2}{(n)_2}\right)^2
 	\end{align*}
 	\begin{align*}
 	\sum_{i}\sum_{j}\sum_{k \in \mathcal{K}{ij}} g_{ij}g_{kl} = \sum_{i}\sum_{j}g_{ij}(2d_i + 2d_j -4g_{ij}) = 4\left(\sum_i{d_i^2} - 2m\right) 
 	\end{align*}
 	Combining these three cases, we get
 	\begin{align*}
 	\frac{1}{4}Var\left( \sum_i\sum_j g_{ij}z_jz_i \right) = 
 	&  m \left( \frac{(n_t)_2}{(n)_2} - \left(\frac{(n_t)_2}{(n)_2}\right)^2\right) + \\
 	& \left(m^2 + m - \sum_i{d_i^2} \right) \left( \frac{(n_t)_4}{(n)_4} - \left(\frac{(n_t)_2}{(n)_2}\right)^2 \right) + \\
 	& \left(\sum_i{d_i^2} - 2m\right)  \left(\frac{(n_t)_3}{(n)_3} - \left(\frac{(n_t)_2}{(n)_2}\right)^2\right)
 	\end{align*}
 	
 	Collecting each term and simplifying, we get
 	\begin{align*}
 	& \frac{1}{4}Var\left( \sum_i\sum_j g_{ij}z_jz_i \right) \\
 	&=  \frac{n_t n_c (n_t-1)(n_c-1)}{n(n-1)(n-2)(n-3)} m  + 
 	\frac{2n_t n_c (n_t-1)}{n^2 (n-1)^2 (n-2)(n-3)}(3n + 3n_t -2n n_t -3) m^2 \\
 	& + \frac{n_t n_c (n_t-1) (n_t-2)}{n(n-1)(n-2)(n-3)} \sum_{i}{d_i^2}\\
 	&= \frac{n_tn_c (n_t-1)}{n(n-1)(n-2)(n-3)}\left( m(n_c-1) + 2m^2\frac{(3n + 3n_t -2n n_t -3)}{n(n-1)} + (n_t-2)\sum_i d_i^2\right)
 	\end{align*}
 	
 \end{proof}
 
 \begin{proposition}
 	$$Var\left(\sum_i\sum_j g_{ij} z_j\right) = \frac{n_tn_c}{n^2} \left(\sum_i{d_i^2} - \frac{\sum_{i\neq j}{d_id_j}}{n-1}\right)
 	$$
 	%\left(\frac{n_t}{n} - \left(\frac{n_t}{n}\right)^2 \right)\sum_i{d_i^2} + \left(\frac{(n_t)_2}{(n)_2} - \left(\frac{n_t}{n}\right)^2 \right) \sum_{i\neq j} d_i d_j$$
 \end{proposition}
 \begin{proof}
 	\begin{align*}
 	Var\left(\sum_i \sum_j g_{ij} z_i \right) &= Var\left( \sum_i d_i z_i\right) = Var(z_1)\sum_i{d_i^2} + Cov(z_1,z_2)\sum_{i\neq j} d_i d_j\\
 	&= \left(\frac{n_t}{n} - \left(\frac{n_t}{n}\right)^2 \right)\sum_i{d_i^2} + \left(\frac{(n_t)_2}{(n)_2} - \left(\frac{n_t}{n}\right)^2 \right) \sum_{i\neq j} d_i d_j
 	\end{align*}
 \end{proof}
 
 \begin{proposition}
 	$$Cov\left(\sum_i\sum_j g_{ij} z_jz_i,\sum_i\sum_j g_{ij} z_j \right) = \frac{2n_t n_c (n_t-1)}{n(n-1)(n-2)} \left(\sum_i d_i^2 - \frac{4m^2}{n}\right)$$
 \end{proposition}
 \begin{proof}
 	\begin{align*}
 	Cov\left(\sum_i\sum_j g_{ij} z_jz_i,\sum_i\sum_j g_{ij} z_j\right)  &= \sum_i\sum_j \sum_k \sum_l g_{ij}g_{kl} Cov(z_iz_j,z_k) = \sum_i\sum_j \sum_k  g_{ij}d_{k} Cov(z_iz_j,z_k)  
 	\end{align*}
 	Consider again, 3 cases:
 	\paragraph{Case 1} $k=i$ and $k=j$
 	\begin{align*}
 	\sum_i\sum_j \sum_k  g_{ij}d_{k} Cov(z_iz_j,z_k) &= \sum_{i}\sum_{j} g_{ij}(d_i Cov(z_iz_j,z_i) + d_jCov(z_iz_j,z_j))\\
 	&= \left(\frac{(n_t)_2}{(n)_2} - \frac{(n_t)_2}{(n)_2}\frac{n_t}{n}\right)  2\sum_i d_i^2
 	\end{align*}
 	\paragraph{Case 2} $k\neq (i,j)$
 	\begin{align*}
 	\sum_i\sum_j \sum_k  g_{ij}d_{k} Cov(z_iz_j,z_k) 
 	&= \left(\frac{(n_t)_3}{(n)_3} - \frac{(n_t)_2}{(n)_2}\frac{n_t}{n}\right)  \sum_i \sum_j g_{ij}\sum_{k \neq (i,j)} d_k\\
 	&= \left(\frac{(n_t)_3}{(n)_3} - \frac{(n_t)_2}{(n)_2}\frac{n_t}{n}\right)  \sum_i \sum_j g_{ij}(2m - d_i - d_j)\\
 	&= \left(\frac{(n_t)_3}{(n)_3} - \frac{(n_t)_2}{(n)_2}\frac{n_t}{n}\right)  (4m^2-2\sum_i d_i^2)
 	\end{align*}
 	
 	Adding these two terms gives simplifying gives the desired result.
 	\begin{align*}
 	Cov\left(\sum_i\sum_j g_{ij} z_jz_i,\sum_i\sum_j g_{ij} z_j \right) &= 
 	\left( \frac{(n_t)_3}{(n)_3} -  \frac{(n_t)_2}{(n)_2} \frac{n_t}{n} \right)4m^2  + \left( \frac{(n_t)_2}{(n)_2} - \frac{(n_t)_3}{(n)_3}\right) 2\sum_i{d_i^2}\\
 	& = \frac{2n_t n_c (n_t-1)}{n(n-1)(n-2)} \left(\sum_i d_i^2 - \frac{4m^2}{n}\right)\\
 	& =\frac{-2n_t(n_t-1) n_c }{n^2(n-1)(n-2)}4m^2 + 2\frac{n_tn_c (n_t-1)}{n(n-1)(n-2)}\sum_{i}{d_i^2}\\
 	\end{align*}
 \end{proof}

\end{document}